\documentclass[final]{siamonline1116}
\pdfoutput=1

\usepackage{color}
\usepackage{graphicx}
\usepackage{dcolumn}
\usepackage{multirow}
\usepackage{amsmath,amssymb,amsopn}
\usepackage{bm}
\usepackage[tight-spacing=true]{siunitx} 

\usepackage{graphicx}
\usepackage{subcaption}

\let\chapter\section 
\usepackage[linesnumbered,ruled,vlined]{algorithm2e}

\usepackage[format=hang, font={footnotesize,sf}, labelfont={bf}, margin=1cm, aboveskip=5pt,position=bottom]{caption}

\usepackage[T1]{fontenc}
\usepackage[ansinew]{inputenc}
\usepackage{lmodern}

\usepackage[letterpaper,centering]{geometry}
\usepackage{longtable}
\usepackage{color,colortbl}

\usepackage{url}

\usepackage[textsize=footnotesize]{todonotes}
\usepackage{arydshln}

\usepackage{tikz}
\usetikzlibrary{calc}
\usepackage{pgfplots}
\pgfplotsset{compat=1.6}
\usepgfplotslibrary{groupplots}

\pgfplotsset{colormap={myCmap}{rgb255=(247,251,255) rgb255=(222,235,247) rgb255=(198,219,239) rgb255=(158,202,225) rgb255=(107,174,214) rgb255=(66,146,198) rgb255=(33,113,181) rgb255=(8,81,156) rgb255=(8,48,107)}}

\definecolor{nutsColor}{rgb}{0.9941,    0.7754,    0.6172}%
\definecolor{tmdrColor}{rgb}{0.5645,    0.7207,    0.8535}%
\definecolor{dramColor}{rgb}{0.84766,0.27734,0.00391}%

\usepackage{dcolumn}

\newcommand{\real}{\mathbb{R}}

\def\argmin{\mathop{\mathrm{argmin}}}

\newcommand{\Varsimp}{\mathbb{V}\mathrm{ar}}
\newcommand{\eqd}{\buildrel d \over =}

\definecolor{nutsColor}{rgb}{0.9941,    0.7754,    0.6172}%
\definecolor{tmdrColor}{rgb}{0.5645,    0.7207,    0.8535}%
\definecolor{dramColor}{rgb}{0.84766,0.27734,0.00391}%

\newcommand{\trv}{\theta}
\newcommand{\rrv}{r}
\newcommand{\pd}{n}

\newcommand{\td}{\pi} 
\newcommand{\rd}{p}

\newcommand{\tmeas}{\mu_\trv}
\newcommand{\rmeas}{\mu_\rrv}

\newcommand{\emap}{T} 
\newcommand{\demap}{{\nabla T}}

\newcommand{\imap}{{\widetilde{T}}} 
\newcommand{\dimap}{{\nabla \widetilde{T} }}

\newcommand{\bspace}{\mathcal{B}(\real^D)} 

\newcommand{\tk}{P} 

\newcommand{\mapp}{\gamma} 
\newcommand{\mappa}{\bar{\mapp}}

\numberwithin{theorem}{section}

\newcommand{\TheTitle}{Transport map accelerated Markov chain Monte Carlo} 
\newcommand{\TheAuthors}{M.\ D.\ Parno and Y.\ M.\ Marzouk }

\headers{\TheTitle}{\TheAuthors}

\title{{\TheTitle}\thanks{Submitted to the editors 9 June 2017.
\funding{This work was supported by the US Department of Energy, Office of Advanced Scientific Computing Research (ASCR), under grant number DE-SC0009297, as part of the DiaMonD Multifaceted Mathematics Integrated Capability Center.}}}

\author{
  Matthew D.\ Parno\thanks{US Army Engineer Research and Development Center (ERDC), Hanover, NH
    (\email{Matthew.D.Parno@usace.army.mil}).}
  \and
  Youssef M.\ Marzouk\thanks{Massachusetts Institute of Technology, Cambridge, MA (\email{ymarz@mit.edu}).}
}

\usepackage{amsopn}

\ifpdf
\hypersetup{
  pdftitle={\TheTitle},
  pdfauthor={\TheAuthors}
}
\fi

\begin{document}
\maketitle

\begin{abstract}
We introduce a new framework for efficient sampling from complex probability distributions, using a combination of transport maps and the Metropolis-Hastings rule. The core idea is to use deterministic couplings to transform typical Metropolis proposal mechanisms (e.g., random walks, Langevin methods) into non-Gaussian proposal distributions that can more effectively explore the target density. Our approach adaptively constructs a lower triangular transport map---an approximation of the Knothe-Rosenblatt rearrangement---using information from previous MCMC states, via the solution of an optimization problem. This optimization problem is convex regardless of the form of the target distribution. It is solved efficiently using a Newton method that requires no gradient information from the target probability distribution; the target distribution is instead represented via samples. Sequential updates enable efficient and parallelizable adaptation of the map even for large numbers of samples. We show that this approach uses inexact or truncated maps to produce an adaptive MCMC algorithm that is ergodic for the exact target distribution. Numerical demonstrations on a range of parameter inference problems show order-of-magnitude speedups over standard MCMC techniques, measured by the number of effectively independent samples produced per target density evaluation and per unit of wallclock time.

\end{abstract}



\section{Introduction}
\label{sec:intro}

Markov chain Monte Carlo (MCMC) algorithms provide an enormously flexible approach for sampling from complex target probability distributions, using only evaluations of an unnormalized probability density \cite{Gelman2003, RobertBook2004, Liu2004, Brooks2011}. Within this general framework, the Metropolis-Hastings algorithm \cite{Metropolis1953,Hastings1970} is one of the most broadly applicable and well studied sampling strategies.
%
It combines a simple proposal distribution with an accept/reject step to create the transition kernel for a Markov chain that has the desired target as its stationary distribution. Under some additional technical conditions on the proposal $q_\trv$ and target density $\td$, the Markov chain is also ergodic \cite{Roberts2004}.

This paper introduces a new approach to the design of Metropolis-Hastings algorithms, based on the adaptive construction of \textit{transport maps} between the target probability distribution and a simple reference distribution. These maps are monotone and typically nonlinear transformations of the target distribution that render it easier to sample, much like a preconditioner expedites the solution of a linear system. To put our approach in context, we first recall some challenges underlying MCMC sampling and current methods for addressing them.



Effective MCMC proposal mechanisms seek to make successive iterates of the Markov chain as independent as possible. When estimating an expectation over the target distribution, efficient ``mixing'' in this sense reduces the variance of estimates computed from the MCMC samples. A useful intuition is that effective MCMC proposals aim to approximate the target distribution at least locally (e.g., in the case of random-walk Metropolis or Langevin proposals) or perhaps globally (e.g., in the case of Metropolis independence samplers). Consider, for example, a Gaussian proposal density centered at the current state of the chain, as in a random-walk Metropolis algorithm. The adaptive Metropolis scheme of \cite{Haario2001} sequentially updates the covariance of this proposal in order to reflect the covariance of $\td$. In a similar fashion, \cite{Atchade2006} uses the empirical covariance of the target to scale proposals in a Metropolis-adjusted Langevin algorithm (MALA), which also uses the gradient of $\td$ to push the proposal mean towards regions of higher target density.

Many other MCMC algorithms use local derivative information to improve sampling 
of the target distribution. Hamiltonian Monte Carlo methods (HMC), as in 
\cite{Neal2011} and \cite{Hoffman2011}, propose samples via trajectories of a 
Hamiltonian dynamical system defined on an augmented state space. Computing 
these trajectories requires many evaluations of the gradient of the target 
density, but can produce large s.pdf that have high acceptance probability. The 
stochastic Newton method of \cite{Martin2012} uses higher-order derivative 
information, in the form of approximate Hessians of the local 
log-posterior, to scale a Gaussian proposal in high dimensions. The 
geometrically-motivated approach of \cite{Girolami2011} also uses 
higher-order derivative information to define a local metric for both 
Langevin proposals and Hamiltonian dynamics on a Riemannian manifold.
Contrasting with these schemes but also related to our work are adaptive Metropolis independence samplers \cite{Andrieu2006}, which 
construct a global approximation of the target using, for example, Gaussian mixtures. This approximation is updated recursively from past MCMC samples using a stochastic approximation scheme.

The theory of optimal transport has a rich history that is somewhat separate 
from stochastic simulation and MCMC. The notion of an optimal transport map 
dates back to \cite{Monge1781}, who sought a deterministic coupling between 
(probability) measures that is optimal in the sense of minimizing an expected 
transport cost. This cost is defined by a function $c(\theta, r)$ 
that can be interpreted as the cost of transporting one unit of mass from 
$\theta$ to $r$. A relaxation of the Monge problem to more general couplings 
was introduced by Kantorovich \cite{Kantorovich1942, Vershik2013}; yet under 
certain conditions, a minimizer of the Kantorovich formulation also solves the 
Monge problem, i.e., is an optimal transport map. For a contemporary development 
of this subject, see \cite{Villani2009, Villani2003} and \cite{Rachev1998}. 
Optimal transport between discrete measures has been used for Bayesian inference in \cite{Reich2013}, where the solution of a discrete assignment problem yields a consistent ensemble transformation scheme to replace resampling, in the context of a Bayesian filter. This problem differs from those considered here, however, as we focus on transport between continuous probability measures. 
\cite{Moselhy2011} introduced the idea of continuous transport maps that characterize the Bayesian posterior distribution as a pushforward of the prior distribution. In this formulation, the transport map is used to generate independent samples from a distribution that in principle can be made arbitrarily close to $\td$. However, constructing sufficiently accurate maps can be computationally taxing.
The implicit sampling approach of \cite{chorin2009implicit,Chorin2010,Morzfeld2012} and the randomize-then-optimize approach of \cite{Bardsley2014} compute the action of certain transport maps sample-by-sample, without representing the maps explicitly. But these samples do not come from $\td$ and thus require reweighing in order to represent the target. 
Implementing either of these approaches requires access to gradients of $\td$. 

In this paper, we will use \textit{approximate} transport maps to achieve \textit{exact} sampling from the target distribution, by integrating 
transport maps with MCMC. We \textit{reverse} the direction of the maps computed in \cite{Moselhy2011}, and adaptively construct our maps (now from the target to a simple reference distribution)
by solving an optimization problem based on MCMC samples. We will show that the optimization problem has a remarkably simple structure: it is convex regardless of the form of the target distribution and separable across dimensions of the parameter space; it also affords substantial opportunities for parallel computation and efficient sequential updating.  Moreover, computing derivatives of the optimization objective requires no derivative information from the target probability density.  We will analyze the scheme from the theoretical perspective of adaptive MCMC, allowing us to establish ergodicity of the resulting chain.
The transport map constructed in this way aims to represent the entire target distribution as the pullback of a Gaussian reference measure, and in that sense our approach is a global one. Unlike adaptive Metropolis independence samplers, however, we approximate the target density not by choosing from a particular family of densities, but by building an invertible transformation between the target distribution and a reference distribution. Critically, this structure enables us to use both local proposals and global/independence proposals, and to transition naturally between the two as the transport map becomes more accurate. The transport map is not tied to any particular type of MCMC proposal; it instead provides a framework for improving many standard proposal schemes.

The remainder of this paper is organized as follows. Section~\ref{sec:maps} will provide relevant background on transport maps and explain how suitable maps can be constructed from samples. Section~\ref{sec:mapmcmc} will formulate the map-based MCMC approach, while Section~\ref{sec:adaptmcmc} will introduce adaptive strategies. A theoretical convergence analysis is provided Section~\ref{sec:convergence}. Section~\ref{sec:perf} compares the performance of map-based MCMC with that of existing state-of-the-art samplers on a range of test problems.

\section{Construction of transport maps}
\label{sec:maps}
Transport maps will be used in Sections \ref{sec:mapmcmc} and \ref{sec:adaptmcmc} to define a new class of MCMC methods.  This section first introduces transport maps in the context of optimal transportation (Section~\ref{sec:maps:overview}) and then describes a practical method for constructing maps from samples (Section~\ref{sec:maps:construct}).

\subsection{Optimal transportation}
\label{sec:maps:overview}

Consider two Borel probability measures on $\real^{\pd}$, $\tmeas$ and $\rmeas$. We will refer to these as the \textit{target} and \textit{reference} measures, respectively, and associate them with random variables $\trv \sim \tmeas$ and $\rrv \sim \rmeas$. A transport map $\emap: \real^{\pd} \rightarrow \real^{\pd}$ is a deterministic transformation that pushes forward $\tmeas$ to $\rmeas$, yielding 
\begin{equation}
\rmeas = \emap_{\sharp} \tmeas .
\label{eq:measconst}
\end{equation}
In other words, $\rmeas(A) = \tmeas \left ( \emap^{-1}(A) \right )$ for any Borel set $A \subseteq \real^{\pd}$. In terms of the random variables, we write $\rrv \eqd \emap (\trv)$, where $\eqd$ denotes equality in distribution. The transport map induces a \textit{deterministic coupling} of probability measures \cite{Villani2009}. 

Of course, there can be infinitely many transport maps between two probability 
measures. On the other hand, it is possible that no transport map exists: 
consider the case where $\tmeas$ has an atom but $\rmeas$ does not. If a 
transport map exists, one way of regularizing the problem and finding a unique 
map is to introduce a cost function $c(\trv,\rrv)$ on $\real^{\pd} \times 
\real^{\pd}$ that represents the work needed to move one unit of mass from 
$\trv$ to $\rrv$. Using this cost function, the total cost of pushing $\tmeas$ 
to $\rmeas$ is
\begin{equation}
  C(\emap) = \int_{\real^{\pd}} c\left(\trv,\emap(\trv)\right) \, d \tmeas(\trv) \label{eq:mongecost}.
\end{equation}
Minimization of this cost subject to the constraint $\rmeas = \emap_{\sharp} \tmeas$ is called the Monge problem, after \cite{Monge1781}.  A transport map satisfying the measure constraint \eqref{eq:measconst} and minimizing the cost in \eqref{eq:mongecost} is an \textit{optimal} transport map. The celebrated result of \cite{Brenier1991}, later generalized by \cite{McCann1995}, shows that this map exists, is unique, and is monotone $\tmeas$-a.e.\ when $\tmeas$ is atomless and the cost function $c(\trv,\rrv)$ is quadratic. Generalizations of this result to other cost functions and spaces have been established in \cite{Champion2011, Ambrosio2013, Feyel2004, Bernard2004}.

The choice of cost function in \eqref{eq:mongecost} naturally influences the structure of the map. For illustration, consider the Gaussian case of $\trv\sim N(0,I)$ and $\rrv\sim N(0,\Sigma)$ for some positive definite covariance matrix $\Sigma$.  The associated transport map is linear: $\emap = S \trv$, where the matrix $S$ is any square root of $\Sigma$. When the transport cost is quadratic, $c(\trv,\rrv) = | \trv - \rrv |^2$, $S$ is the symmetric square root obtained from the eigendecomposition of $\Sigma$, $\Sigma = V \Lambda V^\top$ and $S = V \Lambda^{1/2} V^\top$ \cite{Olkin1982}. If the cost is instead taken to be the following weighted quadratic
\begin{equation}
  c(\trv, \rrv) = \sum_{i=1}^{\pd} t^{i-1} | \trv_i- \rrv_i |^2, \ t > 0,
  \label{eq:rosencost}
\end{equation}
then, as $t \rightarrow 0$, the optimal map becomes lower triangular and equal to
 the Cholesky factor of $\Sigma$. Generalizing to non-Gaussian 
$\tmeas$ and $\rmeas$, as $t \rightarrow 0$,  the optimal maps $\emap_t$ obtained with the cost function 
\eqref{eq:rosencost} are shown by \cite{Carlier2010} and \cite{Bonnotte2013} to 
converge to the \textit{Knothe-Rosenblatt} (KR) rearrangement 
\cite{Rosenblatt1952, Knothe1957} between probability measures. The KR map exists and is 
uniquely defined if $\tmeas$ is absolutely continuous with respect to Lebesgue 
measure.  The KR map also has several useful properties: the Jacobian matrix 
of $\emap$ is lower triangular and has positive diagonal entries $\tmeas$-a.e. 
Because of this triangular structure, the Jacobian determinant and the inverse 
of the map are easy to evaluate. This is an important computational 
advantage that we exploit in Section \ref{sec:maps:construct}.  

We will employ lower triangular maps in our MCMC construction, but 
without directly appealing to the transport cost in \eqref{eq:rosencost}. While 
this cost is meaningful for theoretical analysis and even numerical 
continuation schemes \cite{Carlier2010}, we find that for small $t$, the sequence of weights 
$\{t^{i}\}$ quickly produces numerical underflow as the parameter 
dimension $\pd$ increases. Instead, we will directly impose the lower triangular 
structure and search for a map $\imap$ that \textit{approximately} satisfies the 
measure constraint, i.e., for which $\rmeas\approx\imap_{\sharp}\tmeas$. This 
approach is a key difference between our construction and standard optimal 
transportation. 

Numerical challenges with \eqref{eq:rosencost} are not the only reason to seek approximate maps.
Suppose that the target measure $\tmeas$ is a Bayesian posterior or some other intractable distribution, but let the reference $\rmeas$ be something simpler, e.g., a Gaussian distribution with identity covariance. In this case, the complex structure of $\tmeas$ is captured by the map $\emap$. Sampling and other tasks can then be performed with the simple reference distribution instead of the more complicated distribution. In particular, if a map exactly satisfying \eqref{eq:measconst} were available, sampling the target distribution $\tmeas$ would simply require drawing a sample $\rrv^\prime \sim \rmeas$ and pushing it to the target space with $\theta^\prime = \emap^{-1}(\rrv^\prime)$. This concept was employed by \cite{Moselhy2011} for posterior sampling.
 Depending on the structure of the reference and the target, however, finding an exact map may be computationally challenging. In particular, if the target contains many nonlinear dependencies that are not present in the reference distribution, the \textit{representation} of the map $\emap$ (e.g., in some canonical basis) can become quite complex.
 Hence, it is desirable to work with approximations to $\emap$. Below we will demonstrate that even approximate maps can capture the key structure of the target distribution and thus be used to construct more efficient MCMC proposals.

Another reason for seeking approximate transport maps is regularity. There is an 
extensive theory on the regularity of optimal transport---with much that is 
understood, along with some open questions \cite{Caffarelli1992}. Since we are only concerned with approximate measure transformations, we can impose regularity conditions that may not hold for the optimal map or the KR map. In particular, we will require that $\imap$ and its inverse have continuous derivatives on $\real^{\pd}$, i.e., that $\imap$ be a $C^1$-diffeomorphism. Later we will impose additional constraints on the derivatives of $\imap$, which will prove useful for our theoretical analysis of map-based MCMC.

\subsection{Constructing maps from samples}\label{sec:maps:construct}

As noted above, we will seek transport maps that have a lower triangular structure, i.e.,
\begin{equation}
\emap(\trv_1,\trv_2,\ldots ,\trv_\pd) = \left[\begin{array}{l}\emap_1(\trv_1)\\ \emap_2(\trv_1,\trv_2)\\ \vdots \\ \emap_\pd(\trv_1,\trv_2,\ldots ,\trv_\pd)\end{array}\right],\label{eq:lowtriform}
\end{equation}
where $\trv_i$ denotes the $i$th component of $\trv$ and $\emap_i:\real^i 
\rightarrow \real$ is $i$th component of the map $\emap$. For simplicity, we 
assume that both the target and reference measures are absolutely 
continuous on $\real^\pd$, with densities $\pi$ and $\rd$, respectively. This 
assumption precludes the existence of atoms in $\tmeas$ and thus makes the KR 
coupling well-defined. To find a useful approximation of the KR coupling, we 
will define a map-induced density $\tilde{\td}(\trv)$ and minimize the distance 
between this map-induced density and the target density $\td(\trv)$. The next 
three subsections describe the setup of this optimization problem.

Note that when the reference measure is a standard Gaussian (as we shall prescribe below), the construction of a map from target samples to the reference is a goal shared by the iterative Gaussianization scheme of \cite{Laparra2011} and the density estimation schemes of \cite{Tabak2010,Tabak2013}. Both of these approaches \textit{compose} a series of simple maps (e.g., sigmoid-type functions of one variable in \cite{Tabak2013}) in order to achieve the desired transformation, but can require a large number of such layers in order to converge. Also, the resulting maps are not triangular. Here, we seek to develop a more expressive all-at-once approximation of the triangular KR map.

\subsubsection{Optimization objective}

Let $\rd$ be the probability density associated with the reference measure $\rmeas$, and consider a transformation $\imap(\trv)$ that is monotone and differentiable $\tmeas$-almost everywhere. (In Section~\ref{sec:constraints} we will discuss constraints to ensure monotonicity; moreover, we will employ maps that are everywhere differentiable by construction.) Now consider the pullback of $\rmeas$ through $\imap$. The density of this pullback measure is
\begin{equation}
\tilde{\td}(\trv) = \rd ( \imap(\trv) ) | \det \dimap(\trv) | \label{eq:mapInducedTarget},
\end{equation}
where $\demap(\trv)$ is the Jacobian of the map, evaluated at $\trv$, and $| \det \dimap(\trv) |$ is the absolute value of the Jacobian determinant. We call $\tilde{\td}$ the \textit{map-induced density}.

If the measure constraint $\rmeas = \imap_{\sharp} \tmeas$ were exactly satisfied, the map-induced density $\tilde{\td}$ would equal the target density $\td$.  This suggests finding $\imap$ by minimizing a distance or divergence; to this end, we use the Kullback-Leibler (KL) divergence,
\begin{eqnarray}
D_{\text{KL}}(\td \| \tilde{\td}) & = & \mathbb{E}_{\td}\left[\log{\left(\frac{\td(\trv)}{\tilde{\td}(\trv)}\right)}\right] \label{eq:klexpr} \\
& = & \mathbb{E}_{\td} \bigg[\log{\td(\trv)} - \log{\rd\left(\imap(\trv)\right)} - \log{\left|\det \dimap(\trv) \right|}\bigg].\nonumber
\end{eqnarray}
We can then find transport maps by solving the following optimization problem:
\begin{equation}
\min_{T \in \mathcal{T}} \mathbb{E}_{\td} \bigg[ - \log{\rd\left(\emap(\trv)\right)} - \log{\left|\det \demap(\trv) \right|}\bigg],
\label{eq:opt}
\end{equation}
where $\mathcal{T}$ is some space of lower-triangular functions from $\real^\pd$ to $\real^\pd$. If $\mathcal{T}$ is large enough to include the KR map, then the solution of this optimization problem will exactly satisfy \eqref{eq:measconst}. Note that we have removed the $\log{\td(\trv)}$ term in \eqref{eq:klexpr} from the optimization objective \eqref{eq:opt}, as it is independent of $T$. 
If the exact coupling condition is satisfied, however, then the quantity inside the expectation in \eqref{eq:klexpr} becomes constant in $\trv$. If $\td$ is unnormalized, this constant is in fact the log of the normalizing constant of $\td$.

%

Note that the KL divergence is not symmetric.  We choose the direction above so that we can use Monte Carlo samples (in particular, MCMC samples) to approximate the expectation with respect to $\td$. Furthermore, as we will show below, this direction allows us to dramatically simplify the solution of \eqref{eq:opt} when $\rd$ is Gaussian.
Suppose that we have $K$ samples from $\td$, denoted by $\{\trv^{(1)},\trv^{(2)},\ldots ,\trv^{(K)}\}$. Taking a sample-average approximation (SAA) \cite{Kleywegt2002}, we replace the objective in \eqref{eq:opt} with its Monte Carlo estimate
and, for this fixed set of samples, solve the corresponding deterministic optimization problem:
\begin{equation}
\imap = \arg \min_{T \in\mathcal{T}} \frac{1}{K}\sum_{k=1}^{K} \Bigg[ -\log{\rd\left(\emap(\trv^{(k)})\right)} - \log{\left|\det \demap(\trv^{(k)}) \right|}\Bigg].
\label{eq:mckl}
\end{equation}
The solution $\imap$ is an approximation to the exact transport map for two reasons: first, we have used an approximation of the expectation operator; and second, we have restricted the feasible domain of the optimization problem to $\mathcal{T}$. The specification of $\mathcal{T}$ is the result of constraints, discussed in Section~\ref{sec:constraints}, and of the finite-dimensional parameterization of the map, discussed in Section~\ref{sec:maps:param}.

%








\subsubsection{Constraints}
\label{sec:constraints}
To write the map-induced density $\tilde{\td}$ as in 
\eqref{eq:mapInducedTarget}, it is sufficient that $\imap$ be differentiable and 
monotone, i.e., $( \trv^\prime - \trv) ^\top ( \imap(\trv^\prime) - 
\imap(\trv)  )\geq 0$ for distinct points $\trv, \trv^\prime \in \real^\pd$. 
Since we assume that $\tmeas$ has no atoms, to ensure that the pushforward 
$\imap_\sharp \tmeas$ also has no atoms we only need to require that $\imap$ be strictly monotone.
%
%
To show ergodicity of the MCMC samplers constructed in Sections~\ref{sec:mapmcmc} and \ref{sec:adaptmcmc}, however, we will need to impose the stricter condition that $\imap$ be bi-Lipschitz,
\begin{eqnarray}
 \lambda_{\text{min}}  \|\trv^\prime - \trv \| \leq \|\imap(\trv^\prime) - \imap(\trv)\| \leq  \lambda_{\text{max}} \|\trv^\prime - \trv \| \label{eq:dUB},
\end{eqnarray}
for some $0 < \lambda_{\text{min}} \leq \lambda_{\text{max}} < \infty$. This condition implies that $\imap$ is differentiable almost everywhere. But the maps we will employ are, by construction, everywhere differentiable and lower triangular, and hence the lower Lipschitz condition in \eqref{eq:dUB} is equivalent to a lower bound on the map derivative,
\begin{equation}
\frac{\partial \imap_i}{ \partial \trv_i} \geq \lambda_{\text{min}} \label{eq:derivLB}, \ i=1 \ldots \pd. 
\end{equation}
Since $\imap$ is lower triangular, the Jacobian $\dimap$ is also lower triangular, and \eqref{eq:derivLB} ensures that the Jacobian is positive definite. Because the Jacobian determinant is then positive, we can remove the absolute value from the determinant terms in \eqref{eq:opt}, \eqref{eq:mckl}, and related expressions. This is an important step towards arriving at a convex optimization problem (see Section~\ref{subsubsec:convex}).  We stress that while a nonzero $\lambda_{\text{min}}$ is required for our theoretical analysis, it does not need to be tuned in order to apply the algorithm in practice; typically we just choose a very small value, e.g.,  $\lambda_{\text{min}} = 10^{-8}$.  An explicit value for $\lambda_{\text{max}}$ can also be prescribed, but can instead be defined implicitly through the construction described next.

Many representations of $\imap$ (e.g., polynomial expansions) will yield maps with unbounded derivatives as $\|\trv\|\rightarrow \infty$.  Clearly, such maps would not satisfy the upper bound in \eqref{eq:dUB}.  Fortunately, a simple correction ensures \eqref{eq:dUB} is satisfied. 
Let $\imap: \real^\pd \rightarrow \real^\pd$ be a continuously differentiable function whose derivatives grow without bound as $\|\trv\|\rightarrow \infty$, but are finite within a ball $B(0,R)$ of radius $R<\infty$.  We can satisfy \eqref{eq:dUB} by setting $\imap^R(\trv) = \imap(\trv)$ over $B(0,R)$ and forcing $\imap^R(\trv)$ to be linear outside of this ball.  More precisely, let $w(\trv) := R\frac{\trv}{\|\trv\|}$ be the projection of $\trv$ to the closest point in $B(0,R)$ and let $d(\trv): = \frac{\trv}{\|\trv\|}\cdot \nabla \imap(w(\trv))$ be the directional derivative of $\imap$ at the ball boundary. We then define $\imap^R(\trv)$ in terms of $\imap(\trv)$ as
\begin{equation}
\imap^R(\trv) = \left\{ \begin{array}{ll} \imap(\trv) & \quad \|\trv\|\leq R\\ \imap\left(w(\trv)\right) + d(\trv)(\trv-w(\trv)) &  \quad \|\trv\|> R \end{array}\right. .\label{eq:ubMap}
\end{equation}
Note that a continuously differentiable $\imap(\trv)$ will yield a continuously differentiable $\imap^R(\trv)$. Moreover, if $\imap(\trv)$ satisfies the lower bound in \eqref{eq:dUB}, $\imap^R(\trv)$ will satisfy both the lower and upper bounds in \eqref{eq:dUB}. 


When a finite number of samples are used in the Monte Carlo sum of \eqref{eq:mckl}, $R$ can usually be chosen so that all the samples lie in $B(0,R)$, and hence $\imap$ can be evaluated directly.  In this setting, a value of $R$ need not be explicitly prescribed. However, our asymptotic convergence theory requires finite derivatives of the map as $\|\trv\|\rightarrow \infty$ in order to achieve the correct tail behavior, which is guaranteed by using $\imap^R$ as in \eqref{eq:ubMap}.

Unfortunately, we cannot generally enforce the lower bound in \eqref{eq:derivLB} over the entire support of the target measure.  A weaker, but practically enforceable, alternative is to require the map to be increasing at each sample used to approximate the KL divergence. In other words, we use the constraints
\begin{equation}
\left.\frac{\partial \imap_i}{\partial \trv_i}\right|_{\trv^{(k)}} \geq \lambda_{\text{min}} \quad \forall  i \in\{1,2,\ldots, \pd \}, \ \forall k \in\{1,2,\ldots, K\}\label{eq:mono2}.
\end{equation}
Practically, we have found that \eqref{eq:mono2} is sufficient to ensure the monotonicity of a map represented by a finite basis expansion.

\subsubsection{Convexity and separability of the optimization problem}
\label{subsubsec:convex}

Now we consider the task of minimizing the objective in \eqref{eq:mckl}.  The $1/K$ term can immediately be discarded, and the derivative constraints above let us remove the absolute value from the determinant term. 
While one could tackle the resulting minimization problem directly, we can simplify it further by exploiting the structure of the reference density and the triangular map.

First, we let $\rrv\sim N(0,I)$.  This choice of reference distribution yields
\begin{equation}
\log\rd(\rrv) = -\frac{\pd}{2}\log(2\pi) - \frac12\sum_{i=1}^\pd\rrv_i^2. \label{eq:gaussPart}
\end{equation}
Next, the lower triangular Jacobian $\dimap$ simplifies the determinant term in \eqref{eq:mckl} to give
\begin{equation}
\log{\left|\det \dimap(\trv) \right|} = 
\log{ ( \det\dimap(\trv) )} = \log\left(\prod_{i=1}^{\pd} \frac{\partial \imap_i}{\partial \trv_i}\right) = \sum_{i=1}^{\pd}\log\frac{\partial \imap_i}{\partial \trv_i} \label{eq:detexp}.
\end{equation}
The objective function in \eqref{eq:mckl} now becomes
\begin{equation}
C(\imap) = \sum_{i=1}^{\pd} \sum_{k=1}^{K}\left[  \frac12 \imap_i^2(\trv^{(k)}) - \log\left.\frac{\partial \imap_i}{\partial \trv_i}\right|_{\trv^{(k)}}\right]\label{eq:fullopt}.
\end{equation}
This objective is \textit{separable}: it is a sum of $\pd$ terms, each involving a single component $\imap_i$ of the map. The constraints in \eqref{eq:mono2} are also separable; there are $K$ constraints for each $\imap_i$, and no constraint involves multiple components of the map. Hence the entire optimization problem separates into $\pd$ individual optimization problems, one for each dimension of the parameter space. Moreover, each optimization problem is \textit{convex}: the objective is convex and the feasible domain is closed (note the $\geq$ operator in the linear constraints \eqref{eq:mono2}) and convex.



In practice, we must solve the optimization problem over some finite-dimensional space of candidate maps. Let each component of the map be written as $\imap_i(\trv;\mapp_i)$, $i=1\ldots \pd$, where $\mapp_i \in \real^{M_i}$ is a vector of parameters, e.g., coordinates in some basis. 
The complete map is then defined by the parameters $\mappa = [\mapp_1, \mapp_2,\ldots , \mapp_{\pd}]$. Note that there are distinct parameter vectors for each component of the map. The optimization problem over the parameters remains separable, with each of the $\pd$ different subproblems given by:
\begin{equation}
\begin{aligned}
& \underset{\mapp_i}{\min}
& & \sum_{k=1}^{K}\left[  \frac12 \imap^2_i(\trv^{(k)}; \mapp_i) - \log\left.\frac{\partial \imap_i(\trv;\mapp_i)}{\partial \trv_i}\right|_{\trv^{(k)}}\right]\\
& \text{s.t.}
& & \left.\frac{\partial \imap_i(\trv;\mapp_i)}{\partial \trv_i}\right|_{\trv^{(k)}} \geq \lambda_{\text{min}}, \quad k\in\{1,2,\ldots ,K\},
\end{aligned}\label{eq:basicopt}
\end{equation}
for $i=1\ldots \pd$. All of these optimization subproblems can be solved in parallel without evaluating the target density $\td(\trv)$. Since the map components $\imap_i$ are linear in the coefficients $\mapp_i$, each finite-dimensional problem is still convex.  

\subsection{Map parameterization}
\label{sec:maps:param}
In this work, we parameterize each component of the map $\imap_i$ with a multivariate polynomial expansion.  Each multivariate polynomial $\psi_{\mathbf{j}}$ is defined as
\begin{equation}
\psi_{\mathbf{j}}(\trv) = \prod_{i=1}^{\pd} \varphi_{j_i}(\trv_i).
\end{equation}
where $\mathbf{j} = (j_1,j_2, \ldots , j_{\pd} ) \in \mathbb{N}_0^{\pd}$ is a multi-index and $\varphi_{j_i}$ is a univariate polynomial of degree $j_i$. The univariate polynomials can be chosen from any family of orthogonal polynomials (e.g., Hermite, Legendre, Jacobi); even monomials are sufficient for the present purposes.\footnote{In principle, there is some advantage to choosing polynomials that are orthogonal with respect to the input distribution $\tmeas$, as in polynomial chaos approaches \cite{GhanemBook,LeMaitre2010}.
In the present context, however, we only have samples from $\tmeas$, and this distribution is almost certainly not one of the canonical distributions found in the Wiener-Askey scheme \cite{Xiu2002}. Thus $\tmeas$-orthogonal polynomials are not readily available, and there is little reason to be picky about the choice of polynomial basis.}
Using these multivariate polynomials, we express the map as a finite expansion of the form
\begin{equation}
\imap_i(\trv; \mapp_i) = \sum_{\mathbf{j}\in\mathcal{J}_i} \mapp_{i,\mathbf{j}} \, \psi_{\mathbf{j}}(\trv) \label{eq:polyexpand},
\end{equation}
where $\mathcal{J}_i$ is a set of multi-indices defining the polynomial terms in the expansion.  Notice that the cardinality of the multi-index set defines the dimension of each parameter vector $\mapp_i$, i.e., $M_i = |\mathcal{J}_i|$.  An appropriate choice of each multi-index set $\mathcal{J}_i$ will force the entire map $\imap$ to be lower triangular.  

One simple choice of the multi-index set corresponds to a total-order polynomial basis, where the maximum degree of each multivariate polynomial is bounded by some integer $p \geq 0$:
\[
\mathcal{J}^{TO}_{i} = \{ \mathbf{j} : \| \mathbf{j}\|_1 \leq p, \ j_k=0 \ \forall k > i \}.
\]
The first constraint in this set limits the polynomial order, while the second 
constraint, $j_k=0 \ \forall k > i$, applied over all $i=1\ldots \pd$ components 
of the map, forces $\imap$ to be lower triangular.  
A  smaller multi-index set for large $\pd$ can be obtained by removing all the 
mixed terms in the basis:
\[
  \mathcal{J}^{NM}_{i} = \{ \mathbf{j} : \|\mathbf{j}\|_1 \leq p, \  j_k j_m=0 \ \forall k \neq m,  \  j_k=0 \ \forall  k > i \}.
\] 
An even more parsimonious option is to use diagonal maps, via the multi-index sets
\[
  \mathcal{J}^{D}_{i} = \{ \mathbf{j} : \|\mathbf{j}\|_1 \leq p, \ j_k=0 \  \forall k \neq  i\}.
\] 
We will occasionally use a union of low degree $\mathcal{J}^{TO}_{i}$ and high degree $\mathcal{J}^{D}_{i}$ to define expressive map expansions with a tractable number of terms.

Finally, we emphasize that \textit{any} parameterization of the map that is linear in the coefficients $\mappa$ can be used in the optimization problems defined earlier. While the examples in this paper will focus on polynomial maps, we have also had good success representing the map as a summation of linear terms and radial basis functions \cite{Parno2014thesis}.

\subsection{Solving the map optimization problem}\label{sec:maps:efficient}
Since the map $\imap_i(\trv; \mapp_i)$ is linear in the expansion coefficients $\mapp_i$, the objective in \eqref{eq:fullopt} can be evaluated using efficient matrix-matrix and matrix-vector operations.  We first construct matrices $F_i, G_i\in \real^{K\times M_i}$ with components defined by $[F_{i}]_{k, \mathbf{j}} = \psi_{\mathbf{j}}(\trv^{(k)})$ and $[G_{i}]_{k, \mathbf{j}}  = \left . \frac{\partial \psi_{\mathbf{j}}}{\partial \trv_i}  \right |_{\trv^{(k)}}$ for all $\mathbf{j}\in \mathcal{J}_i$.  Recall that $K$ is the number of samples in our Monte Carlo approximation of the optimization objective.  Using these matrices and the expansion \eqref{eq:polyexpand}, we can rewrite \eqref{eq:fullopt} as
\begin{equation}
\begin{aligned}
& \underset{\mapp_i}{\min}
& & \frac{1}{2} \mapp_i^\top(F_i^\top F_i)\mapp_i - c^\top\log(G_i\mapp_i)\\
& \text{s.t.}
& & G_i\mapp_i \geq \lambda_{\text{min}}, 
\end{aligned}\label{eq:basicopt2}
\end{equation}
where $c$ is a $K$-dimensional vector of ones and the $\log$ is taken componentwise.  Clearly, the objective can be evaluated with efficient numerical linear algebra routines.
 
%
%
%
%
%

Beyond efficient evaluations, the only difference between \eqref{eq:basicopt2} and a simple quadratic program is the log term in the objective.  However, the quadratic term often dominates the log term, making a standard Newton optimizer with backtracking line search quite efficient.  In practice, starting with an identity map, we usually observe convergence in fewer than ten Newton iterations.
Notice also that the log term in \eqref{eq:basicopt2} acts as a barrier function for the constraints.

\section{Map-based MCMC proposals}
\label{sec:mapmcmc}
%

Now we will show how a transport map can be used to modify the Metropolis-Hastings algorithm by equivalently transforming either the target distribution or the proposal mechanism. In this section, we assume that a {fixed} transport map $\imap$ is in hand. Of course, this map must somehow be constructed, and hence the fixed-map approach described here is just an intermediate step in our exposition. The next section (Section~\ref{sec:adaptmcmc}) will use the optimization approaches of Section~\ref{sec:maps} to iteratively build such a map in an adaptive MCMC framework.

A simple Metropolis-Hastings algorithm \cite{Hastings1970,Metropolis1953} generates a new state $\trv^{(k+1)}$ from the current state $\trv^{(k)}$ in two s.pdf.  First, a sample $\trv'$ is drawn from a proposal density $q_{\trv,\bar{\mapp}}( \cdot |\trv^{(k)})$.  Then, an accept-reject step is performed: $\trv^{(k+1)}$ is set to $\trv^\prime$ with probability $\alpha(\trv^\prime,\trv^{(k)})$ and to $\trv^{(k)}$ with probability $1-\alpha(\trv^\prime,\trv^{(k)})$, where
\begin{equation}
\alpha\left(\trv^\prime,\trv^{(k)}\right) = 
\min\left\{1,\frac{\td(\trv^\prime)q_{\theta,\bar{\mapp}}(\trv^{(k)} | 
\trv^\prime)}{\td(\trv^{(k)})q_{\theta,\bar{\mapp}}(\trv^\prime | \trv^{(k)})} 
\right\}\label{eq:basicmh}.
\end{equation}  
The choice of proposal $q_{\trv,\bar{\mapp}}$ controls the dependence between successive states in the MCMC chain through both the acceptance rate and the step size. Knowledge of the target density $\td$ is helpful in designing proposals to make large moves that simultaneously have a high acceptance probability. The scheme presented here encodes information about the target distribution via a transport map $\imap$.

\subsection{MCMC with a fixed transport map}
\label{sec:mapmcmc:proposal}
Assume that we have an approximate transport map $\imap$ between a standard Gaussian reference and the target measure $\tmeas$, i.e., $\rmeas\approx\imap_{\sharp}\tmeas$. The pushforward of the target measure through this map will not be Gaussian. But a map that reduces the optimization objective of Section~\ref{sec:maps} will make the pushforward closer  (in this particular sense) to a standard Gaussian than the original target. We will then use MCMC to sample this pushforward distribution, with a proposal $q_r(\rrv' | \rrv)$. The proposal $q_r$ may be chosen quite freely, and examples below will encompass both local and independence proposals. Equivalently, one can view this process from the perspective of the target space by considering the pullback through the map $\imap$ of the proposal $q_r$; this map-induced proposal is applied to the original target density $\td$. Below we will describe our algorithm from this second perspective, but the first perspective of  transforming or ``preconditioning'' the target density may also provide useful intuition.

Let  $q_r(\rrv^\prime | \rrv)$ be a standard Metropolis-Hastings proposal on the reference space. The pullback of this proposal through $\imap$ induces a target-space proposal density written as
\begin{equation}
q_{\trv,\bar{\mapp}}(\trv^\prime | \trv) = q_\rrv \left ( \imap(\trv^\prime) | \imap(\trv)\right) \left | \det \dimap(\trv^\prime) \right | , \label{eq:mapprop}
\end{equation}
where $\bar{\mapp}$ denotes the dependency of this proposal on the map parameters.  To perform MCMC, we need the ability to evaluate this proposal density and to draw samples from it.  The expression \eqref{eq:mapprop} provides an easy way of evaluating the proposal density. Sampling from the proposal $q_{\trv,\bar{\mapp}}( \cdot | \trv)$ involves three s.pdf: (1) use the current target state $\trv$ to compute the current reference state, $r=\imap(\trv)$; (2) draw a sample $\rrv^\prime \sim q_\rrv(\rrv^\prime | \rrv)$ from the reference proposal; and (3) evaluate the inverse map at $\rrv^\prime$ to obtain a sample from the target proposal: $\theta^\prime = \imap^{-1}(\rrv^\prime)$. These steps are given as lines 4--6 of Algorithm \ref{alg:adaptalg} and illustrated in Figure \ref{fig:propProcess}. Ignoring the adaptation in lines 9--13, Algorithm \ref{alg:adaptalg} is equivalent to a standard Metropolis-Hastings algorithm on the target distribution, using $q_{\trv,\bar{\mapp}}(\trv^\prime | \trv)$ as a proposal.


Because of the map's lower triangular structure, evaluating the inverse map $\imap^{-1}(\rrv)$ only requires $\pd$ one-dimensional nonlinear solves. These one-dimensional problems can be tackled efficiently with a simple Newton method or, if the map is represented with polynomials, with a bisection solver based on Sturm sequences \cite{Wilf1978}.  We utilize the latter approach because of its robustness.  

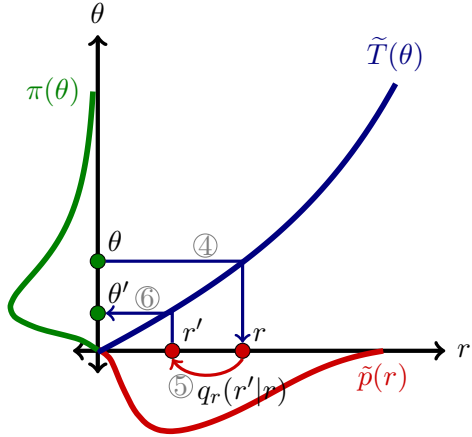
\begin{figure}
\centering


\colorlet{mapColor}{blue!50!black}
\colorlet{targetColor}{green!50!black}
\colorlet{refColor}{red!80!black}

\begin{tikzpicture}[
  input/.style={minimum size=0.2cm,draw=black,fill=red,circle,inner sep=0pt},
  output/.style={minimum size=0.2cm,draw=black,fill=green!70!black,circle,inner sep=0pt},
  joint/.style={draw=black,fill=blue!70!black,circle,inner sep=0pt},
  yscale=1, xscale = 1.1  ]

 draw axes
\draw[<->,ultra thick] (-0.3,0)--(4.2,0) node[right] {$\rrv$};
\draw[<->,ultra thick] (0,-0.3)--(0,4.2) node[above] {$\trv$};

\draw[line width=0.75mm, refColor, domain=0.01:3.45, smooth] plot (\x,{0.125-(0.8/\x)*exp(-0.9*(ln(0.8*\x))*(ln(0.8*\x)))-0.5*exp(-0.7*((\x-1.5))*((\x-1.5)))}) node[below]{$\tilde{p}(\rrv)$};

\draw[line width=0.75mm, mapColor, domain=0:3.6, smooth] plot (\x,{-0.125+0.5*\x+exp(0.8*\x-2.25)}) node[above]{$\imap(\trv)$};

\draw[line width=0.75mm, targetColor, domain=0.01:3.45, smooth] plot ({-(0.8/\x)*exp(-0.9*(ln(\x))*(ln(\x)))},\x) node[left] {$\td(\trv)$};

\node[circle, draw=black,fill=targetColor, minimum size=2mm,inner sep=0pt] (t1) at (0,1.192049908620754) {};
\node[above right] at (t1) {$\trv$};

\node[circle, draw=black, fill=refColor,minimum size=2mm,inner sep=0pt] (r1) at (1.75,0) {};
\node[above right] at (r1) {$\rrv$};
\node[circle,draw=black!50,inner sep=0pt,minimum size=3mm] at (1.3,1.4) {\color{black!50}$4$};
\draw[->,very thick, mapColor] (t1.east) -| (r1.north);

\node[circle, draw=black, fill=refColor,minimum size=2mm,inner sep=0pt] (r2) at (0.9,0) {};
\node[above right] at (r2) {$\rrv^\prime$};
\draw[->,very thick, refColor] (r1.south) to[out=-120,in=-60] (r2.south);
\node[below] at($ (r1.south) - (0,1mm) $) {$q_r(\rrv^\prime|\rrv)$};
\node[below right,circle,draw=black!50,inner sep=0pt,minimum size=3mm] at ($ (r2.south) - (0,2.5mm) $) {\color{black!50}$5$};

\node[circle, draw=black, fill=targetColor,minimum size=2mm,inner sep=0pt] (t2) at (0,0.5) {};
\node[above right] at (t2) {$\trv^\prime$};
\draw[->,very thick, mapColor] (r2.north) |- (t2.east);
\node[circle,draw=black!50,inner sep=0pt,minimum size=3mm] at (0.6,0.7) {\color{black!50}$6$};

\end{tikzpicture}


\caption[MCMC proposal process using transport maps.]{Illustration of the Metropolis-Hastings proposal process in transport map accelerated MCMC.  
The gray circled numbers on each arrow correspond to the line number in Algorithm \ref{alg:adaptalg}.}
\label{fig:propProcess}
\end{figure}


\subsection{Derivative-based proposals}
\label{sec:mapmcmc:deriv}
An important feature of our approach is that the map-induced proposal $q_{\trv,\bar{\mapp}}(\trv^\prime | \trv)$ requires derivative information from the target density $\td(\trv)$ if and only if the reference proposal $q_\rrv(\rrv^\prime | \rrv)$ explicitly requires derivative information.  
%
We also note that Algorithm \ref{alg:adaptalg} does not require $\td(\trv)$ to take any particular form (e.g., to be a Bayesian posterior or to result from a Gaussian prior). The ability to work with arbitrary target distributions for which derivative information may not be available is a distinction from many recent sampling approaches, such as Riemannian manifold MCMC \cite{Girolami2011}, the No-U-Turn Sampler of \cite{Hoffman2011}, or optimization-based samplers such as implicit sampling or RTO \cite{Morzfeld2012,Bardsley2014}.  
That said, though our approach can perform quite well without derivative information, we can still accommodate proposals that employ it.  



The reference proposal $q_\rrv$ is applied to the pushforward distribution of the target $\pi$ through the map $\imap$. Let  $\tilde{\rd}$ denote the corresponding pushforward density. 
Taking advantage of the map's lower triangular structure,  we can write the logarithm of this density as
\begin{equation}
\log\tilde{\rd}(\rrv) = \log\td \left (\imap^{-1}(\rrv) \right ) + \sum_{i=1}^\pd\log\frac{\partial \widetilde{T}_i^{-1}}{\partial \rrv_i} \label{eq:referenceTarget}.
\end{equation}
We will use the chain rule to obtain the gradient of this expression.  First, make the substitution $\rrv=\imap(\trv)$ and take the gradient with respect to $\trv$:
\begin{equation}
\nabla_\trv \log\tilde{\rd} \left (\imap(\trv) \right ) =  \nabla_\trv\log\td(\trv) -\sum_{i=1}^\pd \left(\frac{\partial \imap_i}{\partial \trv_i}\right)^{-1} H_i(\trv),
\end{equation}
where $H_i$ is a row vector of second derivatives coming from the determinant term: $H_i(\trv) = \left[\begin{array}{cccc}\frac{\partial^2 \imap_i}{\partial \trv_1 \partial \trv_i} & \frac{\partial^2 \imap_i}{\partial \trv_2 \partial \trv_i}  & \cdots & \frac{\partial^2 \imap_i}{\partial \trv_\pd \partial \theta_i}\end{array}\right]$.  Accounting for our change of variables, we now have an expression for the reference gradient given by
\begin{equation}
\nabla_r \log\tilde{\rd}(\rrv) =  \bigg(\nabla_\trv\log\td(\trv) -\sum_{i=1}^\pd \left(\frac{\partial \imap_i}{\partial \trv_i}\right)^{-1} H_i(\trv)\bigg) \left [ \nabla \imap(\trv) \right ]^{-1} \label{eq:refgrad}.
\end{equation}  
Note that this expression is only valid at $\trv=\imap^{-1}(\rrv)$. 

The lower triangular structure allows us not only to expand the determinant and obtain \eqref{eq:refgrad}, but also to apply the inverse Jacobian $(\nabla T(\trv))^{-1}$ easily through forward substitution.  Furthermore, computing the Jacobian $\nabla \imap(\trv)$  or the second derivatives in $H_i(\trv)$ is trivial when polynomials or other standard basis functions are used to parameterize the map.

\section{Adaptive transport map MCMC}
\label{sec:adaptmcmc}

Given more samples of the target distribution, we can construct a more accurate transport map, which in turn yields a more efficient map-accelerated proposal. Hence, we adaptively construct the the map $\imap$ as the MCMC chain progresses.

\subsection{Adaptive algorithm overview}
\label{sec:adaptmcmc:overview}

In our adaptive MCMC approach, we initialize the sampler with a simple map $\imap_0$ and update the map every $K_U$ s.pdf using the previous states of the MCMC chain. The map update uses these samples to define the optimization problem \eqref{eq:basicopt}, the solution of which yields a new map. This approach is conceptually similar to the adaptive Metropolis algorithm of \cite{Haario2001}.  In \cite{Haario2001}, however, previous states are used to update the covariance matrix of a Gaussian proposal; in the present case, previous states are used to construct a nonlinear transport map that yields more general non-Gaussian proposals. 

The most straightforward version of our adaptive algorithm would find the coefficients $\mapp_i$ for each component of the map by solving \eqref{eq:basicopt} directly.  However, when the number of existing samples $K$ is small or if the initial s.pdf of the chain mix poorly, the Monte Carlo sum in \eqref{eq:basicopt} will be a poor approximation of the true integral, producing maps that do not capture the structure of $\td$.  This is a standard issue in adaptive MCMC. One way to overcome this problem is to start adapting the map only after some initial exploration of the parameter space, i.e., after drawing a sufficient number of MCMC samples using the initial map $\imap_0$. A more efficient alternative, however, is to introduce a regularization term $g(\mapp_i)$ into the objective, allowing the map to start adapting much earlier. The purpose of this term is to ensure that the map does not prematurely collapse onto one region of the target space; such a collapse would make it difficult for the chain to efficiently explore the entire support of $\td$. Regularization yields the following modified objective:
\begin{equation}
\begin{aligned}
& \underset{\mapp_i}{\min}
& & g(\mapp_i) + \sum_{k=1}^{K}\left[  \frac12\imap^2_i(\trv^{(k)}; \mapp_i) - \log\left.\frac{\partial \imap_i(\trv;\mapp_i)}{\partial \trv_i}\right|_{\trv^{(k)}}\right]\\
& \text{s.t.}
& & \left.\frac{\partial \imap_i(\trv;\mapp_i)}{\partial \trv_i}\right|_{\trv^{(k)}} \geq  \lambda_{\text{min}}, \quad \forall k\in\{1,2,\ldots,K\} .
\end{aligned}\label{eq:regopt}
\end{equation}
In practice, we choose $g(\mapp_i)$ to prevent $\imap$ from deviating too strongly from the identity map, particularly when $K$ is small.  If additional problem structure such as the covariance of $\td$ were known, it could also be incorporated into the regularization term. But in the typical case, we use a simple quadratic penalty function centered on the coefficients of the identity map: letting $\mapp_i^{\text{Id}}$ denote the coefficients of the identity map, we put $g(\mapp_i) = k_R \|\mapp_i-\mapp_i^{\text{Id}}\|^2$ where $k_R$ is a user-defined regularization parameter. We have found that on most problems, small values of $k_R$ yield similar performance. (In the numerical examples below, we mostly set $k_R=10^{-4}$.)  Because we have discarded the $\frac{1}{K}$ coefficient scaling the Monte Carlo sum in \eqref{eq:regopt}, the second term of the objective overwhelms the regularization term as the number of samples grows, and the value of $k_R$ eventually becomes unimportant.

Lines 9--13 of Algorithm \ref{alg:adaptalg} show how we incorporate the map update into our adaptive MCMC framework.

\begin{algorithm}[h!]
\DontPrintSemicolon
\KwIn{Initial state $\trv_0$, initial vector of transport map parameters $\mappa_0$, reference proposal $q_\rrv( \cdot \vert \rrv^{(k)})$, number of s.pdf $K_U$ between map adaptations, total number of steps $L$.}
\KwOut{MCMC samples of the target distribution, $\left\{\trv^{(1)},\trv^{(2)},\ldots,\trv^{(L)}\right\}$}
Set state $\trv^{(1)} = \trv_0$\;
Set parameters $\mappa^{(1)} = \mappa_0$\;
\For{$k \gets 1 \ldots L-1$}{
Compute the reference state, $\rrv^{(k)} = \imap(\trv^{(k)}; \mappa^{(k)})$\;
Sample the reference proposal, $\rrv^\prime\sim q_\rrv(\cdot\vert\rrv^{(k)})$\;
Compute the target proposal sample, $\trv^\prime = \imap^{-1}(\rrv^\prime; \mappa^{(k)})$\;
Calculate the acceptance probability:
\[
\alpha = \min\left\{1, \frac{\td(\imap^{-1}(\rrv^\prime; \mappa^{(k)}))}{\pi(\imap^{-1}(\rrv^{(k)}; \mappa^{(k)}))} \frac{q_r\left(\rrv^{(k)}\vert\rrv^\prime\right)}{q_r\left(\rrv^\prime\vert\rrv^{(k)}\right)}\frac{\det[D\imap^{-1}(\rrv^\prime; \mappa^{(k)})]}{\det[D\imap^{-1}(\rrv^{(k)}; \mappa^{(k)})]}\right\}
\]\;
Set $\trv^{(k+1)}$ to $\trv^\prime$ with probability $\alpha$; else set $\trv^{(k+1)} = \trv^{(k)}$.\;
\eIf{$(k\bmod{K_U}) = 0$}{
\For{$i\gets 1 \textbf{ to } \pd$}{
Update $\mapp_i^{(k+1)}$ by solving (\ref{eq:regopt}) with $\left\{\trv^{(1)},\trv^{(2)},\ldots,\trv^{(k+1)}\right\}$\;
}
}{
$\mappa^{(k+1)} = \mappa^{(k)}$\;
}
}
\Return{\textrm{Target samples} $\left\{\trv^{(1)},\trv^{(2)},\ldots,\trv^{(L)}\right\}$}\;
\caption{MCMC algorithm with adaptive map}
\label{alg:adaptalg}
\end{algorithm}

\subsection{Sequential map updates}
\label{sec:adaptmcmc:mapupdate}

At first glance, updating the map every $K_U$ MCMC iterations might seem computationally taxing.  Fortunately, the form of the optimization problem in \eqref{eq:regopt} allows for efficient updates.  When $K_U$ is small relative to the current number of s.pdf $K$, the objective function in \eqref{eq:regopt} changes little between updates and the previous map coefficients provide a good initial guess for the new optimization problem.  Thus new optimal coefficients can be found in only a few Newton iterations, sometimes only one or two. As the timing results in Section \ref{sec:perf} show, even for long chains (large $K$), the advantage of using the map to define $q_{\trv,\bar{\mapp}}$ greatly outweighs the computational costs of sequential map updates.  

We also note that the optimization could be performed with stochastic approximation (SA) techniques \cite{Kushner2003,Andrieu2006}, in which case each map update would use only a portion of the chain, and would have a cost independent of $K$. Our tests with $K$ up to $\num{5e5}$ have shown SAA to be more efficient, but even longer chains might favor an SA approach.

%

\subsection{Monitoring map convergence}
\label{sec:adaptmcmc:convg}
As the map in Algorithm \ref{alg:adaptalg} is adapted,  the pushforward of $\td$ through the map becomes closer to the reference Gaussian, and the best choice of reference proposal $q_\rrv(\rrv \vert \rrv^\prime)$ will evolve as well.  A small-scale random walk proposal may be appropriate at early iterations, but a larger and perhaps position-independent proposal may be advantageous as the map captures more of the target distribution's structure.  By monitoring the difference between $\tilde{\rd}$ \eqref{eq:referenceTarget} and the uncorrelated standard Gaussian density, we can adapt the reference proposal $q_\rrv$ to better explore the changing $\tilde{\rd}$.

To this end, it is important to have an indicator of the map's current accuracy. In the discussion following \eqref{eq:klexpr}, we noted that $\log \td - \log p \circ \imap - \log | \det \dimap |$ becomes a constant function of $\trv$ when an exact transformation \eqref{eq:measconst} between the target and reference is achieved. 
A useful way to monitor the map's convergence is then to calculate the variance of this quantity,
\begin{equation}
\sigma^2_M = \Varsimp_{\trv} {\left[\log \td(\trv) -  \log p\left(\imap(\trv)\right) -\log \left | \det \dimap(\trv) \right | \right] } .
\label{eq:mapvar}
\end{equation}
A variance of zero indicates that the map is exact: $\tilde{\rd}$ is a standard Gaussian. 
Asymptotically, as $\sigma^2_M \rightarrow 0$, the KL divergence \eqref{eq:klexpr} becomes $2 \sigma^2_M$ \cite{Moselhy2011}. 



\subsection{Choice of reference proposal}
\label{sec:adaptmcmc:refprop}

Until now, we have left the choice of reference proposal $q_\rrv(\rrv^\prime \vert \rrv)$  rather open.  Indeed, any non-adaptive proposal, including both independence proposals and random walk proposals, could be used within our framework.  Figure \ref{fig:propTypes} shows some typical proposals on both the reference space and the target space.  In this section, we describe a few reference proposals that we will use in our numerical demonstrations, with particular attention to how they are implemented within the transport map framework. This selection is far from exhaustive, and is only intended to indicate how the transport map can dictate the choice of reference proposal.

\newcommand{\propScale}{0.22}
\begin{figure}[h]
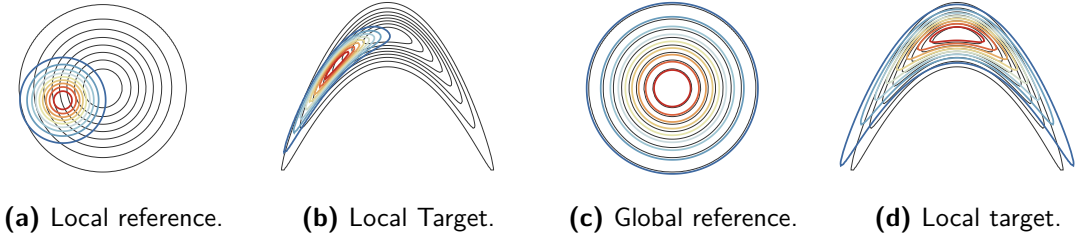

\centering

\begin{subfigure}{0.24\textwidth}
\input{pics/mcmc/LocalMapReference.tex}
\caption{Local reference.}
\end{subfigure}
\begin{subfigure}{0.24\textwidth}
\input{pics/mcmc/LocalMapBanana.tex}
\caption{Local Target.}
\end{subfigure}
\begin{subfigure}{0.24\textwidth}
\input{pics/mcmc/GlobalMapReference.tex}
\caption{Global reference.}
\end{subfigure}
\begin{subfigure}{0.24\textwidth}
\input{pics/mcmc/GlobalMapBanana.tex}
\caption{Local target.}
\end{subfigure}


 \caption[Illustration of map-induced proposals.]{Example proposals in the reference space and the target space.  Plots of both $q_\rrv(\rrv^\prime \vert \rrv)$ and $q_\trv(\trv^\prime\vert\rrv)=q_\rrv(\rrv^\prime \vert \rrv) |D\imap(\trv^\prime)|$ are shown for local and independence (global) proposals. The black contours depict the target distributions while the colored contours illustrate the proposal densities.}
 \label{fig:propTypes}
 \end{figure}
 
\textbf{Metropolis-adjusted Langevin (MALA) proposal:}  Discretizing an appropriate Langevin equation yields a proposal of the form:
\begin{equation}
q_{\text{MALA}}(\rrv^\prime \vert \rrv) = \mathcal{N} \left( \rrv + \frac{(\Delta \tau)^2}{2} \nabla_{\rrv} \log \tilde{\rd}(\rrv), \, (\Delta \tau)^2 \, I \right),
\label{eq:malaprop}
\end{equation}
with a s.pdfize $(\Delta \tau)^2$ and a symmetric positive definite matrix $\Sigma$ \cite{Roberts1996}.  

\textbf{Delayed rejection proposals:} The delayed-rejection (DR) MCMC scheme of \cite{Mira2001} allows several proposals to be attempted during each MCMC step. With such a multi-stage proposal, we can try a larger or more aggressive proposal at the first stage, followed by more conservative proposals likely to produce accepted moves.  We use this scheme to define $q_\rrv(\rrv^\prime \vert \rrv)$ in two ways.  

Our first instantiation of DR employs a standard Gaussian as an independence proposal in the first stage, followed by a Gaussian random walk proposal in the second stage. Our motivation for this global-then-local strategy is the evolving nature of $\tilde{\rd}(\rrv)$.  Initially, $\tilde{\rd}(\rrv)$ will resemble the target density, which is more efficiently sampled by the random walk proposal; we need samples to be accepted in order to build a good map. As the map adapts, however, $\tilde{\rd}(\rrv)$ will approach a standard normal density, which can be efficiently explored by the position-independent first stage. DR naturally trades off between these alternatives.
%
%
Figure~\ref{fig:propTypes} illustrates the difference between local and independence proposals for a simple banana-shaped distribution.
Our second instantiation of DR employs two symmetric random-walk proposals, the first with a larger variance and the second with a smaller variance.

\section{Convergence analysis}
\label{sec:convergence}
This section investigates conditions under which our adaptive algorithm yields an ergodic chain. Proofs of the lemmas are deferred to Appendix \ref{sec:detailedConvergence}.

\subsection{The need for bounded derivatives}
\label{sec:convergence:bounds}
Consider a random walk proposal on the reference space, $q_\rrv(\rrv^\prime | \rrv) = N(\rrv, \sigma^2 I)$ with some fixed variance $\sigma^2$.  For illustration, assume that the target density is a standard normal distribution: $\td(\trv) = N(0,I)$.  The RWM algorithm is geometrically ergodic for any density satisfying the following two conditions (see Theorem 4.3 of \cite{Jarner2000}):
\begin{equation}
\limsup_{\| \trv \|\rightarrow \infty} \frac{\trv}{\| \trv \|} \cdot \nabla\log\td(\trv) = -\infty,\label{eq:superLight}
\end{equation}
and
\begin{equation}
\lim_{\| \trv \|\rightarrow \infty} \frac{\trv}{\| \trv \|} \cdot \frac{\nabla\log\td(\trv)}{\|\nabla\log\td(\trv)\|} < 0.\label{eq:curveCond}
\end{equation}
Densities that satisfy \eqref{eq:superLight} are called super-exponentially light.
   It is easy to show that our example Gaussian density satisfies these conditions.  In Algorithm \ref{alg:adaptalg}, however, instead of applying the RWM proposal to $\td$ directly, we apply the RWM proposal to the map-induced density in \eqref{eq:referenceTarget}.
If the conditions in (\ref{eq:ubMap}) are not satisfied, we can show that even when $\td$ is Gaussian, any monotone {polynomial} map with degree greater than one results in a density $\tilde{\rd}(\rrv)$ that is no longer super-exponentially light.  For example, let $\imap$ have a maximum polynomial degree of $M>1$, with $M$ odd.  Then:
\begin{eqnarray}
\limsup_{\| \rrv \|\rightarrow \infty} \frac{\rrv}{\| \rrv \|} \cdot \nabla\log\tilde{\rd}(\rrv) &=& \limsup_{\| \rrv \|\rightarrow \infty}\frac{1}{\|r\|}\sum_{i=1}^\pd r_i \left(\frac{\partial \imap^{-1}_i}{\partial \rrv_i}\right)^{-1}\frac{\partial^2 \imap^{-1}_i}{\partial \rrv_i^2} \nonumber \\
&=&\limsup_{\| \rrv \|\rightarrow \infty}\frac{\pd}{\|r\|}\left(\frac{1}{M}-1\right)= 0 .\label{eq:zeroexplight}
\end{eqnarray}
%
%
Clearly, the map-induced density is not super-exponentially light.  We have therefore jeopardized the geometric ergodicity of our sampler on a simple Gaussian target.  Additional restrictions on the map are needed to ensure convergence. 

The loss of geometric ergodicity in (\ref{eq:zeroexplight}) is due to the unbounded derivatives of nonlinear polynomial maps, which do not satisfy (\ref{eq:dUB}). Unbounded derivatives of $\imap$ imply that $\imap^{-1}$ has derivatives that approach zero as $\| \rrv \|\rightarrow \infty$, which leads to (\ref{eq:zeroexplight}).  More intuitively, without an upper bound on their derivatives, polynomial maps move too much weight to the tails of $\tilde{\rd}$.  In the next section, we show that the conditions in \eqref{eq:dUB} ensure the ergodicity of Algorithm \ref{alg:adaptalg}, even with map adaptation.

\subsection{Convergence of the adaptive algorithm}
Our goal in this section is to show that the adaptive Algorithm \ref{alg:adaptalg} produces samples that can be used in Monte Carlo approximations.  We thus need to show that Algorithm \ref{alg:adaptalg} is ergodic for the target density $\td(\trv)$.  

Assume that the target density is finite, continuous, and super-exponentially light.  (Note that certain densities which are not super-exponentially light can be transformed to super-exponentially light 
densities using the techniques from \cite{Johnson2012}.) Also assume that the reference proposal $q_\rrv(\rrv^\prime \vert \rrv)$ is Gaussian with bounded mean.
Furthermore, let $\Gamma$ be the space of the map parameters $\bar{\mapp}$ such that $\imap(\trv;\bar{\mapp})$ satisfies the bi-Lipschitz condition given by (\ref{eq:dUB}).

The map at iteration $k$ of the MCMC chain is defined by the coefficients $\bar{\mapp}^{(k)}$.  Let $\tk_{\bar{\mapp}^{(k)}}$ be the transition kernel of the chain at iteration $k$, constructed from the map $\imap(\trv;\bar{\mapp}^{(k)})$, the target space proposal in (\ref{eq:mapprop}), and the Metropolis-Hastings kernel:
\begin{equation} 
\tk_{\bar{\mapp}^{(k)}}(\trv,\mathcal{A}) = \int_{\mathcal{A}} \left (
  \alpha(\trv^\prime, \trv)q_{\trv,\bar{\mapp}^{(k)}}(\trv^\prime |
  \trv) +\left(1-r(\trv)\right)\delta_\trv(\trv^\prime) \right ) d\trv^\prime \label{eq:mhkernel}.
\end{equation} 
Here  $q_{\trv,\bar{\mapp}^{(k)}}$ is the map-induced proposal density from 
(\ref{eq:mapprop}),  $\alpha(\trv^\prime, \trv)$ is the acceptance probability 
defined in (\ref{eq:basicmh}), and $r(\trv) = \int \alpha(\trv^\prime, \trv) 
q_{\trv,\bar{\mapp}^{(k)}}(\trv^\prime | \trv) d\trv^\prime$.   Now, following 
\cite{Roberts2007} and \cite{Bai2009}, we can establish the ergodicity of our adaptive 
algorithm by showing that it satisfies two conditions: diminishing adaptation 
and containment. Diminishing adaptation is defined as follows:

\vspace{0.3cm}
\begin{definition}[Diminishing adaptation]\label{def:dim}
For any starting point $x^{(0)}$ and initial set of map parameters
$\bar{\mapp}^{(0)}$, a transition kernel $\tk_{\bar{\mapp}^{(k)}}$
satisfies the diminishing adaptation condition when
\begin{equation}
\lim_{k\rightarrow \infty}\sup_{x \in \real^\pd}\left\| \tk_{\bar{\mapp}^{(k)}}(x,\cdot) - \tk_{\bar{\mapp}^{(k+1)}}(x,\cdot)\right\|_{TV} = 0 \quad\text{  in probability}
\end{equation}
where $\|\cdot\|_{TV}$ denotes the total variation norm.  
\end{definition}
\vspace{0.3cm}

Instead of working with the containment condition directly (see \cite{Bai2009} or \cite{Roberts2007}), we will show that our adaptive MCMC algorithm instead satisfies the simultaneous strongly aperiodic geometric ergodicity (SSAGE) condition.

\vspace{0.3cm} 
\begin{definition}[SSAGE]
\label{def:ssage}
Simultaneous strongly aperiodic geometric ergodicity (SSAGE) is the condition that there exist a measurable set $C\in \bspace$, a drift function $V:\real^\pd\rightarrow [1,\infty)$, and scalars $\delta>0$, $\lambda<1$, and $b<\infty$ such that $\sup_{x\in C} V(x) < \infty$ and the following two conditions hold:
\begin{enumerate}
\item \emph{(Minorization)} For each vector of map parameters $\bar{\mapp}\in\Gamma$, there is a probability measure $\nu_{\bar{\mapp}}(\cdot)$ defined on $C\subset \real^\pd$ with $\tk_\mapp(x,\cdot)\geq \delta \nu_{\bar{\mapp}}(\cdot)$ for all $x\in C$.
\item \emph{(Simultaneous drift)} $\int_{\real^\pd} V(x) \tk_{\bar{\mapp}}(x,dx) \leq \lambda V(x) + b I_C(x)$ for all $\bar{\mapp}\in\Gamma$ and $x\in\real^\pd$.
\end{enumerate}
\end{definition}
 \vspace{0.3cm}
 
By Theorem 3 of \cite{Roberts2007}, SSAGE ensures the containment condition.  The following three lemmas establish diminishing adaptation and SSAGE.  
 In the following, let $C=B(0,R_C)$ be a ball of radius $R_C>0$ and let $V(x) = k_v\td^{-\alpha}(x)$ for some $\alpha\in(0,1)$ and $k_v = \sup_x\td^{\alpha}(x)$.   Also, assume that $\td(x)>0$ for all $x\in C$.  For this choice of $V(x)$ and our assumption that $\td(x)>0$ for $x\in C$, we have that $\sup_{x\in C} V(x) < \infty$. 
 
Because the reference proposal is Gaussian with bounded mean, we can find two 
scalars $k_1$ and $k_2$, and two zero-mean Gaussian densities $g_1$ and 
$g_2$, such that the reference proposal is bounded as
 \begin{equation}
 k_1g_1(\rrv^\prime-\rrv)\leq q_\rrv(\rrv^\prime | \rrv) \leq k_2g_2(\rrv^\prime - \rrv)\label{eq:qrBounds}.
 \end{equation}
The bounds in (\ref{eq:dUB}) then imply that the target space proposal can also be bounded.  This result is captured in Lemma \ref{lem:propBnds}.

\vspace{0.3cm}
  \begin{lemma}[Bounded target space proposal]
  \label{lem:propBnds}
 For any map coefficients $\bar{\mapp} \in \Gamma$, the map-induced proposal $q_{\trv,\bar{\mapp}}(\trv^\prime | \trv)$ is bounded as
 \begin{equation}
 k_L g_L (\trv^\prime-\trv)\leq q_{\trv,\bar{\mapp}}(\trv^\prime | \trv) \leq k_Ug_U(\trv^\prime - \trv)\label{eq:qtBounds},
 \end{equation}
   where $k_L = k_1 \lambda_{\mathrm{min}}^\pd$, $k_U = k_2 \lambda_{\mathrm{max}}^\pd$,  $g_L(x) = g_1( \lambda_{\mathrm{max}} x)$, and $g_U(x) = g_2( \lambda_{\mathrm{min}} x)$.  
   \end{lemma}
   \vspace{0.3cm}

The upper and lower bounds in (\ref{eq:qtBounds}) are key to our proof of convergence.  In fact, with these bounds, the proofs of Lemma \ref{lem:minor} and Lemma \ref{lem:drift} below closely follow the proof of Proposition 2.1 in \cite{Atchade2006}.  Again, proofs of these results are left to the appendix.
  
 \vspace{0.3cm}
\begin{lemma}[Diminishing adaptation of Algorithm \ref{alg:adaptalg}]
\label{lem:dimAdapt}
Let the map parameters $\bar{\mapp}$ be restricted to a compact subset of $\Gamma$.  Then, the sequence of transition kernels defined by the update step in lines 9--13 of Algorithm \ref{alg:adaptalg} satisfies the diminishing adaptation condition. 
\end{lemma}
\vspace{0.3cm}

\begin{lemma}[Minorization condition for Algorithm \ref{alg:adaptalg}]
\label{lem:minor}
There is a scalar $\delta$ and a set of probability measures $\nu_{\bar{\mapp}}$ defined on $C$ such that $\tk_{\bar{\mapp}}(x,\cdot) \geq \delta \nu_{\bar{\mapp}}(\cdot)$ for all $x\in C$ and $\bar{\mapp}\in\Gamma$.
\end{lemma}
\vspace{0.3cm}

\begin{lemma}[Drift condition for Algorithm \ref{alg:adaptalg}]
\label{lem:drift}
For all points $x\in \real^\pd$ and all feasible map parameters $\bar{\mapp}\in\Gamma$, there are scalars $\lambda$ and $b$ such that $\int_{\real^\pd} V(x) \tk_{\bar{\mapp}}(x,dx) \leq \lambda V(x) + b I_C(x)$
\end{lemma}
\vspace{0.3cm}


\noindent With Lemmas \ref{lem:dimAdapt}--\ref{lem:drift} in hand, Theorem \ref{thm:ergodic} finally yields the ergodicity of our adaptive algorithm.
\vspace{0.3cm}

\begin{theorem}[Ergodicity of Algorithm \ref{alg:adaptalg}]
\label{thm:ergodic}
Algorithm \ref{alg:adaptalg} is ergodic for the target distribution $\td(\trv)$ 
when $\bar{\mapp}$ is constrained to a compact set within which $\imap(\trv;\bar{\mapp})$ is 
guaranteed to satisfy (\ref{eq:dUB}) for all $\trv \in \real^\pd$.
\end{theorem}
\begin{proof} 
Lemmas \ref{lem:minor} and \ref{lem:drift} ensure that SSAGE is satisfied, which subsequently ensures containment.  The diminishing adaptation property from Lemma \ref{lem:dimAdapt} combined with SSAGE implies ergodicity by Theorem 3 of \cite{Roberts2007}.
\end{proof}

\section{Numerical examples}
\label{sec:perf}
Here we compare the performance of Algorithm \ref{alg:adaptalg} with that of several existing MCMC methods, including delayed rejection adaptive Metropolis (DRAM) \cite{Haario2006}, simplified manifold MALA (sMMALA) \cite{Girolami2011}, adaptive MALA (AMALA) \cite{Atchade2006}, and the No-U-Turn Sampler (NUTS) \cite{Hoffman2011}.
For a full comparison, we will pair transport maps with several different reference proposal mechanisms: a random walk (TM+RW), both varieties of delayed rejection discussed in Section~\ref{sec:adaptmcmc:refprop} (denoted by TM+DRG for the global/independence proposal and TM+DRL for local proposals), and a MALA proposal (TM+LA).  To explore the strengths and weaknesses of each algorithm, we consider three test problems that provide a range of target distributions.

Throughout our results, the minimum effective sample size (ESS) over all parameter dimensions is used to evaluate MCMC performance. We run \emph{multiple} independent chains for each sampler and extract the median integrated autocorrelation time for each dimension, then take the worst case over dimensions; details on this ESS evaluation are provided in Appendix \ref{sec:esscalc}. Larger effective sample sizes correspond to smaller variances of estimates computed from MCMC samples. To illustrate the computational cost of each method, we also report the ESS normalized by run time and by the number of function evaluations.  Posterior density evaluations and gradient evaluations are summed when normalizing by ``function evaluation.''

\subsection{Biochemical oxygen demand model}
\label{sec:perf:bod}
In water quality monitoring, the simple biochemical oxygen demand model given by $B(t) = \trv_0(1-\exp(-\trv_1 t))$ is often fit to observations of $B(t)$ at early times (e.g., $t<5$)~\cite{Sullivan2010}.  In this example, we wish to infer $\trv_1$ and $\trv_2$ given $N$ observations at times $\{t_1, t_2, \ldots, t_N\}$. We use $20$ observations evenly spread over $[1,5]$, with additive Gaussian errors, $y(t_i) = \trv_0(1-\exp(-\trv_1 t_i)) + e$, where $e\sim N(0,\sigma_B^2)$ and $\sigma_B^2=2 \times 10^{-4}$. 

Our synthetic data come from evaluating $B(t_i)$ with $\trv_0=1$ and $\trv_1=0.1$ and sampling $e$.  Using a uniform improper prior over $\real^2$, we have the target posterior given by
\begin{equation}
\log\td(\trv_0,\trv_1) =  - 2\pi \sigma_B^2  - \frac12 \sum_{i=1}^2 \left[\trv_0(1-\exp(-\trv_1 t_i))- y(t_i)\right]^2
\end{equation}
It is easy to obtain gradients of the posterior density, allowing us to again compare many different MCMC algorithms. For each algorithm, we run 30 independent chains starting at the posterior mode; each chain is run for $\num{7.5e4}$ iterations, with the first $\num{1e4}$ iterations discarded as burn-in. Results are shown in Table \ref{tab:bodPerf1}.

\begin{table}
\caption[Comparison of MCMC performance on BOD problem.]{\label{tab:bodPerf1} Performance of MCMC samplers on the BOD problem. $\tau_{\text{max}}$ is the maximum integrated autocorrelation time, where the maximum is taken over all dimensions; ESS is the corresponding minimum effective sample size. Result are averaged over multiple independent runs of each sampler, and $\sigma_{\tau}$ is the empirical standard deviation of $\tau_{\text{max}}$ over these runs.}
\centering

\small
\fbox{\begin{tabular}{l|rr|rrr|D{.}{.}{4}D{.}{.}{2}}
Method & $\tau_{\max}$ & $\sigma_\tau$ & ESS & ESS/sec & ESS/eval &  \begin{minipage}{0.75cm}\centering Rel.\\[-0.15cm] ESS/sec\end{minipage} & \begin{minipage}{0.75cm}\centering Rel.\\[-0.15cm] ESS/eval\end{minipage} \\\hline
DRAM &  59.2 & 24.6 & 551 & 1.04e-01 & 4.23e-03 & 1.00 & 1.00\\
NUTS &  14.7 &  1.0 & 2214 & 4.97e-02 & 1.20e-03 & 0.48 & 0.28\\
sMMALA &  84.4 & 14.4 & 385 & 1.05e-03 & 2.57e-03 & 0.01 & 0.61\\
AMALA &  42.1 & 11.7 & 771 & 1.46e-01 & 5.14e-03 & 1.40 & 1.22\\\hline
TM+DRG &   2.1 &  0.7 & 15660 & 1.44e+00 & 1.61e-01 & 13.85 & 38.06\\
TM+DRL &   4.5 &  0.5 & 7174 & 6.13e-01 & 5.90e-02 & 5.89 & 13.95\\
TM+RWM &   5.0 &  0.2 & 6558 & 7.98e-01 & 8.73e-02 & 7.67 & 20.64\\
TM+LA &  854.9 & 340.3 & 38 & 2.94e-03 & 2.53e-04 & 0.03 & 0.06\\
\end{tabular}}

\end{table}

%
%

In this example, we represent the map with total-order Hermite polynomials of degree three.  The additional nonlinear terms help capture the changing posterior correlation structure shown in Figure \ref{fig:bodPosterior}, which is challenging for standard samplers to explore.  Methods like DRAM and AMALA may capture the global covariance, but this covariance is often not representative of the local structure and does not provide enough information for efficient posterior sampling. Other methods like sMMALA and NUTS use derivative information to capture local geometry, but the local geometry varies considerably and is not sufficiently representative of the global structure, making it difficult for these samples to take large jumps through the parameter space.  Our transport map proposals, on the other hand, are capable of capturing the global non-Gaussian structure of Figure \ref{fig:bodPosterior}; in fact, the pushforward of this target density through the map becomes much more Gaussian, as shown in Figure \ref{fig:bodPosteriorMap}.  Map-based methods with global independence proposals (e.g., TM+DRG) can then efficiently ``jump'' across the entire parameter space, yielding the much shorter integrated autocorrelation times shown in Table \ref{tab:bodPerf1}. 

Another interesting result in Table \ref{tab:bodPerf1} is the poor performance of TM+LA.  In this example, the basic MALA algorithm was not able to sufficiently explore the space on its own (or equivalently, with an initial identity map); hence, poor exploration in the early stages of Algorithm~\ref{alg:adaptalg} hindered good adaptation and resulted in the inefficient sampling shown here.  Because of this poor performance, the TM+LA algorithm will not be employed in our other test problems.

\begin{figure}[h]
\centering
\begin{subfigure}{0.40\textwidth}
\centering
\begin{tikzpicture}

\begin{axis}[%
width=0.9\textwidth,
height=0.9\textwidth,
xmin=0.4,
xmax=1.3,
ymin=0.08,
ymax=0.28,
xlabel={$\trv_1$},
ylabel={$\trv_2$},
axis on top=true,
]
\addplot graphics[xmin=0.4,ymin=0.08,xmax=1.3,ymax=0.28] {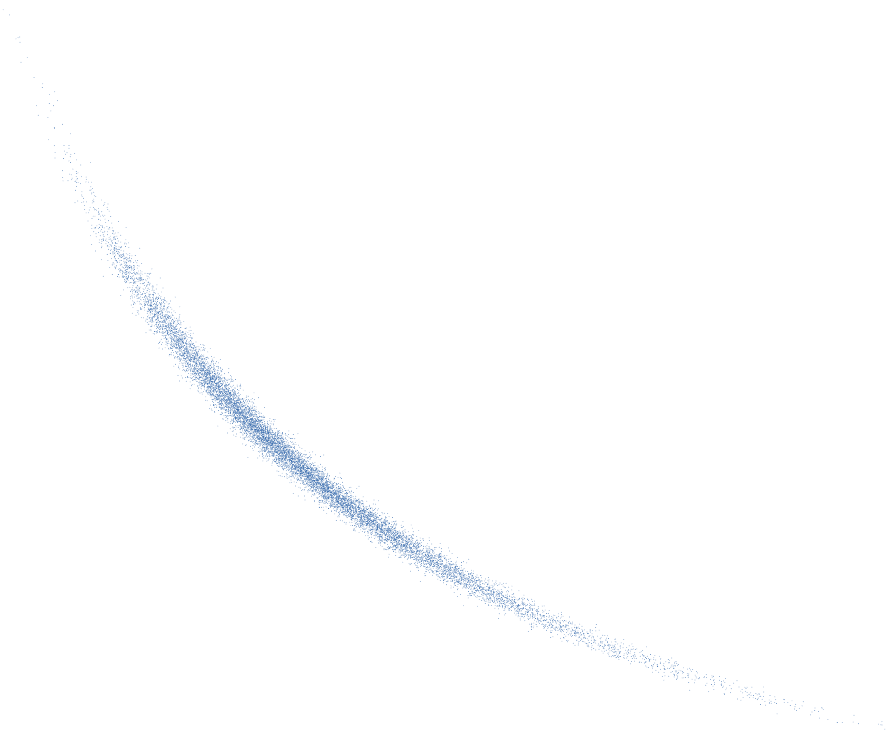};
\end{axis}
\end{tikzpicture}%
\caption{Scatter plot of BOD posterior samples.}
\label{fig:bodPosterior}
\end{subfigure}
\begin{subfigure}{0.40\textwidth}
\centering
\input{pics/mcmc/BOD_Final_Reference.tex}
\caption{Contours of the map-induced density in the reference space $\tilde{\rd}(\rrv)$.}
\label{fig:bodPosteriorMap}
\end{subfigure}

\caption[Comparison of BOD posterior and map-preconditioned density.]{The narrow high-density region and changing correlation structure of the target distribution on the left is difficult for many samplers.  The transport map approach, after adaptation, pushes forward the original target to the distribution shown on the right, which can be sampled much more effectively.}

\end{figure}

\subsection{Predator-prey system}
\label{sec:perf:predprey}
The previous example has a posterior density whose derivatives are easy to evaluate in closed form.  However, many realistic inference problems involve complex likelihoods for which derivative information is expensive to compute. This example illustrates such a situation; we consider parameter inference in an ODE model of a predator-prey system,
 \begin{eqnarray}
  \frac{dP}{dt} & = & rP\left(1-\frac{P}{K}\right) - s\frac{PQ}{a+P}\nonumber\\
  \frac{dQ}{dt} & = & u\frac{PQ}{a+P} - vQ  , \label{eq:predPreyModel}
 \end{eqnarray}
where $(P, Q)$ are the prey and predator populations and $r$, $K$, $s$, $a$, $u$, and $v$ are model parameters. See \cite{Rockwood2006} for model details and the ecological meaning of these parameters.
In addition to these six parameters, we infer the initial conditions $P(0)$ and $Q(0)$ from five noisy observations of both $P$ and $Q$ at times regularly spaced on $[0,50]$.  The observations are perturbed with independent Gaussian observational errors with mean zero and variance $10$. We generate the data using the following ``true'' parameter values:
  \begin{equation}
\left[ P^\ast(0), Q^\ast(0), r^\ast, K^\ast, s^\ast, a^\ast, u^\ast, v^\ast \right]^T = \left[50, 5, 0.6, 100, 1.2, 25, 0.5, 0.3\right]^T \label{eq:truePredPreyParams}.
 \end{equation}
The MCMC chain is run on a set of parameters $\theta$ that are scaled by these true parameters.

 
The prior for this problem is uniform over the intersection of a hypercube in parameter space,  $[0.001, 50]^8$, and the set of parameters that produce cyclic solutions.  The cyclic solution requirement can be enforced by examining the Jacobian of \eqref{eq:predPreyModel} at its fixed points. A fixed point, denoted by $[P_f,Q_f]$, must satisfy $P_f>0$ and $Q_f>0$ and the Jacobian of the right hand side of (\ref{eq:predPreyModel}) must have eigenvalues with positive real components when evaluated at $[P_f,Q_f]$ \cite{Strogatz2001}.
 
 \begin{figure}
\centering

\newcommand{\groupPlotWidth}{0.2000\textwidth}
\begin{tikzpicture}[scale=0.6]
\begin{groupplot}[group style={every plot/.style={enlargelimits=true},group size=8 by 8,xlabels at=edge bottom,ylabels at=edge left, xticklabels at=edge bottom,yticklabels at=edge left,vertical sep=1pt,horizontal sep=1pt},height=\groupPlotWidth,width=\groupPlotWidth,ticks=none,enlargelimits=false]

\nextgroupplot[title={\large$P(0)$}]
\addplot graphics[xmin=0.737,ymin=1.255,xmax=0.000,ymax=6.133] {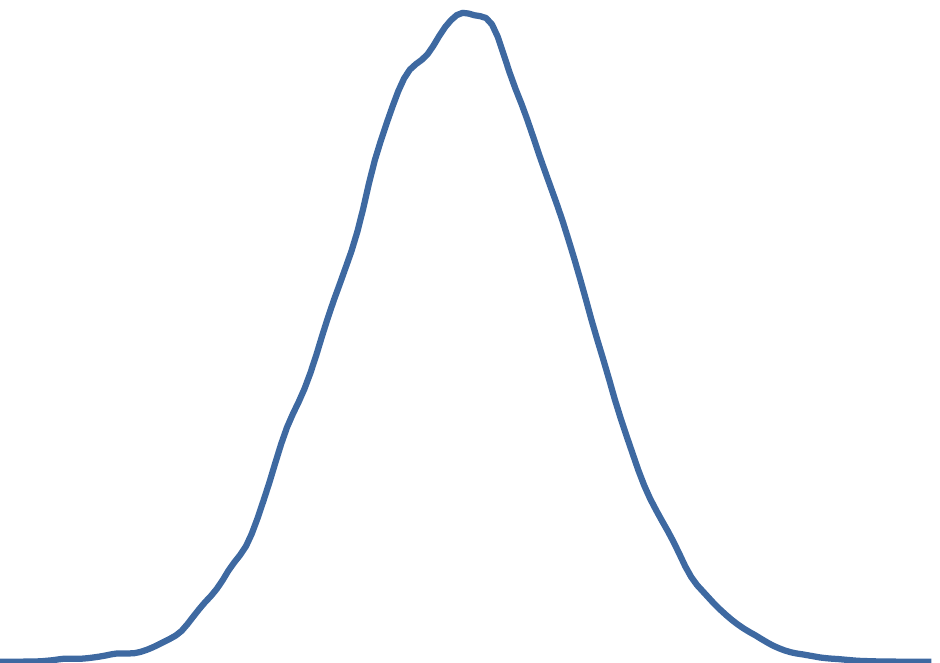};

\nextgroupplot[group/empty plot]
\nextgroupplot[group/empty plot]
\nextgroupplot[group/empty plot]
\nextgroupplot[group/empty plot]
\nextgroupplot[group/empty plot]
\nextgroupplot[group/empty plot]
\nextgroupplot[group/empty plot]

\nextgroupplot
\addplot graphics[xmin=0.737,ymin=1.255,xmax=0.203,ymax=2.628] {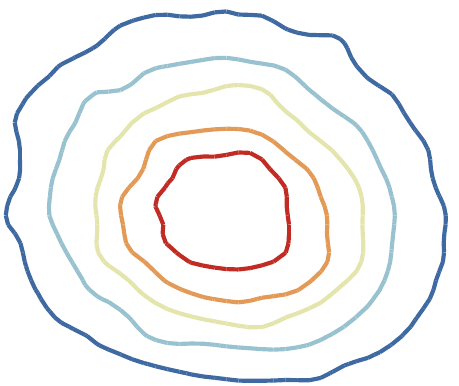};
\nextgroupplot[title={\large$Q(0)$}]
\addplot graphics[xmin=0.203,ymin=2.628,xmax=0.000,ymax=1.226] {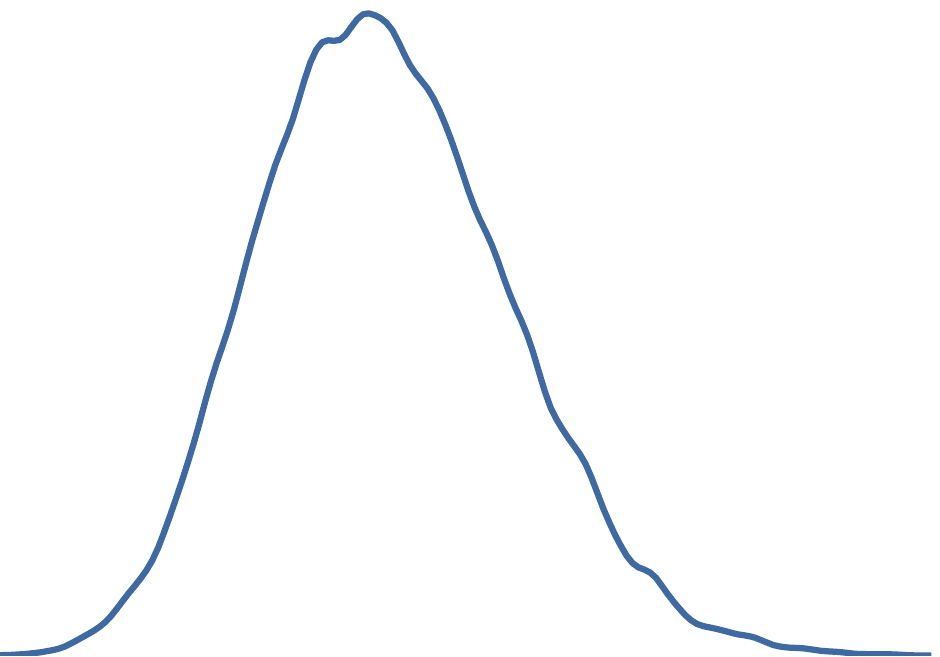};

\nextgroupplot[group/empty plot]
\nextgroupplot[group/empty plot]
\nextgroupplot[group/empty plot]
\nextgroupplot[group/empty plot]
\nextgroupplot[group/empty plot]
\nextgroupplot[group/empty plot]

\nextgroupplot
\addplot graphics[xmin=0.737,ymin=1.255,xmax=0.269,ymax=1.583] {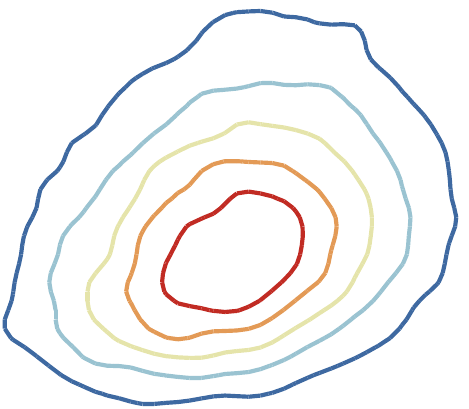};
\nextgroupplot
\addplot graphics[xmin=0.203,ymin=2.628,xmax=0.269,ymax=1.583] {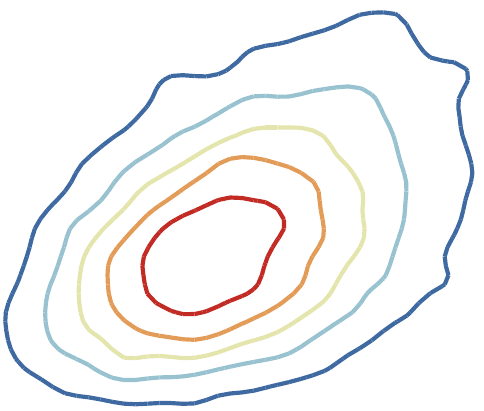};
\nextgroupplot[title={\large$r$}]
\addplot graphics[xmin=0.269,ymin=1.583,xmax=0.000,ymax=2.212] {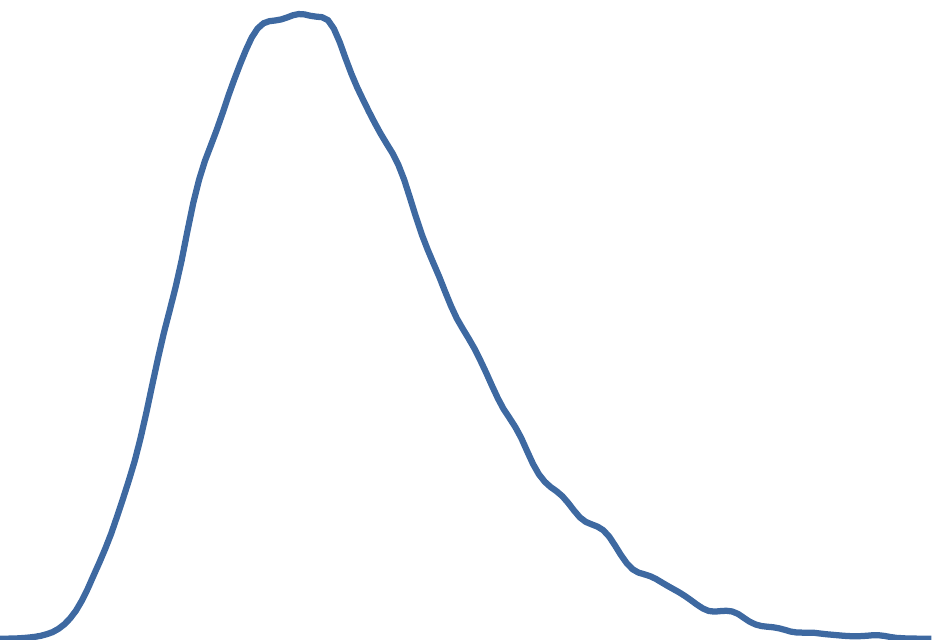};

\nextgroupplot[group/empty plot]
\nextgroupplot[group/empty plot]
\nextgroupplot[group/empty plot]
\nextgroupplot[group/empty plot]
\nextgroupplot[group/empty plot]

\nextgroupplot
\addplot graphics[xmin=0.737,ymin=1.255,xmax=0.895,ymax=3.660] {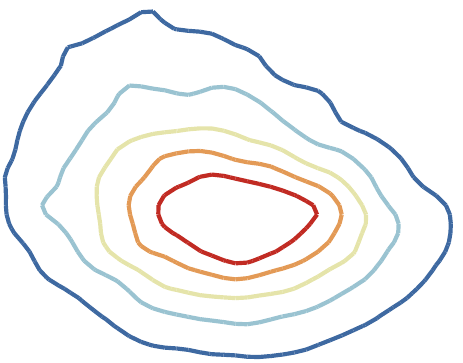};
\nextgroupplot
\addplot graphics[xmin=0.203,ymin=2.628,xmax=0.895,ymax=3.660] {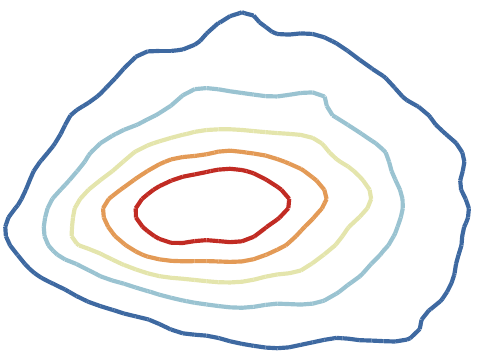};
\nextgroupplot
\addplot graphics[xmin=0.269,ymin=1.583,xmax=0.895,ymax=3.660] {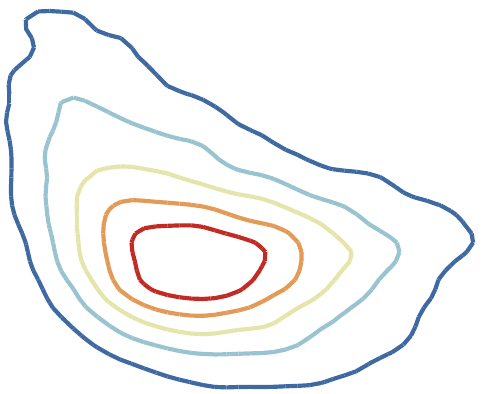};
\nextgroupplot[title={\large$K$}]
\addplot graphics[xmin=0.895,ymin=3.660,xmax=0.000,ymax=1.288] {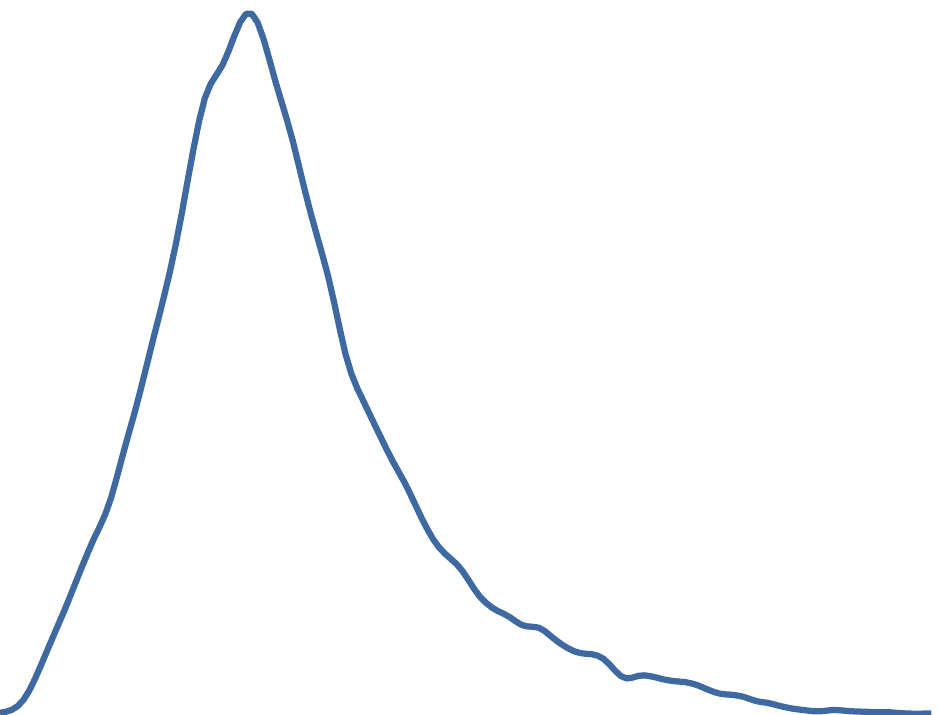};

\nextgroupplot[group/empty plot]
\nextgroupplot[group/empty plot]
\nextgroupplot[group/empty plot]
\nextgroupplot[group/empty plot]

\nextgroupplot
\addplot graphics[xmin=0.737,ymin=1.255,xmax=0.605,ymax=5.573] {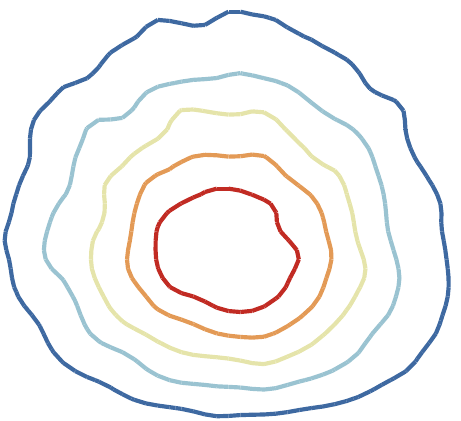};
\nextgroupplot
\addplot graphics[xmin=0.203,ymin=2.628,xmax=0.605,ymax=5.573] {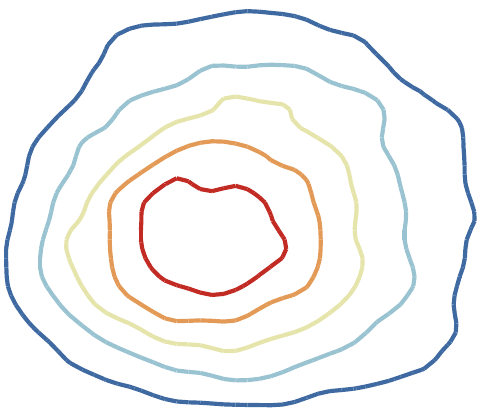};
\nextgroupplot
\addplot graphics[xmin=0.269,ymin=1.583,xmax=0.605,ymax=5.573] {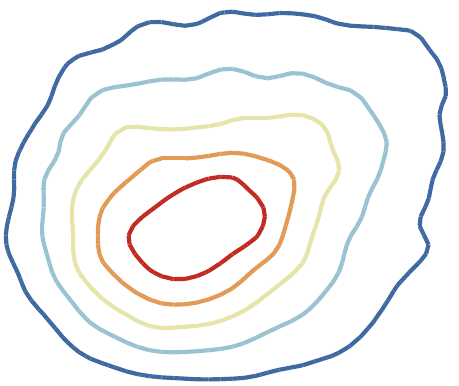};
\nextgroupplot
\addplot graphics[xmin=0.895,ymin=3.660,xmax=0.605,ymax=5.573] {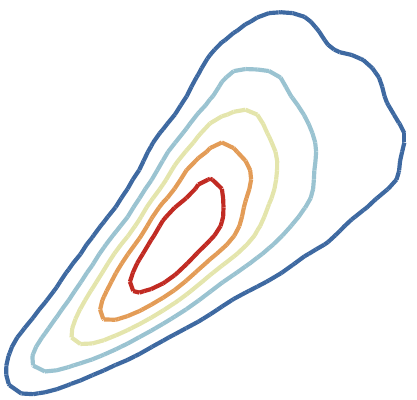};
\nextgroupplot[title={\large$s$}]
\addplot graphics[xmin=0.605,ymin=5.573,xmax=0.000,ymax=0.507] {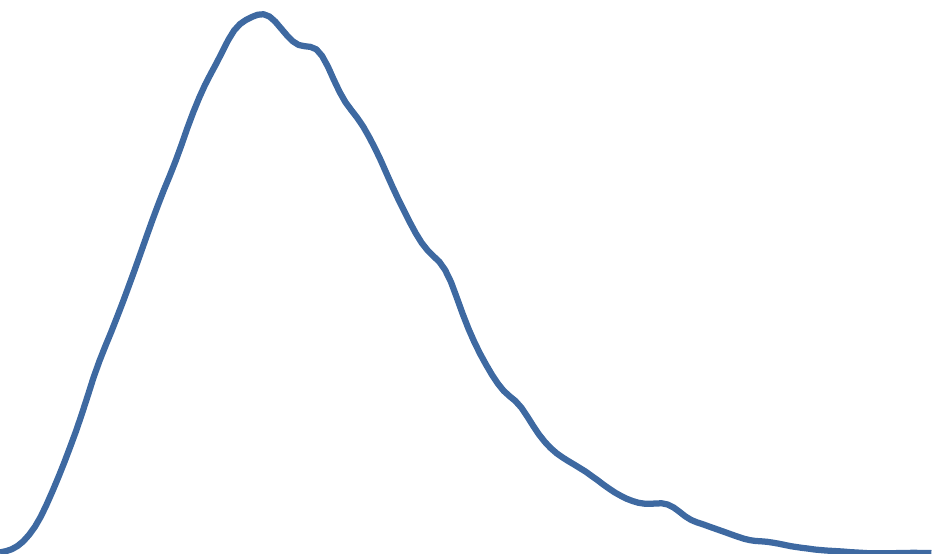};

\nextgroupplot[group/empty plot]
\nextgroupplot[group/empty plot]
\nextgroupplot[group/empty plot]

\nextgroupplot
\addplot graphics[xmin=0.737,ymin=1.255,xmax=0.566,ymax=9.850] {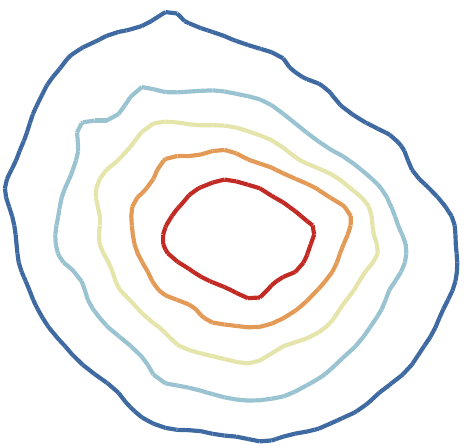};
\nextgroupplot
\addplot graphics[xmin=0.203,ymin=2.628,xmax=0.566,ymax=9.850] {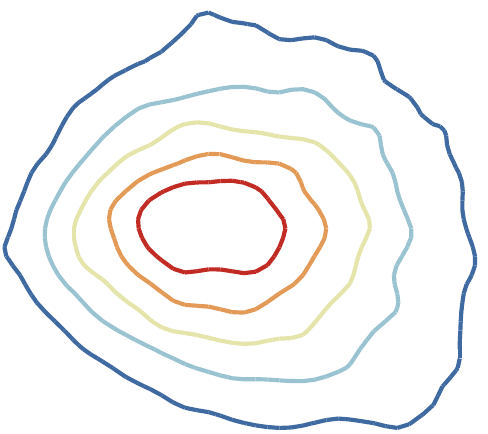};
\nextgroupplot
\addplot graphics[xmin=0.269,ymin=1.583,xmax=0.566,ymax=9.850] {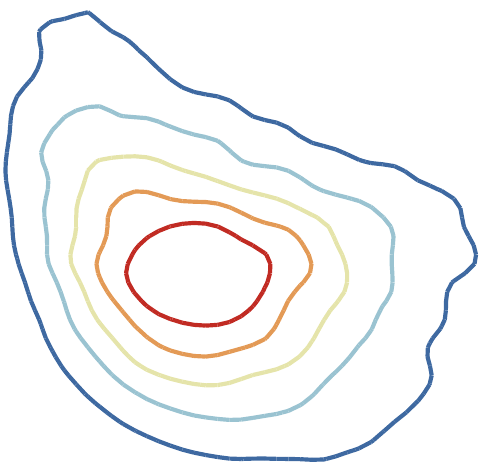};
\nextgroupplot
\addplot graphics[xmin=0.895,ymin=3.660,xmax=0.566,ymax=9.850] {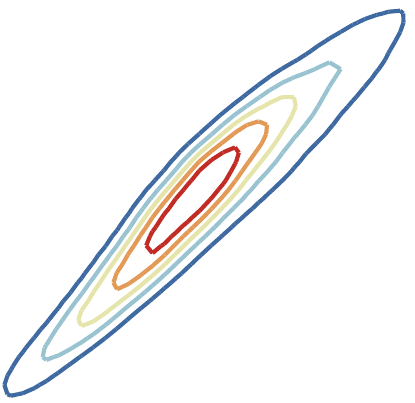};
\nextgroupplot
\addplot graphics[xmin=0.605,ymin=5.573,xmax=0.566,ymax=9.850] {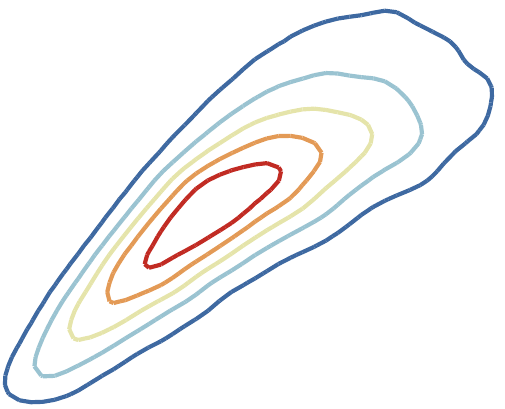};
\nextgroupplot[title={\large$a$}]
\addplot graphics[xmin=0.566,ymin=9.850,xmax=0.000,ymax=0.304] {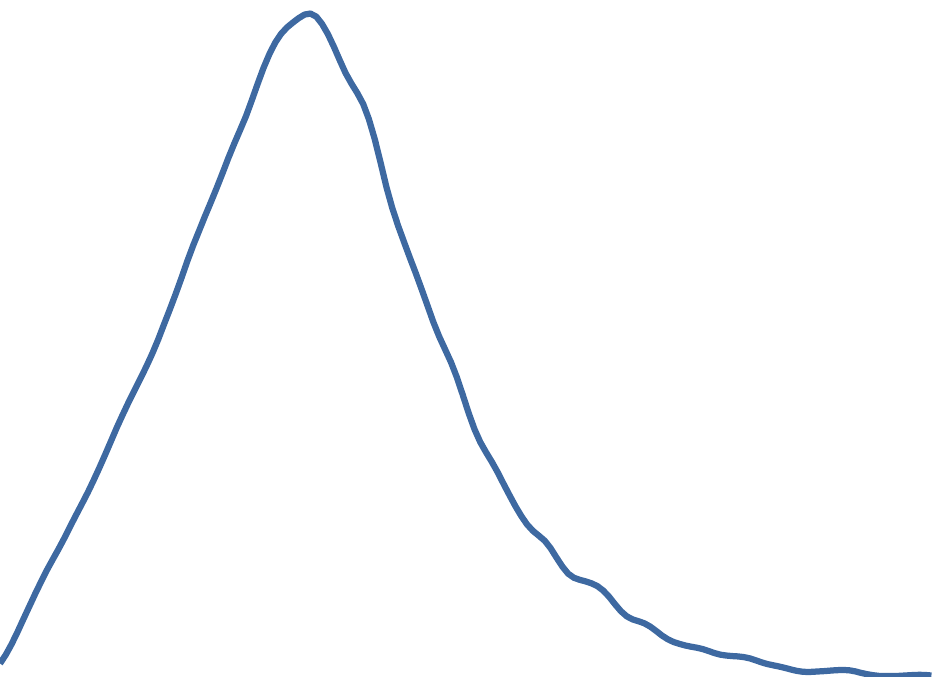};

\nextgroupplot[group/empty plot]
\nextgroupplot[group/empty plot]

\nextgroupplot
\addplot graphics[xmin=0.737,ymin=1.255,xmax=0.134,ymax=2.544] {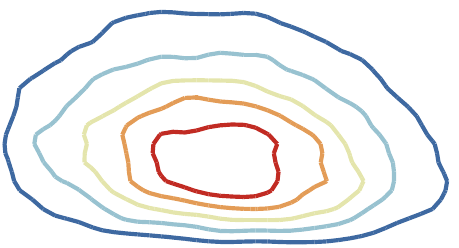};
\nextgroupplot
\addplot graphics[xmin=0.203,ymin=2.628,xmax=0.134,ymax=2.544] {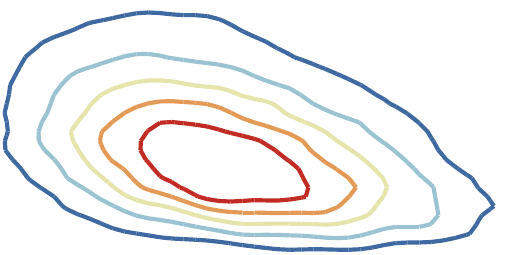};
\nextgroupplot
\addplot graphics[xmin=0.269,ymin=1.583,xmax=0.134,ymax=2.544] {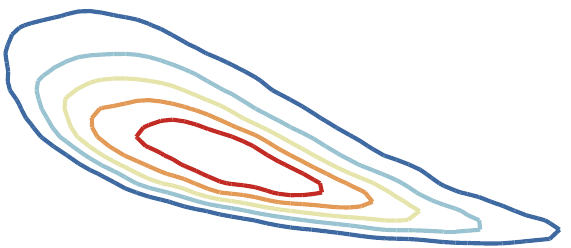};
\nextgroupplot
\addplot graphics[xmin=0.895,ymin=3.660,xmax=0.134,ymax=2.544] {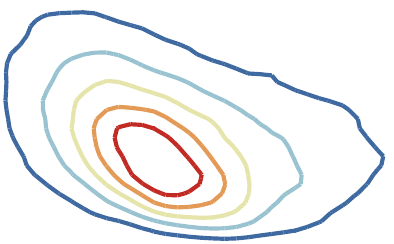};
\nextgroupplot
\addplot graphics[xmin=0.605,ymin=5.573,xmax=0.134,ymax=2.544] {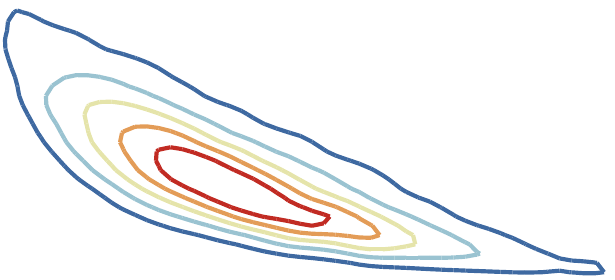};
\nextgroupplot
\addplot graphics[xmin=0.566,ymin=9.850,xmax=0.134,ymax=2.544] {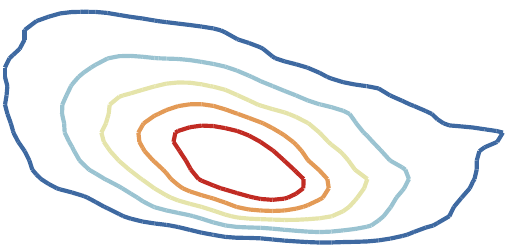};
\nextgroupplot[title={\large$u$}]
\addplot graphics[xmin=0.134,ymin=2.544,xmax=0.000,ymax=1.875] {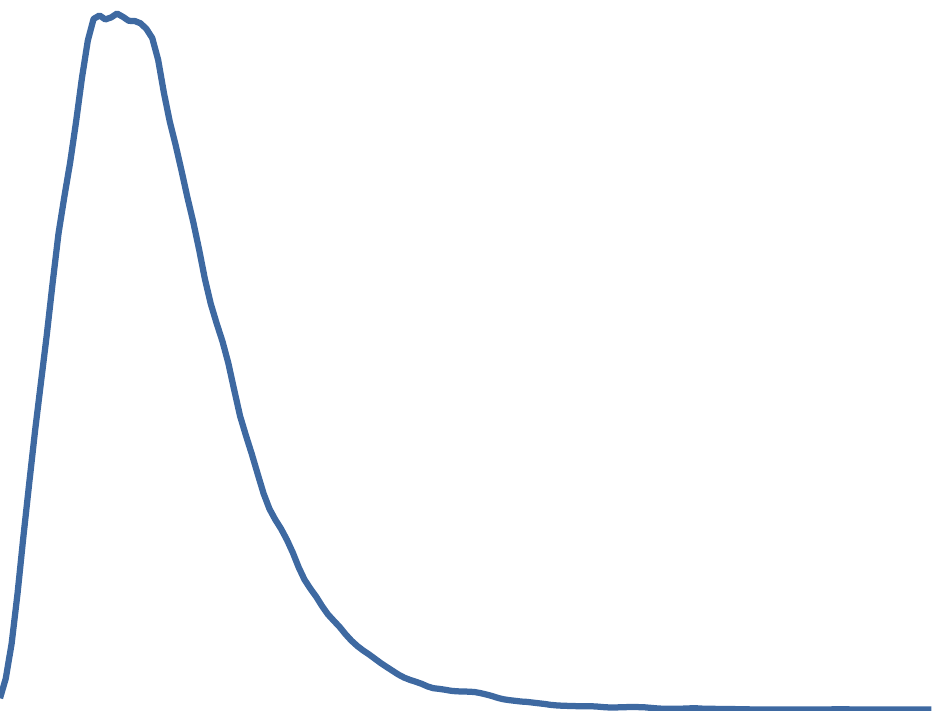};

\nextgroupplot[group/empty plot]

\nextgroupplot
\addplot graphics[xmin=0.737,ymin=1.255,xmax=0.273,ymax=2.028] {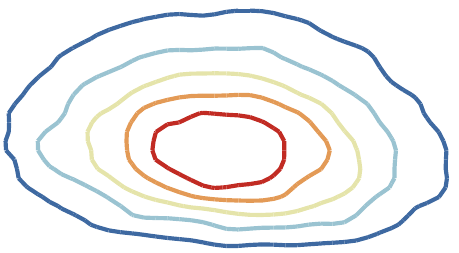};
\nextgroupplot
\addplot graphics[xmin=0.203,ymin=2.628,xmax=0.273,ymax=2.028] {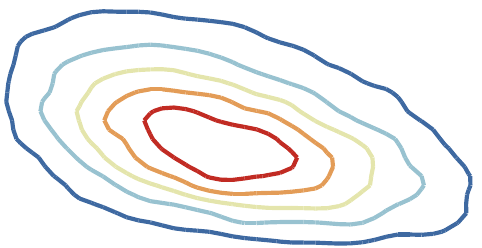};
\nextgroupplot
\addplot graphics[xmin=0.269,ymin=1.583,xmax=0.273,ymax=2.028] {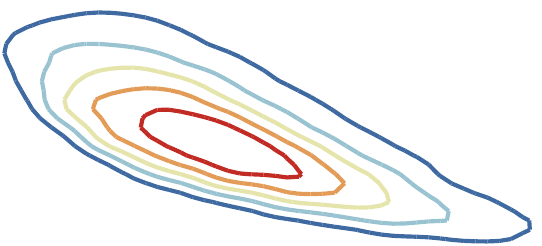};
\nextgroupplot
\addplot graphics[xmin=0.895,ymin=3.660,xmax=0.273,ymax=2.028] {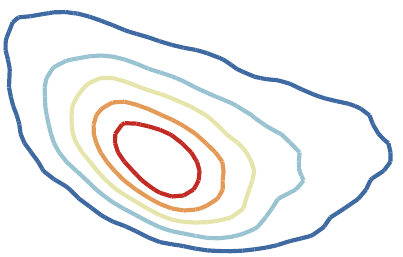};
\nextgroupplot
\addplot graphics[xmin=0.605,ymin=5.573,xmax=0.273,ymax=2.028] {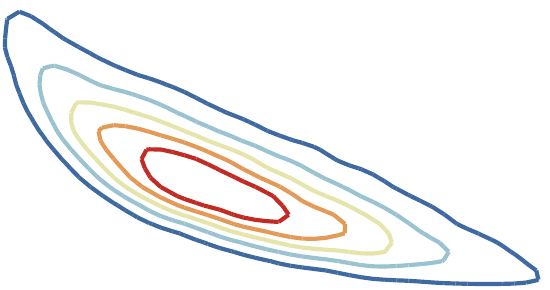};
\nextgroupplot
\addplot graphics[xmin=0.566,ymin=9.850,xmax=0.273,ymax=2.028] {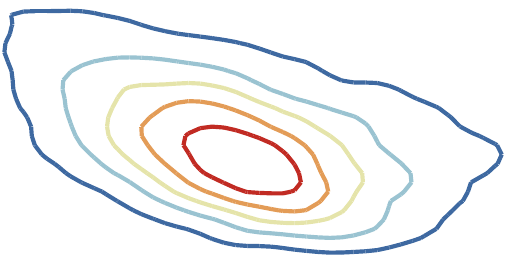};
\nextgroupplot
\addplot graphics[xmin=0.134,ymin=2.544,xmax=0.273,ymax=2.028] {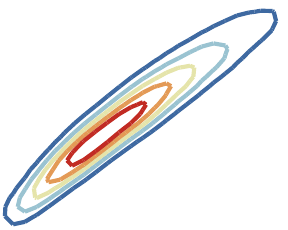};
\nextgroupplot[title={\large$v$}]
\addplot graphics[xmin=0.273,ymin=2.028,xmax=0.000,ymax=2.486] {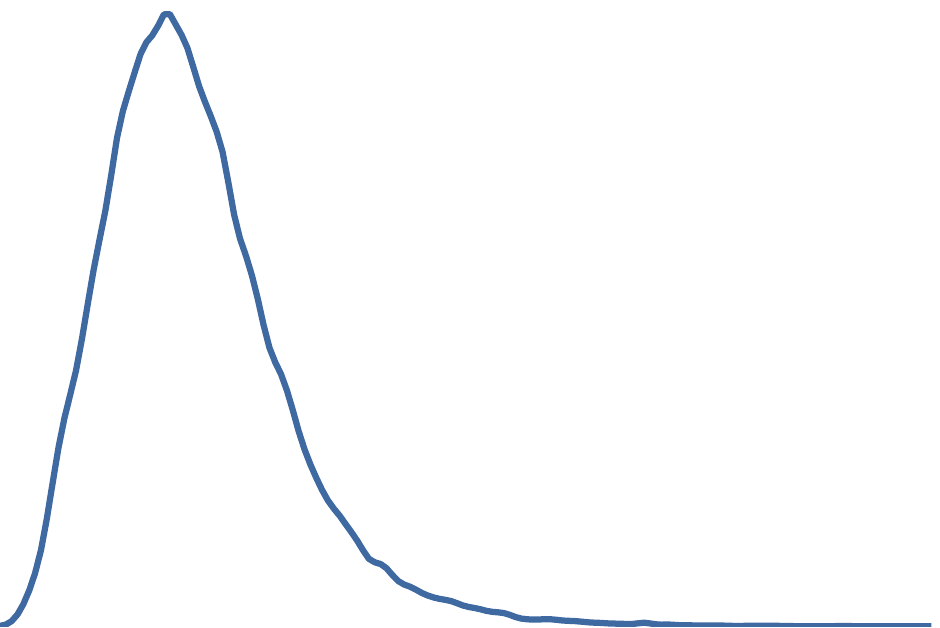};

\end{groupplot}
\end{tikzpicture}

\caption[Predator-prey posterior distribution]{Posterior distribution for the predator-prey inference example.}
\label{fig:predPreyPost}
\end{figure}

The posterior distribution of the parameters is shown in Figure \ref{fig:predPreyPost}. While not as narrow as the BOD posterior, this target distribution is non-Gaussian and its various marginals have changing local correlation structures.  Figure \ref{fig:predPreyTrace} shows trace plots for each algorithm, while Table \ref{tab:predPreyPerf} shows a performance comparison of the samplers. Results are computed for two different chain lengths: chains of $\num{1.2e5}$ s.pdf, with the first $\num{5e4}$ steps discarded as burn-in; and longer chains with $\num{5e5}$ total steps, discarding the first $\num{2e5}$ as burn-in.  The longer chains are intended as a check to validate the perfomance conclusions drawn from shorter chains in the other examples. The transport map algorithms used multivariate Hermite polynomials of total degree three.

For the shorter chains, each algorithm was started at the posterior mode, and 30 independent runs of each sampler were used to generate the results.  The longer chains were started with random initial points taken from the prior, and 100 independent runs of each sampler were performed.  All derivative information was computed by solving the forward sensitivity equations corresponding to (\ref{eq:predPreyModel}).  
%
Even though we would expect NUTS to have a large effective sample size on this problem, NUTS was not included here because of the intractable number of gradient evaluations it required. Our initial tests indicated that roughly 40 days would be required to run our full numerical comparison with NUTS.

\begin{figure}
\centering

\begin{tikzpicture}

\begin{groupplot}[group style={every plot/.style={enlargelimits=false},group size=3 by 2, xlabels at=edge bottom,ylabels at=edge left, xticklabels at=edge bottom,yticklabels at=edge left,vertical sep=3pt,horizontal sep=3pt},enlargelimits=false,axis on top=true, xtick={0,50000,100000}, scaled x ticks=false,
    xticklabels={$0$,$5e4$,$1e5$}, xmin=200000,xmax=250000,ymin=0.25,ymax=3.0,height=3.5cm,width=5cm,no markers,every axis title/.style={below right,at={(0,1)}}, xlabel={MCMC Step}, ylabel={$P(0)$}]
    \nextgroupplot[title=DRAM]
    \addplot graphics[xmin=200000,ymin=0.25,xmax=250000,ymax=2.5] {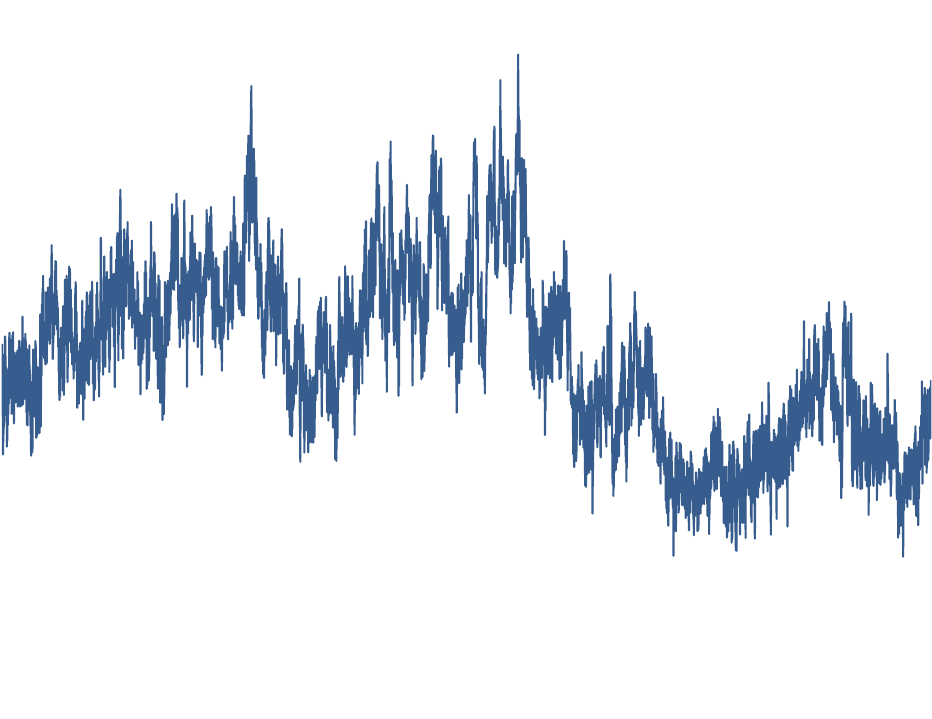};
    \nextgroupplot[title=sMMALA]
    \addplot graphics[xmin=200000,ymin=0.25,xmax=250000,ymax=2.5] {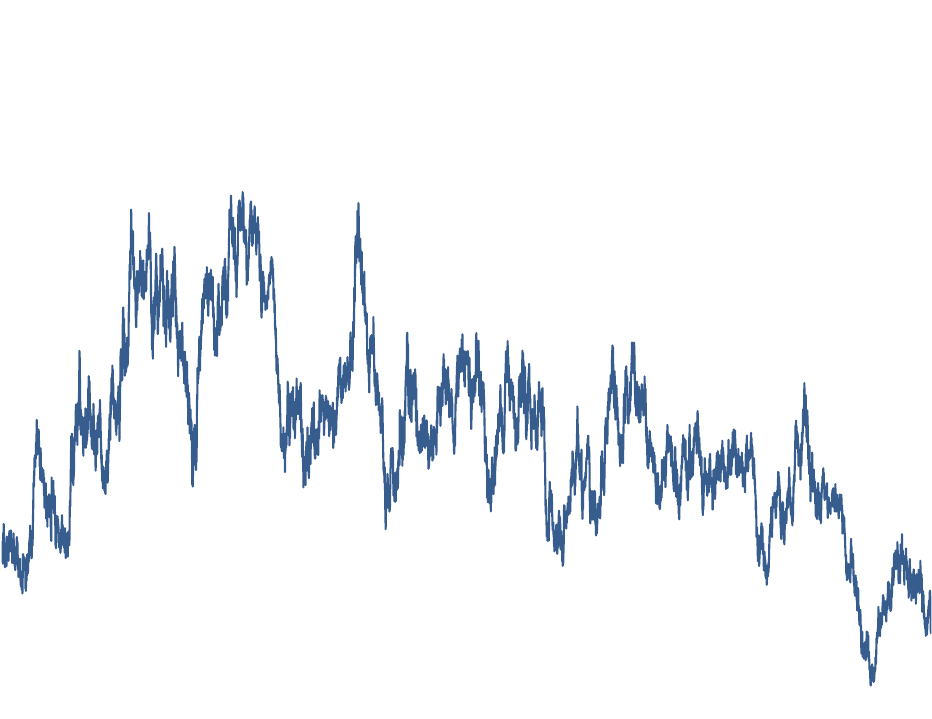};
    \nextgroupplot[title=AMALA]
    \addplot graphics[xmin=200000,ymin=0.25,xmax=250000,ymax=2.5] {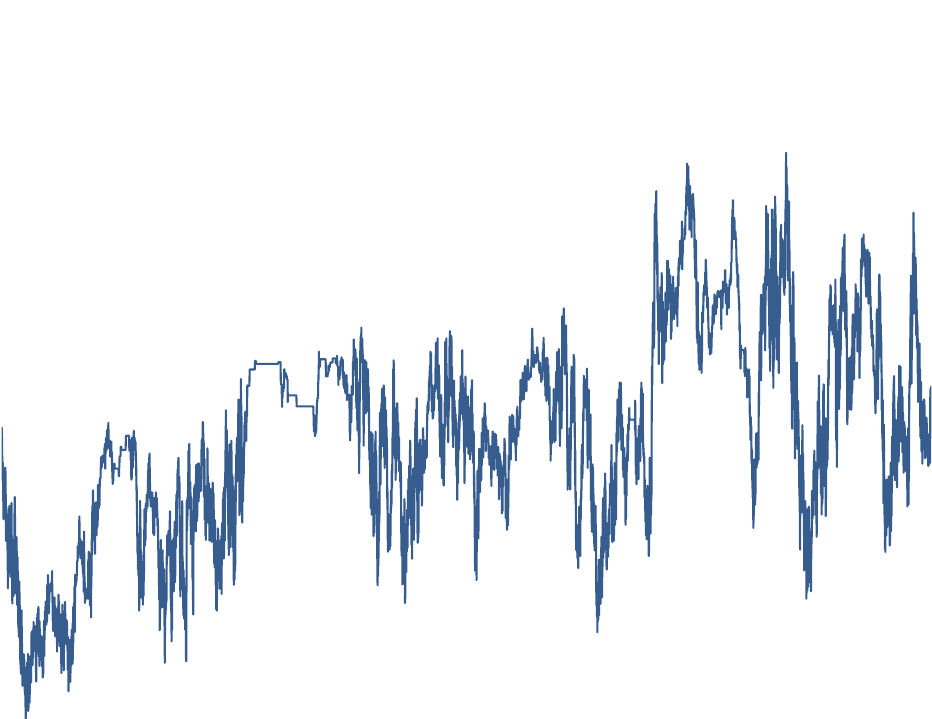};
    \nextgroupplot[title=TMDRG]
    \addplot graphics[xmin=200000,ymin=0.25,xmax=250000,ymax=2.5] {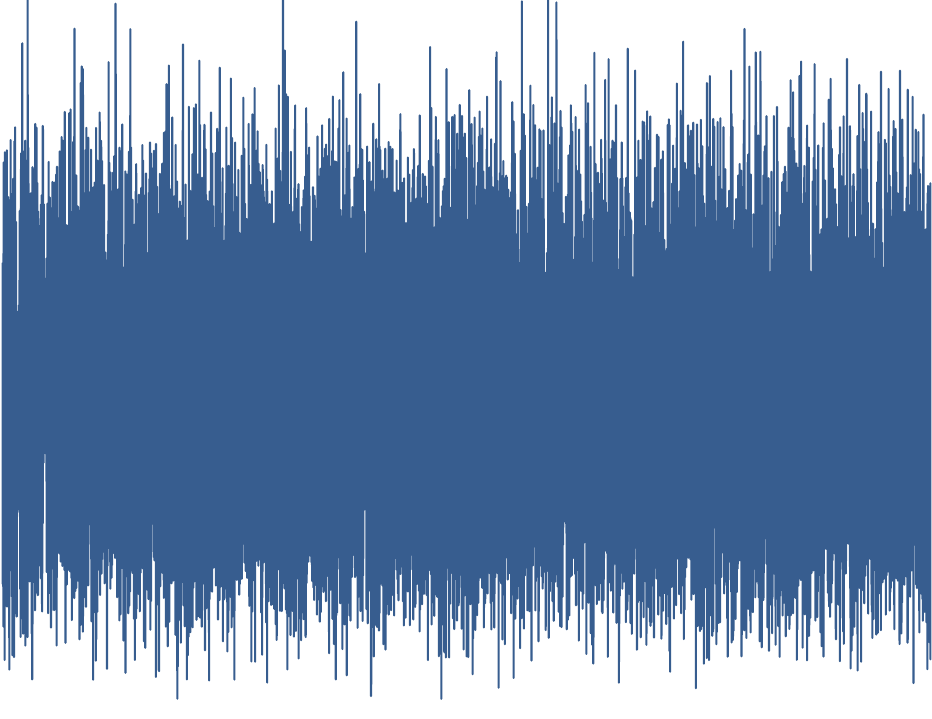};
    \nextgroupplot[title=TMDRL]
    \addplot graphics[xmin=200000,ymin=0.25,xmax=250000,ymax=2.5] {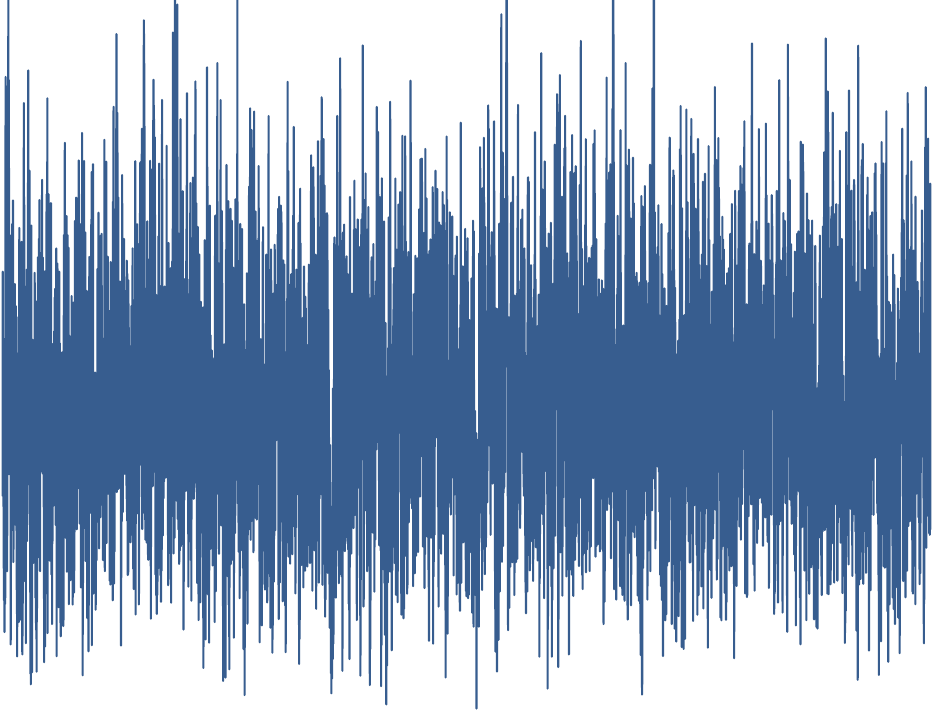};
    \nextgroupplot[title=TMRWM]
    \addplot graphics[xmin=200000,ymin=0.25,xmax=250000,ymax=2.5] {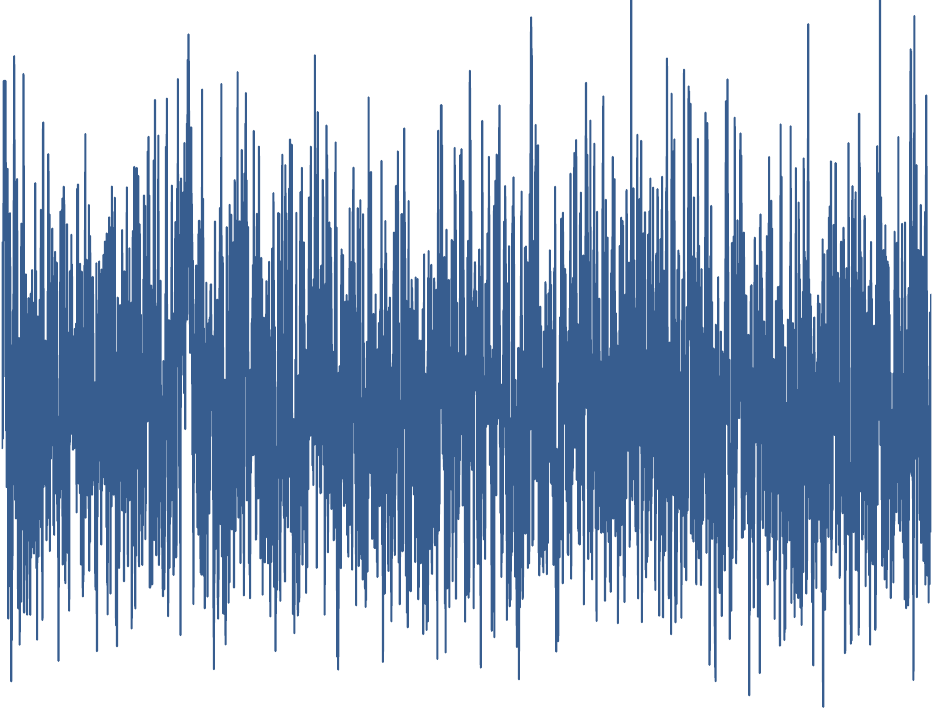};
    
\end{groupplot}
\end{tikzpicture}

\caption{Trace of MCMC chains for the parameter $P(0)$ on the predator-prey problem.  These plots show the $\num{5e4}$ s.pdf occurring just after $\num{2e5}$ burn in steps, for a realization of the long-chain cases. The map-accelerated approaches show significantly better mixing.}
\label{fig:predPreyTrace}
\end{figure}

\begin{table}
\caption[Predator-prey MCMC performance comparison]{\label{tab:predPreyPerf} Performance of MCMC samplers on the predator-prey parameter inference problem. Column headings are as described in Table~\ref{tab:bodPerf1}. The ``long'' results use a chain of $5\times10^5$ total s.pdf, while the ``short'' results use chains of length $\num{5e4}$.  The long and short chains were generated on different platforms, so the timing results should not be compared directly.  Also, because the chains are different lengths, the raw ESS values should also not be compared directly.  The relative results are normalized by DRAM-Short values for short chains and by DRAM-Long values for long chains.}
\centering
\small
\fbox{\begin{tabular}{lc|rr|rrr|D{.}{.}{4}D{.}{.}{2}}
Method & \begin{minipage}{1cm}Chain Length\end{minipage} & $\tau_{\max}$ & $\sigma_\tau$ & ESS & ESS/sec & ESS/eval &  \begin{minipage}{0.75cm}\centering Rel.\\[-0.15cm] ESS/sec\end{minipage} & \begin{minipage}{0.75cm}\centering Rel.\\[-0.15cm] ESS/eval\end{minipage} \\\hline
\multirow{ 2}{*}{DRAM} & 5e4 & 4131.5 & 2613.7 & 8 & 7.0e-06 & 1.7e-04 & 1.0 & 1.0\\\cline{2-9}
 & 5e5 &6673.8 & 5950.5 & 22 & 1.6e-05 & 2.7e-05 &  1.0 &  1.0 \\\hline
\multirow{ 2}{*}{sMMALA} &  5e4 & 1913.4 & 521.8 & 18 & 2.9e-06 & 3.7e-04 & 0.43 & 2.2\\\cline{2-9}
 & 5e5 &6365.7 & 3508.8 & 23 & 4.5e-06 & 2.2e-05  & 0.28 & 0.81\\\hline
\multirow{ 2}{*}{AMALA} &  5e4 & 1244.6 & 858.3 & 26 & 3.9e-06 & 5.4e-04   & 0.56 & 3.2\\\cline{2-9}
 & 5e5 &4323.8 & 3611.8 & 34 & 5.9e-06 & 2.6e-05 & 0.37 & 0.95\\\hline\hline
\multirow{ 2}{*}{TM+DRG} &  5e4 & 27.3 & 26.3 & 1280 & 1.4e-04 & 2.6e-02 & 20 & 150\\\cline{2-9}
& 5e5 &18.0 & 19.5 & 8344 & 9.3e-04 & 1.2e-02 & 59 & 420\\\hline
\multirow{ 2}{*}{TM+DRL} &  5e4 & 32.8 & 16.7 & 1067 & 1.2e-04 & 2.1e-02 & 17 & 130\\\cline{2-9}
& 5e5 &24.7 &  7.5 & 6081 & 6.7e-04 & 8.3e-03 & 42 & 300\\\hline
\multirow{ 2}{*}{TM+RWM} &  5e4 & 42.9 & 21.3 & 790 & 9.2e-05 & 1.6e-02 & 13 & 93\\\cline{2-9}
& 5e5 &32.7 & 15.6 & 4585 & 5.4e-04 & 1.1e-02 & 34 & 390\\\hline
\end{tabular}}

\end{table}

As in the BOD example, map-accelerated algorithms using independence proposals have dramatically shorter integrated autocorrelation times. For the longer chains, TM+DRG yields an ESS about 380 times larger than that of DRAM. Moreover, in terms of ESS per posterior evaluation, TM+DRG is 420 times more efficient than DRAM.
We also observe good agreement between the longer-chain and shorter-chain results; trends are the same in both cases.
Overall, the gradient-based methods showed relatively poor performance. sMMALA in particular suffers from nearly singular metrics. We found that tuning the step size in sMMALA was difficult.  On the other hand, the derivative-free methods were easier to tune and had much better performance.  Even when normalized by run time, the ESS/sec of TM+DRG is still more than one order of magnitude larger than that of DRAM. While posterior evaluations in this example are not trivially cheap, the ESS/evaluation represents the limiting behavior of the algorithm as evaluations become the dominant cost of an MCMC step; here we see improvements of at least two orders of magnitude over the baseline schemes.  
 
Results for the longer chains are generally more favorable for the transport map approaches.  We believe this is caused by two factors: first, the burn-in is smaller in relative terms for the longer chains, which reduces wasted computational effort; second, the adaptive proposals have more time to accurately characterize the posterior. In longer trace plots, we observe that the adaptation is negligible after approximately $\num{1e5}$ s.pdf, which suggests that the different burn-in lengths dominate the difference between long and short chains. 



\subsection{Maple sap exudation}\label{sec:perf:maple}
This section presents an inference problem based on the system of differential-algebraic equations introduced in \cite{Ceseri2013} to describe microscale sap dynamics in a maple tree during spring freeze-thaw cycles. The posterior in this 10-dimensional problem is particularly challenging to explore, and helps illustrate aspects of the map adaptation process. The nonlinear forward model has three state variables describing the positions of gas, liquid, and ice interfaces ($s_{gi}(t)$, $s_{iw}(t)$, and $r(t)$) as well as a state variable $U(t)$ representing the volume of melted ice.  These variables are related via the following differential-algebraic equations:
\begin{align}
2\rho_is_{gi}(t)\dot{s}_{gi}(t) &= \frac{\rho_w}{\pi L^f}\dot{U}(t) - 2(\rho_w-\rho_i)s_{iw}(t)\dot{s}_{iw}(t) \label{eq:MapleDae1}\\
\lambda \rho_w\dot{s}_{iw}(t) &= -\kappa(x) \partial_x T(x,t) \quad \text{at} \quad x=s_{iw}(t) \label{eq:MapleDae2}\\
N\dot{U}(t) &= - \frac{KA}{\rho_w g W}\left[p_w^v(t)-p_g^f(t) - RT(R^f,t)c_s^v\right] \label{eq:MapleDae3}\\
r(t)\dot{r}(t) &= - \frac{N\dot{U}(t)}{2\pi L^v}.\label{eq:MapleDae4}
\end{align} 
In addition to the state equations, the model is closed with five algebraic relations:
\begin{align}
p_g^f(t) &=  p_g^f(0)\left(\frac{s_{gi}(0)}{s_{gi}(t)}\right)^2 & & \label{eq:MapleAlg1}\\
p_w^v(t) &= p_g^v(x,t) + \frac{\sigma}{r(t)} &  \text{at } & x=R^f+R^v-r \label{eq:MapleAlg2} \\
p_g^v(t) &= \frac{\rho_g^v(x,t) R T_g^v(x,t)}{M_g}& \text{at } & x=R^f+R^v-r \label{eq:MapleAlg3}\\
 c_g^v(t) &= \frac{H}{M_g} \rho_g^v(x,t) & \text{at } & x=R^f+R^v-r \label{eq:MapleAlg4} \\
 \rho_g^v(x,t) &= \frac{\rho_g^v(x,0)V_g^v(0) - \left.M_g c_g^v(\tilde{t})\left(V^v-V_g^v(\tilde{t})\right)\right|_{\tilde{t}=0}^t}{V_g^v(t)} & \text{at } & x=R^f+R^v-r.\label{eq:MapleAlg5}
\end{align}
In this system, $T(x,t)$ is a transient temperature field, $[\rho_i, \rho_w, \lambda, R, g, \sigma, H, M_g]$ are physical constants, and the parameters $[V^v, V_g^v, N,K,A,W,L^f,L^v,c_s^v]$ are inference targets.  The initial conditions $s_{gi}(0)$, $s_{iw}(0)$, and $r(0)$ are also inference targets.  For additional details on the model and its parameters, see Appendix \ref{sec:maplemodel}.  

We describe the model parameters with a random variable $\trv$ taking values in $\real^{10}$.  As detailed in Appendix \ref{sec:maplemodel}, the components of $\trv$ are scaled and combined to obtain the model parameters and initial conditions.  We choose $\trv$ so that each component of the prior $\td(\trv)$ is independent.  In particular, the prior is given by $\trv_{1:3} \sim U[-1,1]^3$ and $\trv_{4:10} \sim N(0,I)$.  Noisy observations of $p_w^v(t)$ at 100 times equally spaced over $t\in[0,1209600]$ are combined with an independent additive Gaussian error model to define the likelihood function $\td(\trv|d)$.  The additive errors are identically distributed with zero mean and a standard deviation of $1000$ pascals.  

The posterior distribution is illustrated in Figure \ref{fig:maplePost} and the performance of several algorithms is summarized in Table \ref{tab:maplePerf}.  We restricted this study to derivative-free MCMC samplers, due to the complexity of computing derivative information with the maple forward model.  To obtain our performance results, we ran MCMC chains of $\num{2e5}$ s.pdf each, discarding the first $\num{1e5}$ samples as burn-in. As before, the ESS values reported in Table \ref{tab:maplePerf} represent a minimum over all ten components of each chain, calculated after burn-in. Yet the number of evaluations and the run times reported in the table reflect the cost of {all} $\num{2e5}$ steps, \textit{including} burn in. Hence these are conservative numbers that include the computational effort required for adaptation. 50 repetitions of each sampler were used to obtain these performance evaluations. We again used cubic total-degree polynomial maps.

\begin{figure}[h!]
\centering

\tabskip=0pt
\valign{#\cr
  \hbox{%
    \begin{subfigure}[b]{0.7\textwidth}
    \centering

\newcommand{\groupPlotWidth}{0.3000\textwidth}
\begin{tikzpicture}[scale=0.5]
\begin{groupplot}[group style={every plot/.style={enlargelimits=true},group size=10 by 10,xlabels at=edge bottom,ylabels at=edge left, xticklabels at=edge bottom,yticklabels at=edge left,vertical sep=1pt,horizontal sep=1pt},height=\groupPlotWidth, width=\groupPlotWidth,ticks=none,enlargelimits=false]

\nextgroupplot[title={\huge$\theta_{1}$}]
\addplot graphics[xmin=-0.825,ymin=0.584,xmax=0.000,ymax=1.451] {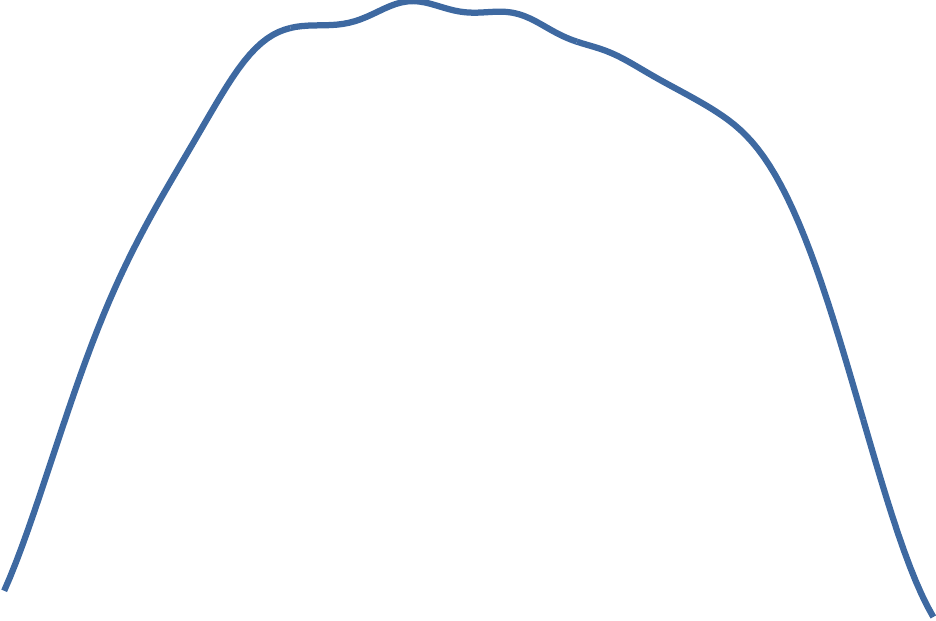};

\nextgroupplot[group/empty plot]
\nextgroupplot[group/empty plot]
\nextgroupplot[group/empty plot]
\nextgroupplot[group/empty plot]
\nextgroupplot[group/empty plot]
\nextgroupplot[group/empty plot]
\nextgroupplot[group/empty plot]
\nextgroupplot[group/empty plot]
\nextgroupplot[group/empty plot]

\nextgroupplot
\addplot graphics[xmin=-0.825,ymin=0.584,xmax=-0.995,ymax=0.997] {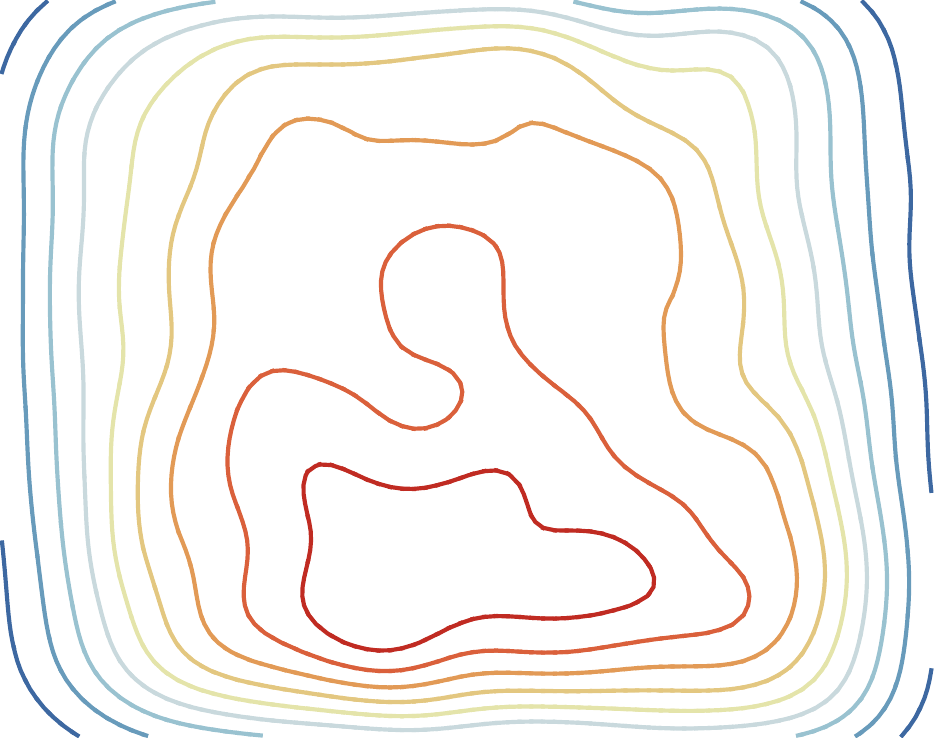};
\nextgroupplot[title={\huge$\theta_{2}$}]
\addplot graphics[xmin=-0.995,ymin=0.997,xmax=0.000,ymax=0.920] {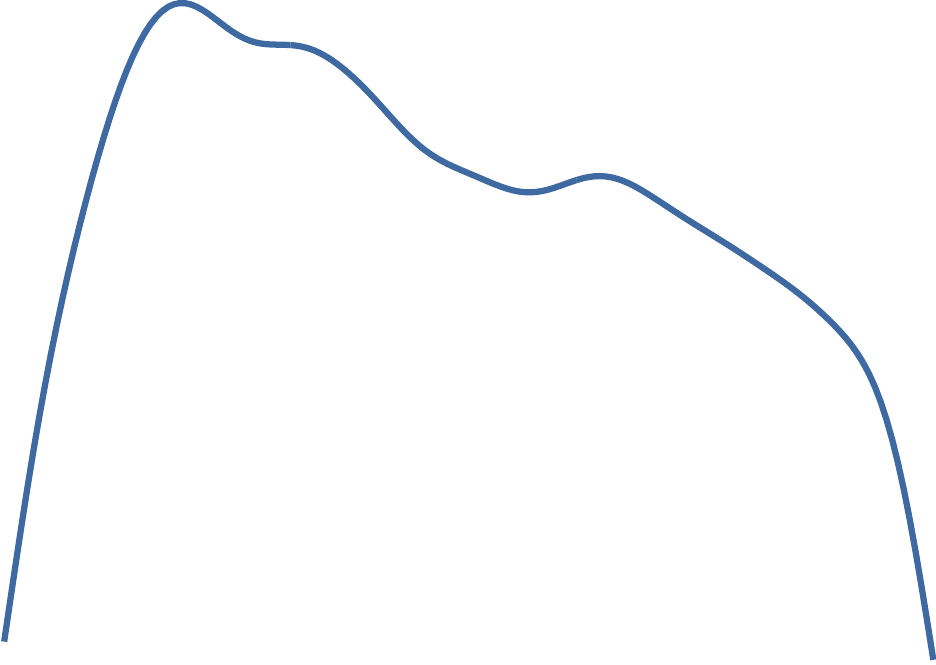};

\nextgroupplot[group/empty plot]
\nextgroupplot[group/empty plot]
\nextgroupplot[group/empty plot]
\nextgroupplot[group/empty plot]
\nextgroupplot[group/empty plot]
\nextgroupplot[group/empty plot]
\nextgroupplot[group/empty plot]
\nextgroupplot[group/empty plot]

\nextgroupplot
\addplot graphics[xmin=-0.825,ymin=0.584,xmax=-0.940,ymax=0.996] {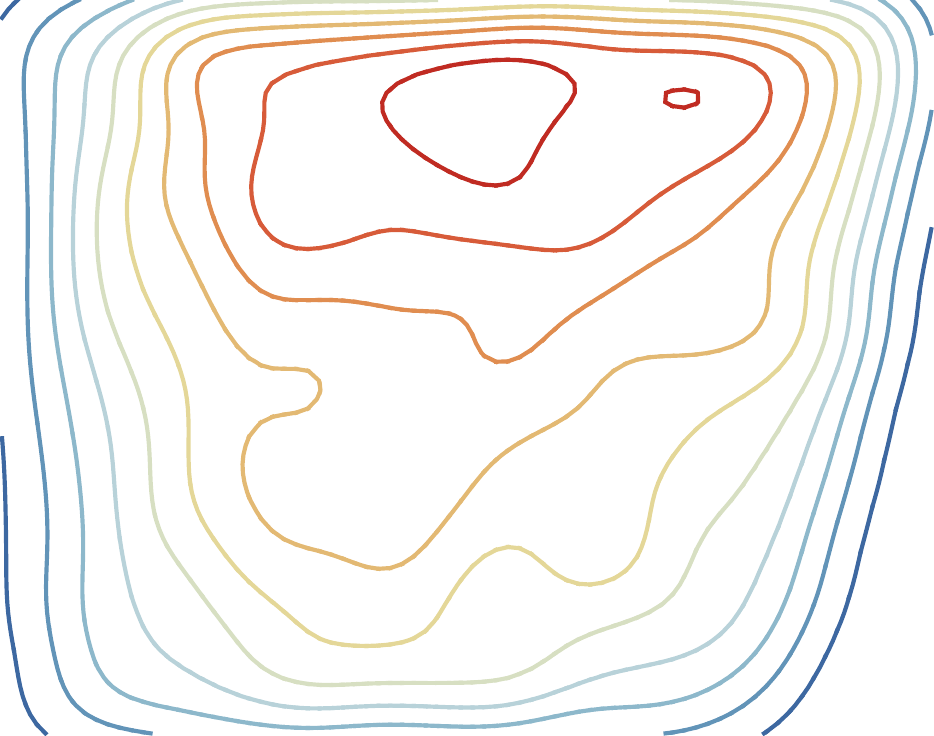};
\nextgroupplot
\addplot graphics[xmin=-0.995,ymin=0.997,xmax=-0.940,ymax=0.996] {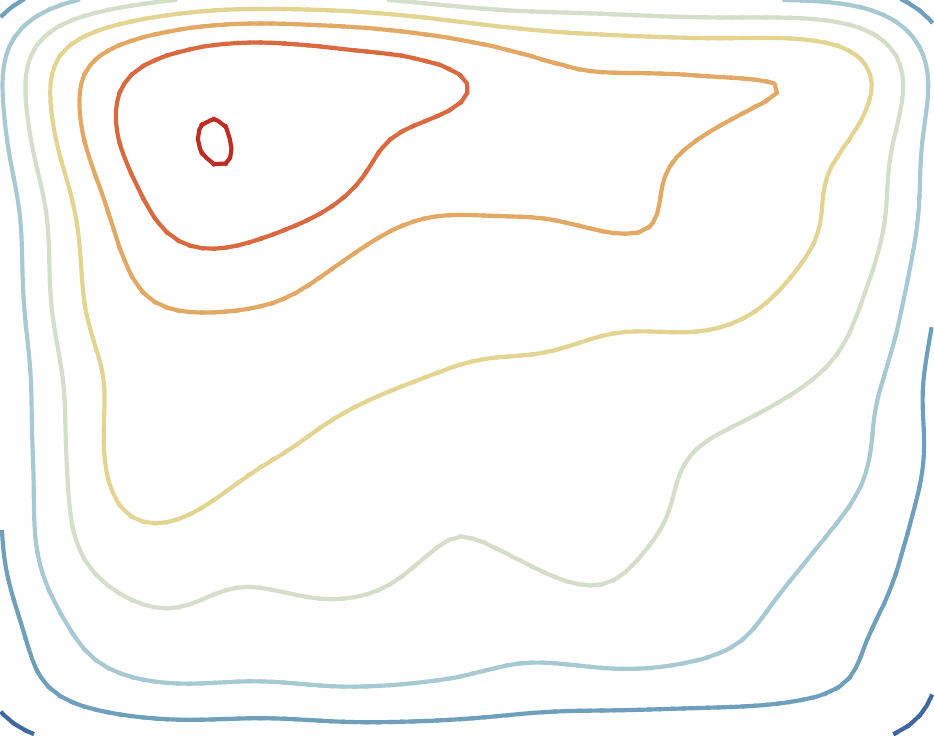};
\nextgroupplot[title={\huge$\theta_{3}$}]
\addplot graphics[xmin=-0.940,ymin=0.996,xmax=0.000,ymax=1.077] {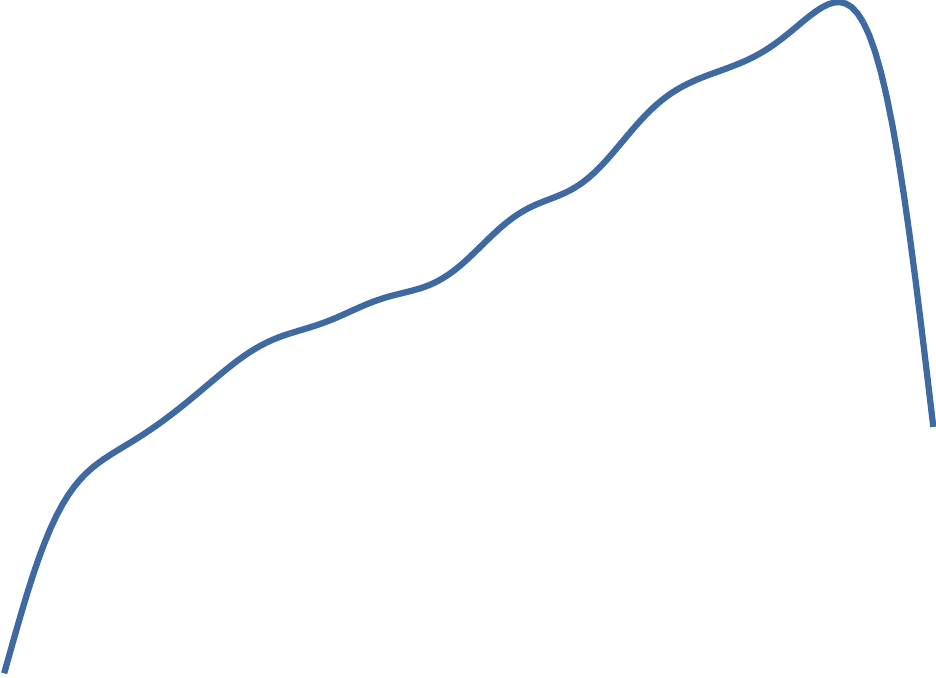};

\nextgroupplot[group/empty plot]
\nextgroupplot[group/empty plot]
\nextgroupplot[group/empty plot]
\nextgroupplot[group/empty plot]
\nextgroupplot[group/empty plot]
\nextgroupplot[group/empty plot]
\nextgroupplot[group/empty plot]

\nextgroupplot
\addplot graphics[xmin=-0.825,ymin=0.584,xmax=-1.089,ymax=1.165] {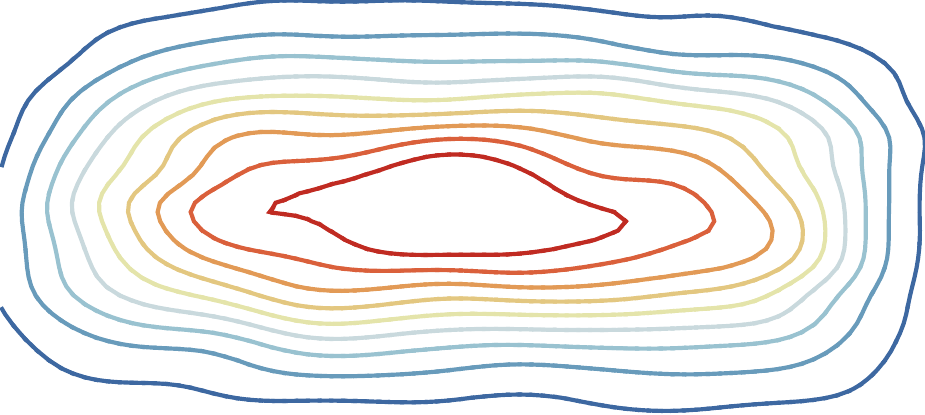};
\nextgroupplot
\addplot graphics[xmin=-0.995,ymin=0.997,xmax=-1.089,ymax=1.165] {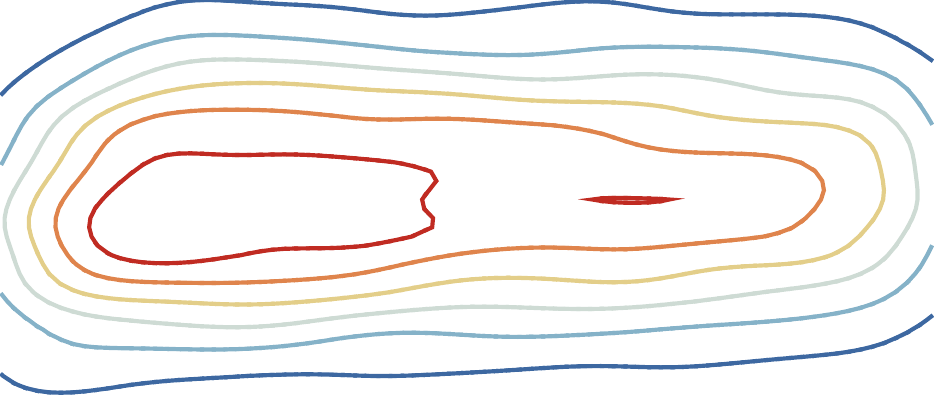};
\nextgroupplot
\addplot graphics[xmin=-0.940,ymin=0.996,xmax=-1.089,ymax=1.165] {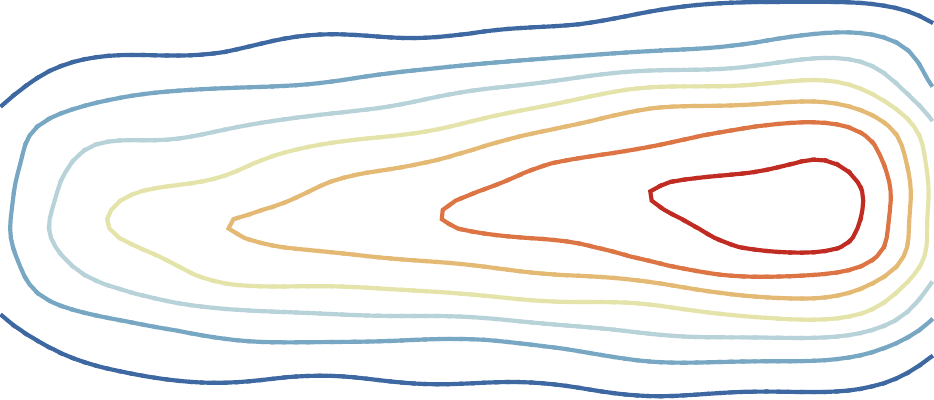};
\nextgroupplot[title={\huge$\theta_{4}$}]
\addplot graphics[xmin=-1.089,ymin=1.165,xmax=0.000,ymax=1.132] {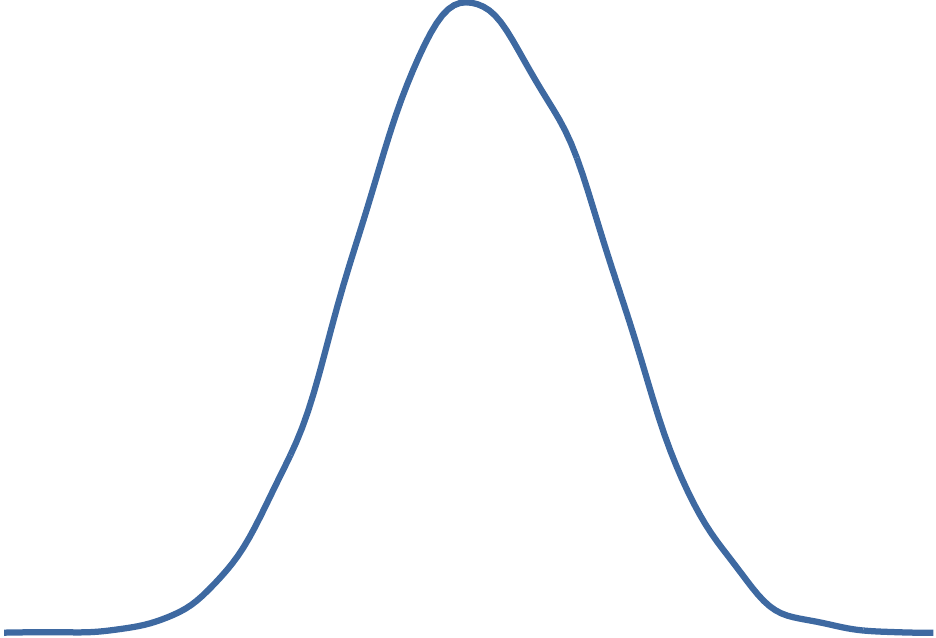};

\nextgroupplot[group/empty plot]
\nextgroupplot[group/empty plot]
\nextgroupplot[group/empty plot]
\nextgroupplot[group/empty plot]
\nextgroupplot[group/empty plot]
\nextgroupplot[group/empty plot]

\nextgroupplot
\addplot graphics[xmin=-0.825,ymin=0.584,xmax=-1.972,ymax=1.994] {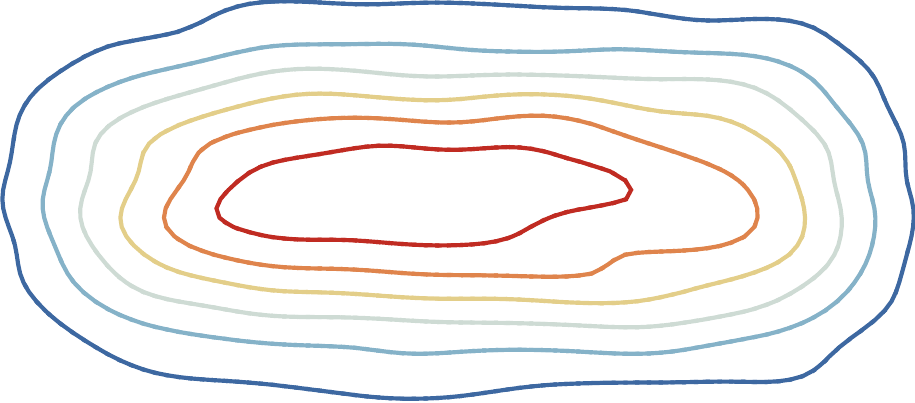};
\nextgroupplot
\addplot graphics[xmin=-0.995,ymin=0.997,xmax=-1.972,ymax=1.994] {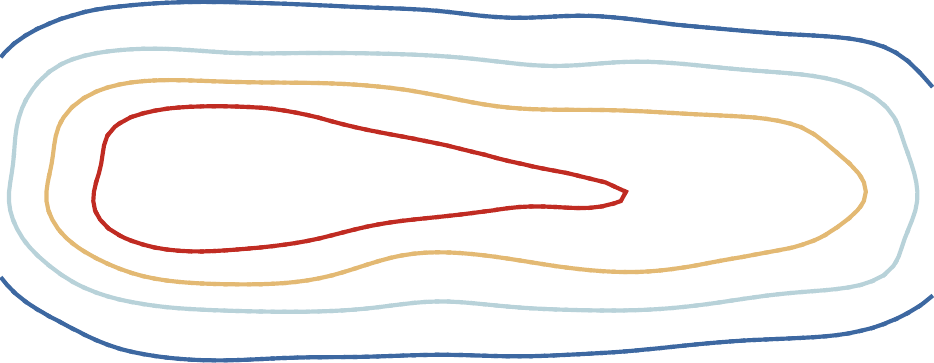};
\nextgroupplot
\addplot graphics[xmin=-0.940,ymin=0.996,xmax=-1.972,ymax=1.994] {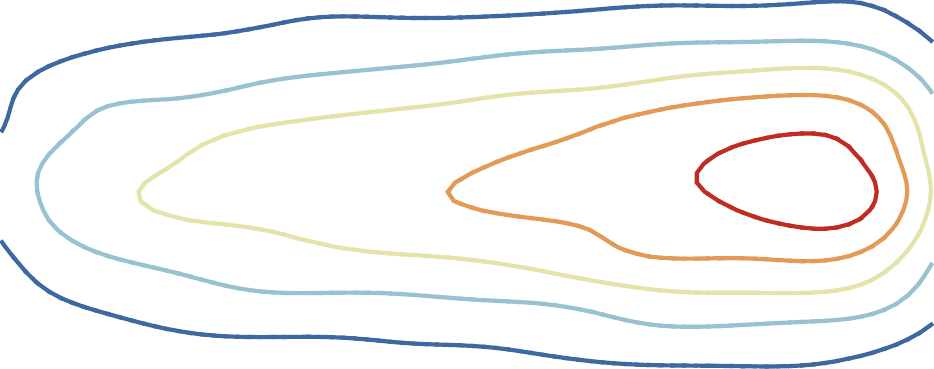};
\nextgroupplot
\addplot graphics[xmin=-1.089,ymin=1.165,xmax=-1.972,ymax=1.994] {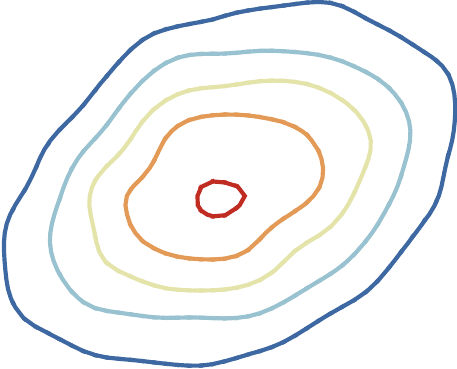};
\nextgroupplot[title={\huge$\theta_{5}$}]
\addplot graphics[xmin=-1.972,ymin=1.994,xmax=0.000,ymax=0.834] {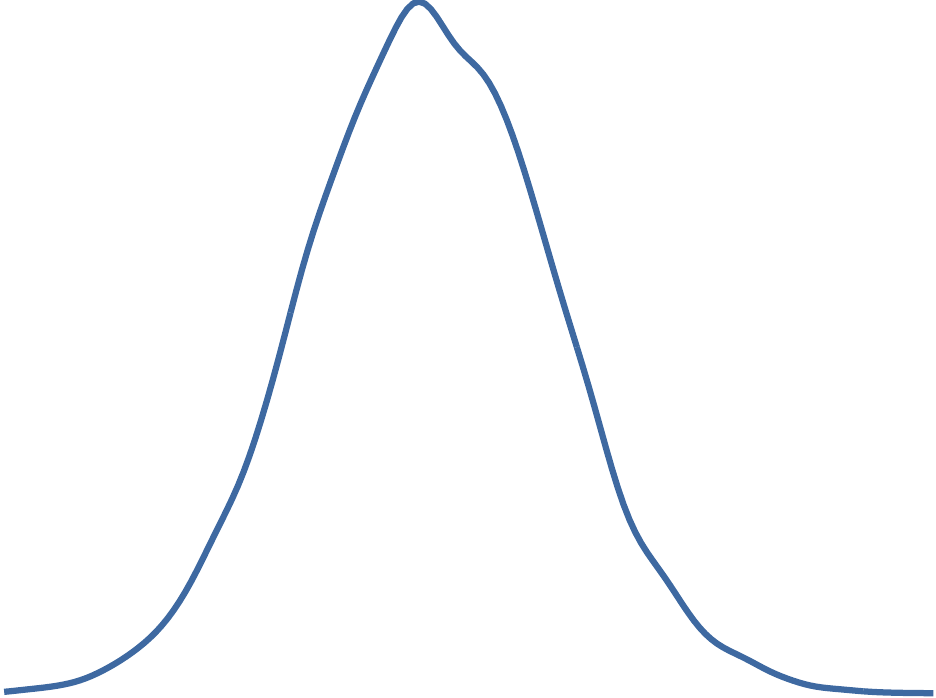};

\nextgroupplot[group/empty plot]
\nextgroupplot[group/empty plot]
\nextgroupplot[group/empty plot]
\nextgroupplot[group/empty plot]
\nextgroupplot[group/empty plot]

\nextgroupplot
\addplot graphics[xmin=-0.825,ymin=0.584,xmax=-0.585,ymax=0.378] {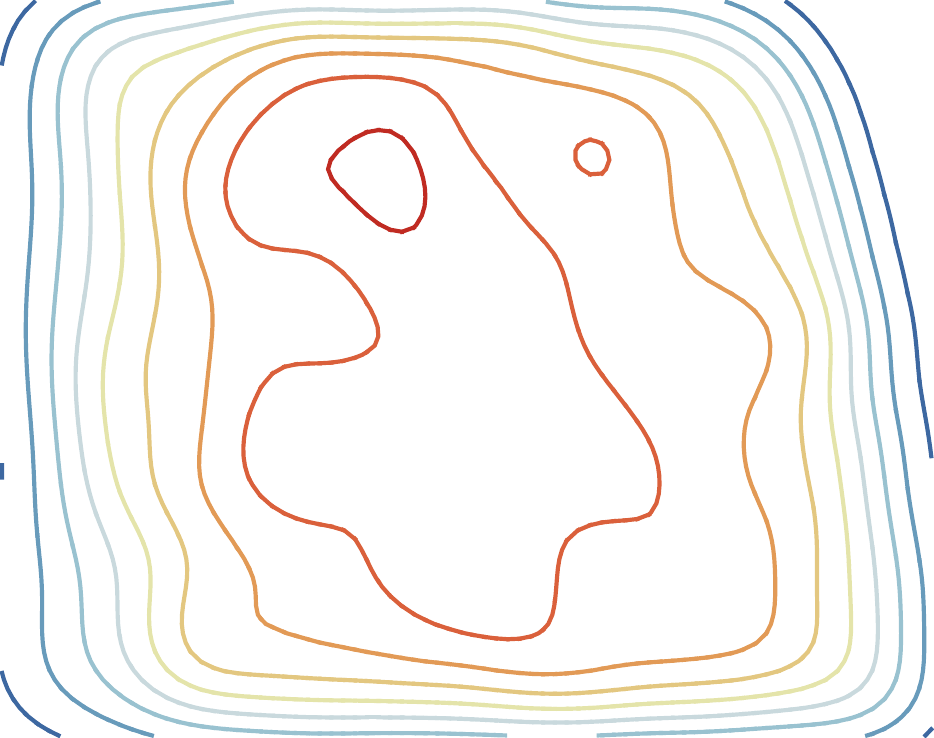};
\nextgroupplot
\addplot graphics[xmin=-0.995,ymin=0.997,xmax=-0.585,ymax=0.378] {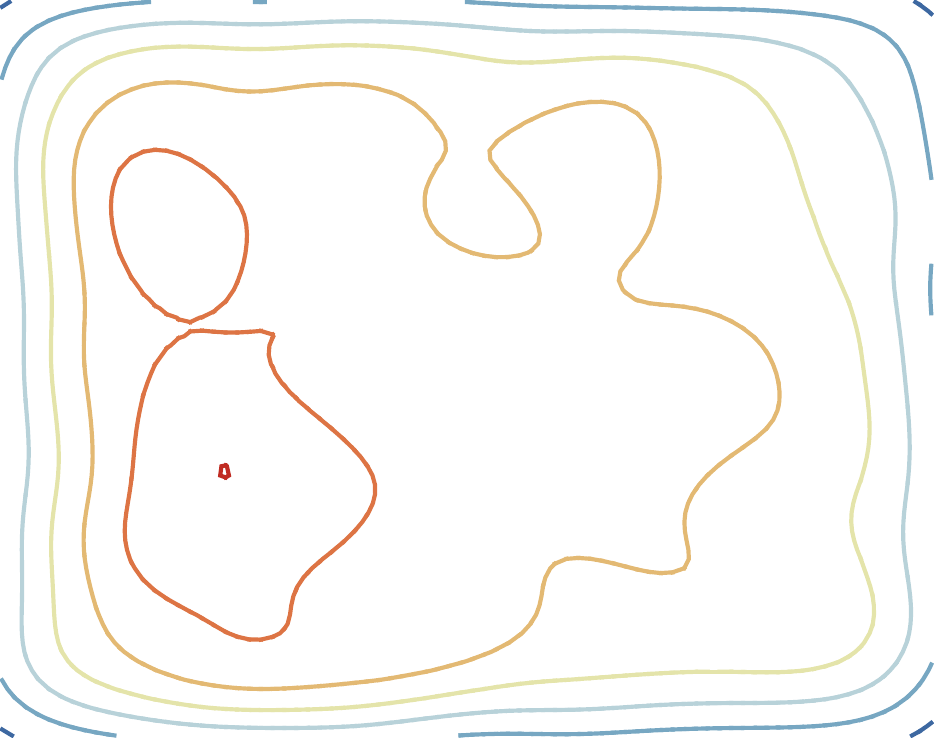};
\nextgroupplot
\addplot graphics[xmin=-0.940,ymin=0.996,xmax=-0.585,ymax=0.378] {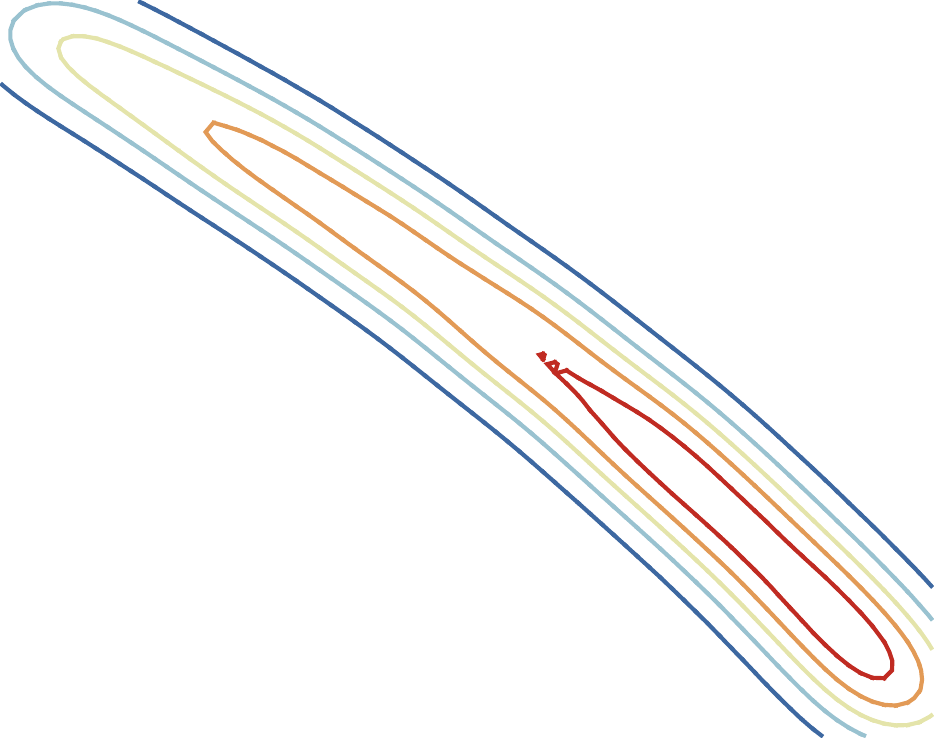};
\nextgroupplot
\addplot graphics[xmin=-1.089,ymin=1.165,xmax=-0.585,ymax=0.378] {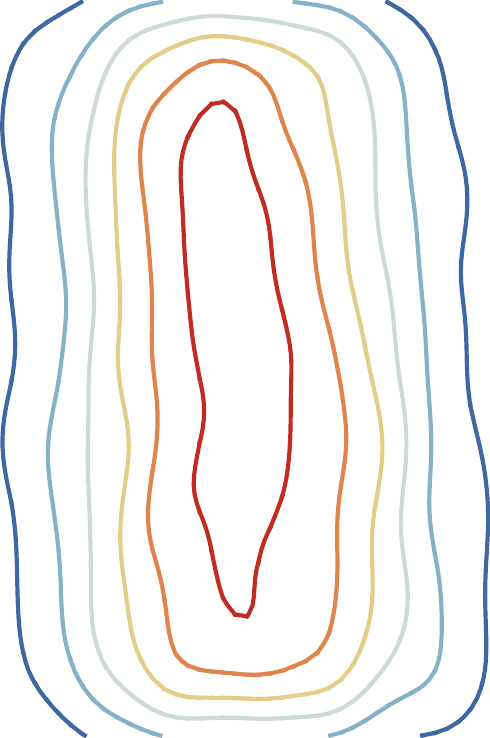};
\nextgroupplot
\addplot graphics[xmin=-1.972,ymin=1.994,xmax=-0.585,ymax=0.378] {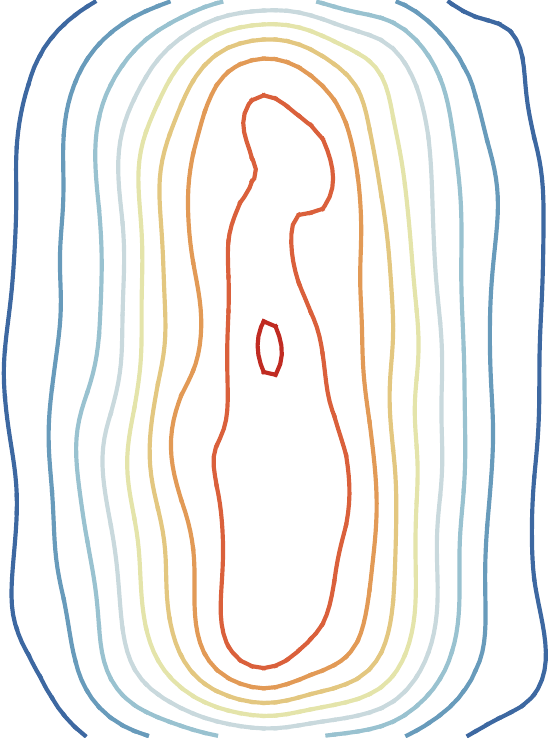};
\nextgroupplot[title={\huge$\theta_{6}$}]
\addplot graphics[xmin=-0.585,ymin=0.378,xmax=0.000,ymax=2.168] {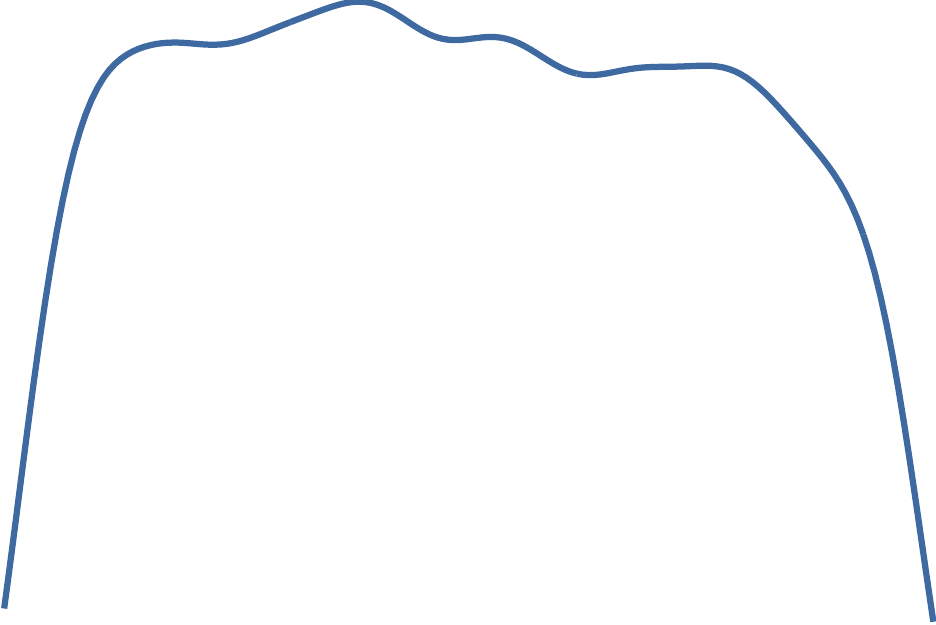};

\nextgroupplot[group/empty plot]
\nextgroupplot[group/empty plot]
\nextgroupplot[group/empty plot]
\nextgroupplot[group/empty plot]

\nextgroupplot
\addplot graphics[xmin=-0.825,ymin=0.584,xmax=-0.421,ymax=1.011] {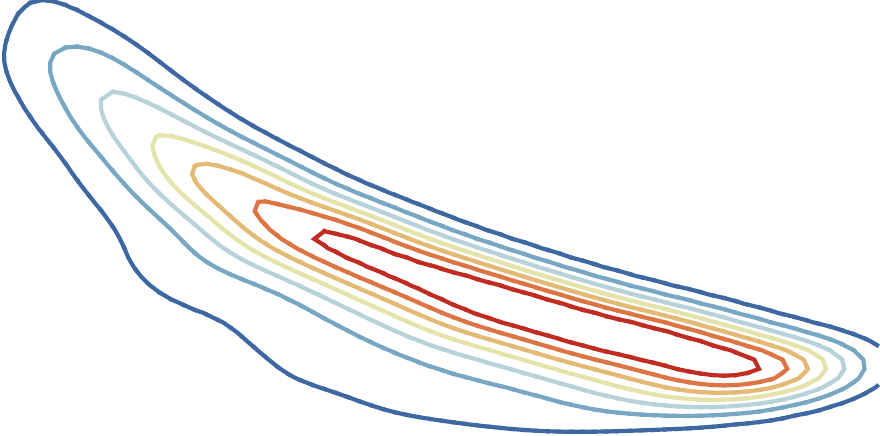};
\nextgroupplot
\addplot graphics[xmin=-0.995,ymin=0.997,xmax=-0.421,ymax=1.011] {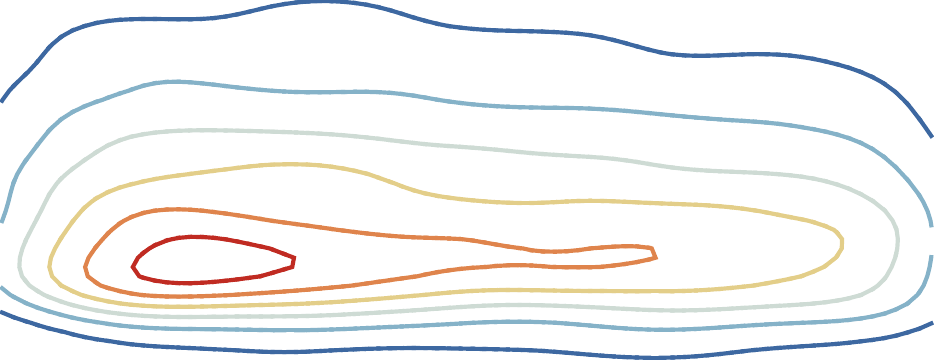};
\nextgroupplot
\addplot graphics[xmin=-0.940,ymin=0.996,xmax=-0.421,ymax=1.011] {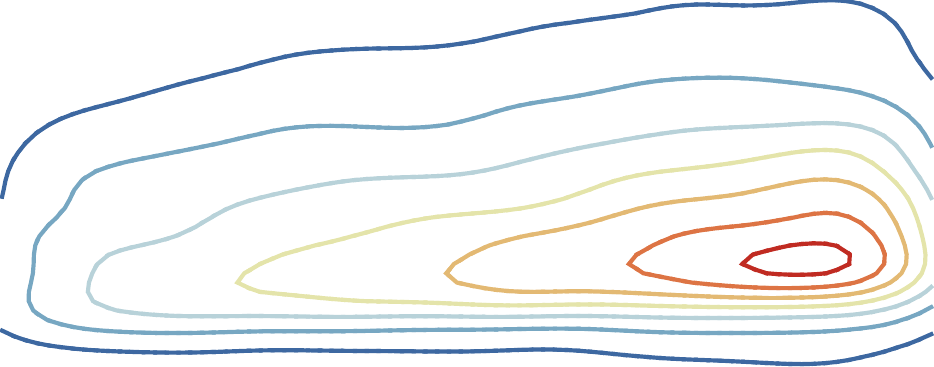};
\nextgroupplot
\addplot graphics[xmin=-1.089,ymin=1.165,xmax=-0.421,ymax=1.011] {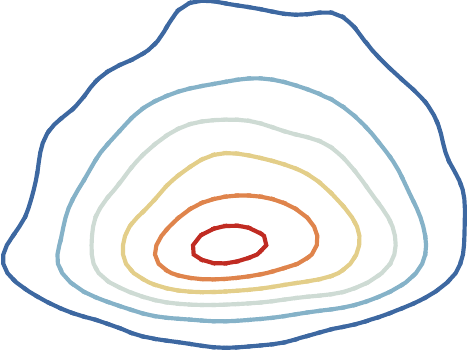};
\nextgroupplot
\addplot graphics[xmin=-1.972,ymin=1.994,xmax=-0.421,ymax=1.011] {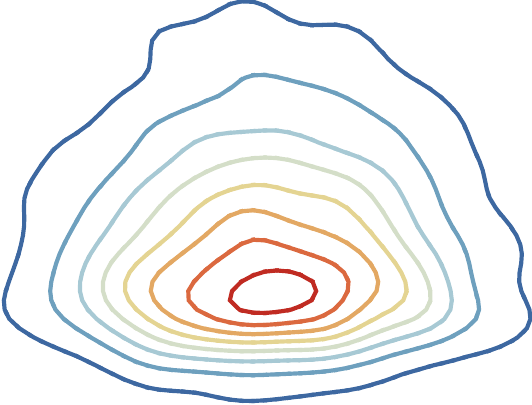};
\nextgroupplot
\addplot graphics[xmin=-0.585,ymin=0.378,xmax=-0.421,ymax=1.011] {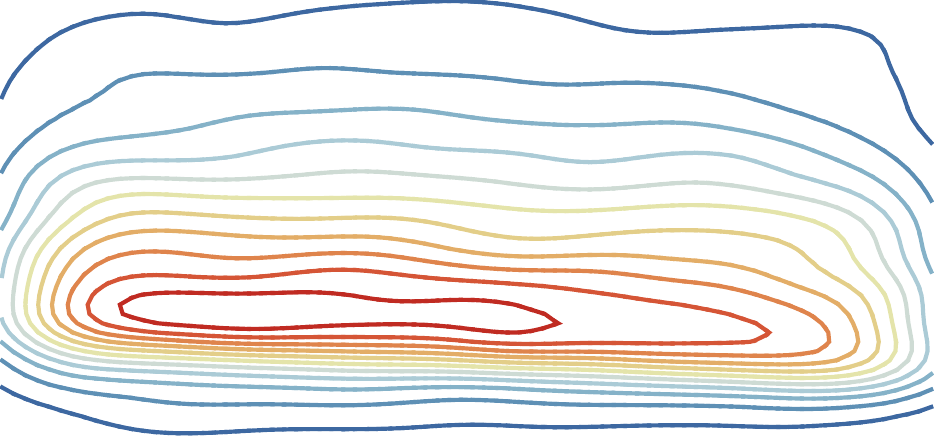};
\nextgroupplot[title={\huge$\theta_{7}$}]
\addplot graphics[xmin=-0.421,ymin=1.011,xmax=0.000,ymax=1.730] {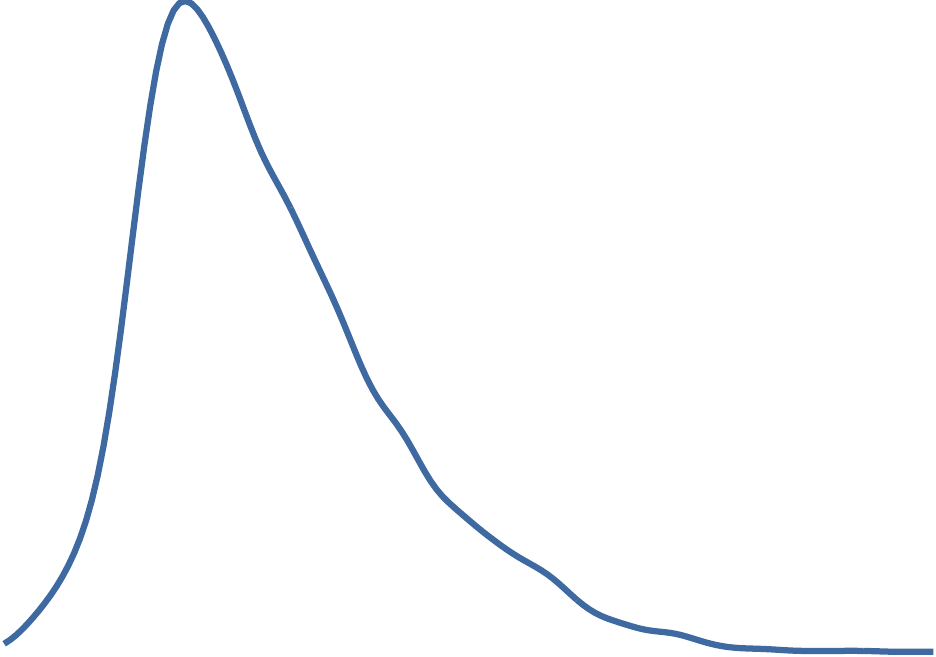};

\nextgroupplot[group/empty plot]
\nextgroupplot[group/empty plot]
\nextgroupplot[group/empty plot]

\nextgroupplot
\addplot graphics[xmin=-0.825,ymin=0.584,xmax=-2.078,ymax=1.456] {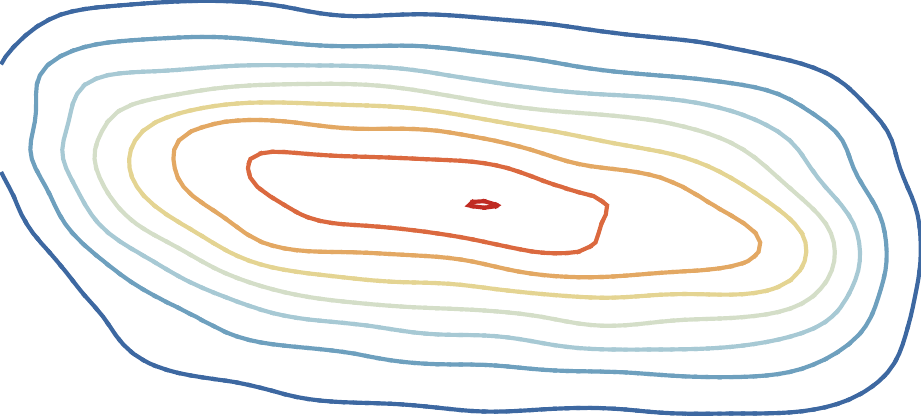};
\nextgroupplot
\addplot graphics[xmin=-0.995,ymin=0.997,xmax=-2.078,ymax=1.456] {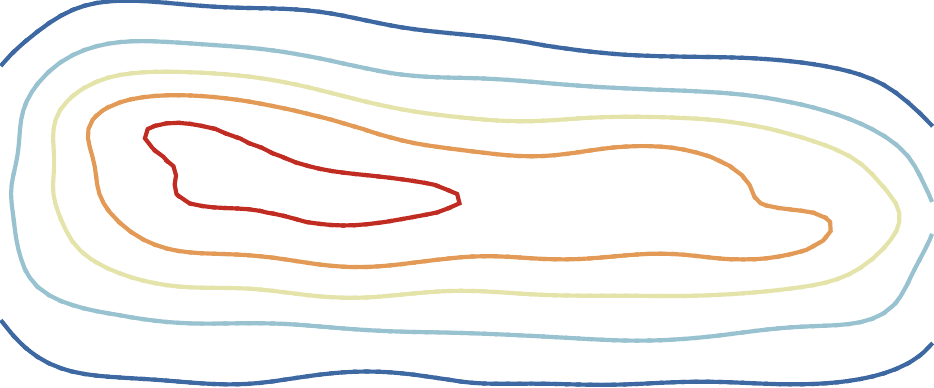};
\nextgroupplot
\addplot graphics[xmin=-0.940,ymin=0.996,xmax=-2.078,ymax=1.456] {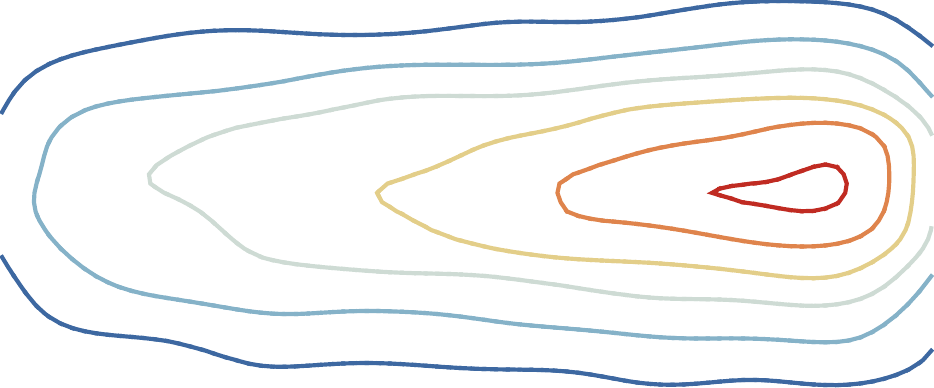};
\nextgroupplot
\addplot graphics[xmin=-1.089,ymin=1.165,xmax=-2.078,ymax=1.456] {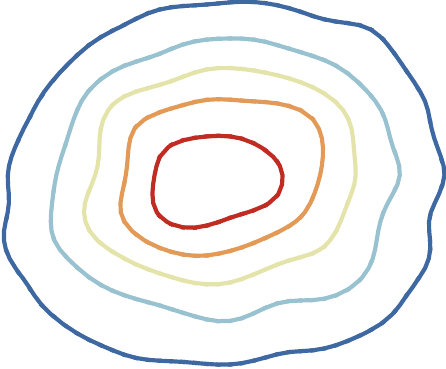};
\nextgroupplot
\addplot graphics[xmin=-1.972,ymin=1.994,xmax=-2.078,ymax=1.456] {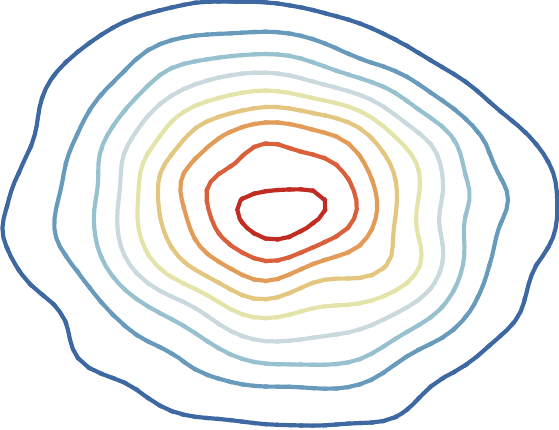};
\nextgroupplot
\addplot graphics[xmin=-0.585,ymin=0.378,xmax=-2.078,ymax=1.456] {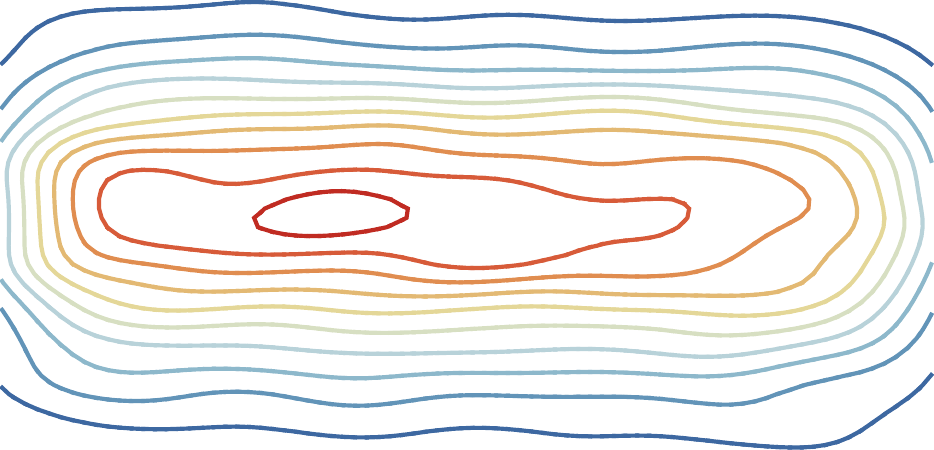};
\nextgroupplot
\addplot graphics[xmin=-0.421,ymin=1.011,xmax=-2.078,ymax=1.456] {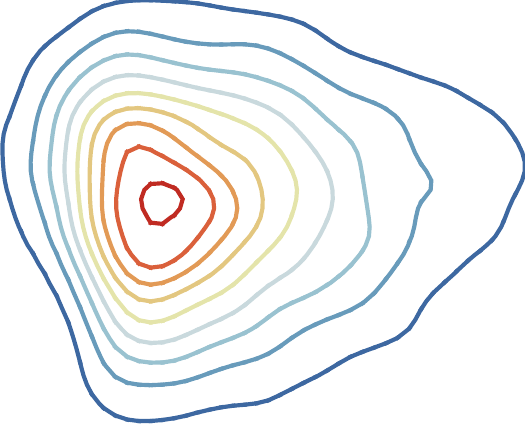};
\nextgroupplot[title={\huge$\theta_{8}$}]
\addplot graphics[xmin=-2.078,ymin=1.456,xmax=0.000,ymax=0.920] {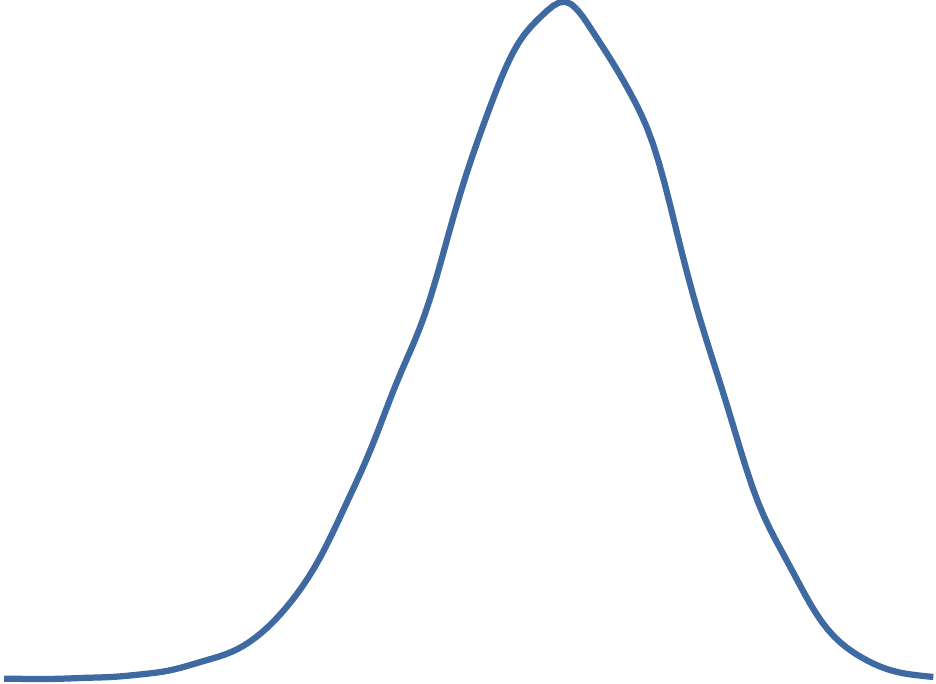};

\nextgroupplot[group/empty plot]
\nextgroupplot[group/empty plot]

\nextgroupplot
\addplot graphics[xmin=-0.825,ymin=0.584,xmax=-2.830,ymax=2.885] {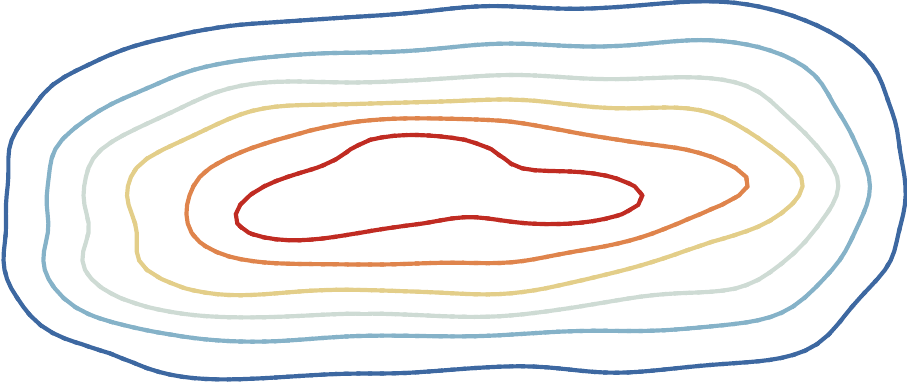};
\nextgroupplot
\addplot graphics[xmin=-0.995,ymin=0.997,xmax=-2.830,ymax=2.885] {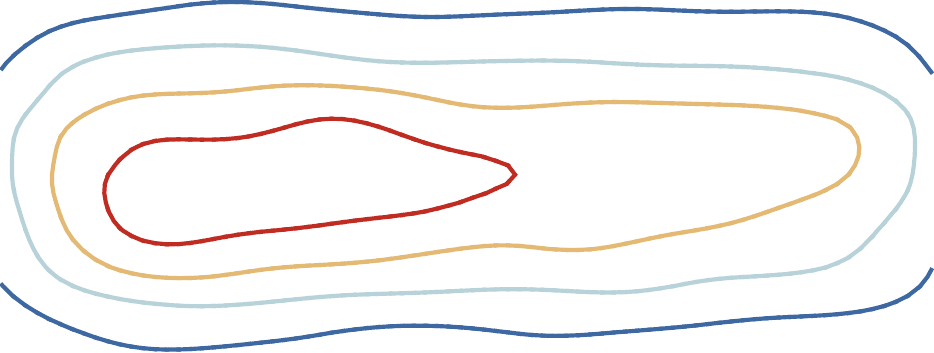};
\nextgroupplot
\addplot graphics[xmin=-0.940,ymin=0.996,xmax=-2.830,ymax=2.885] {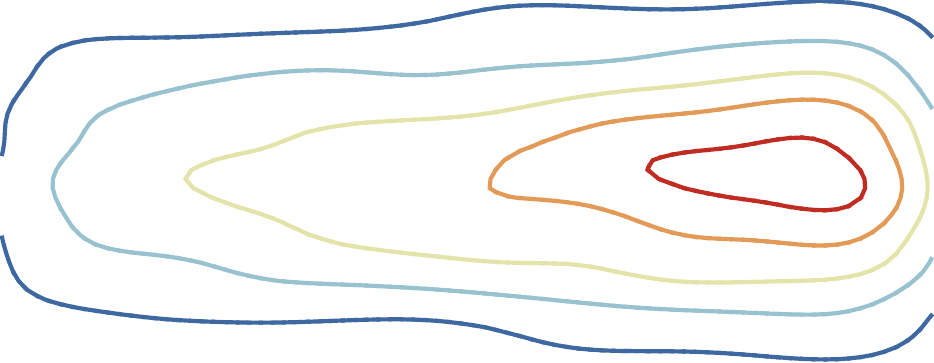};
\nextgroupplot
\addplot graphics[xmin=-1.089,ymin=1.165,xmax=-2.830,ymax=2.885] {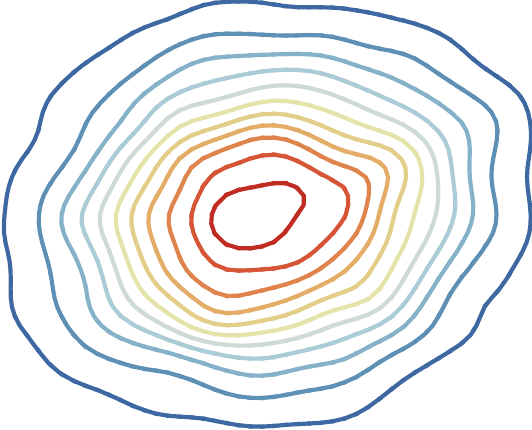};
\nextgroupplot
\addplot graphics[xmin=-1.972,ymin=1.994,xmax=-2.830,ymax=2.885] {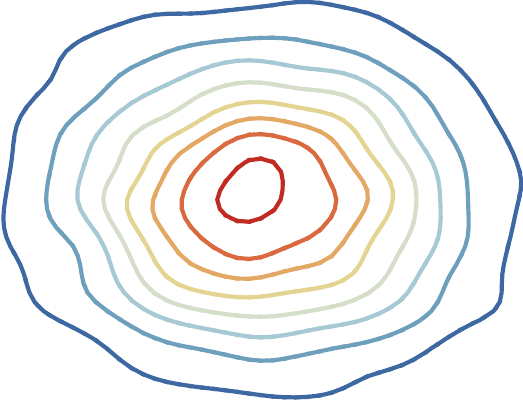};
\nextgroupplot
\addplot graphics[xmin=-0.585,ymin=0.378,xmax=-2.830,ymax=2.885] {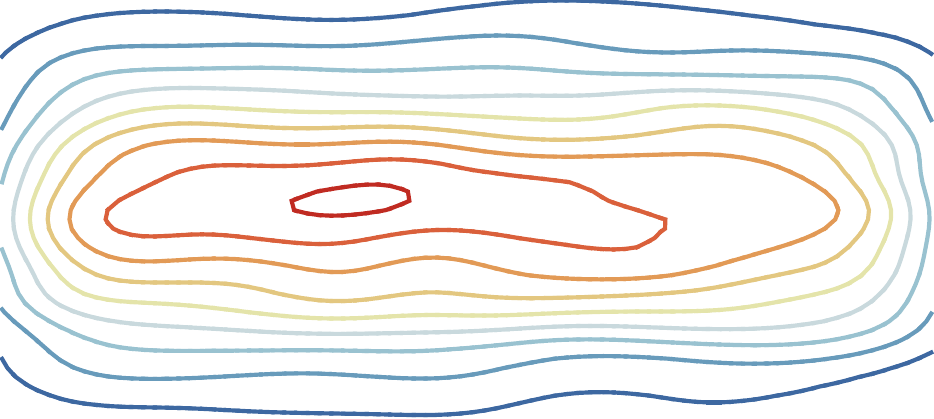};
\nextgroupplot
\addplot graphics[xmin=-0.421,ymin=1.011,xmax=-2.830,ymax=2.885] {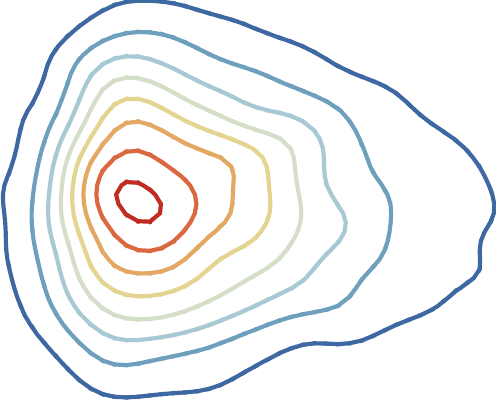};
\nextgroupplot
\addplot graphics[xmin=-2.078,ymin=1.456,xmax=-2.830,ymax=2.885] {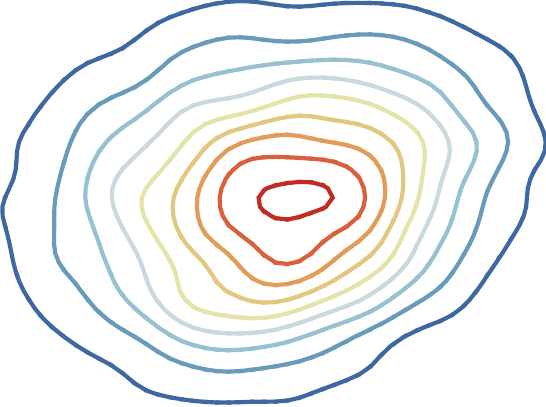};
\nextgroupplot[title={\huge$\theta_{9}$}]
\addplot graphics[xmin=-2.830,ymin=2.885,xmax=0.000,ymax=0.579] {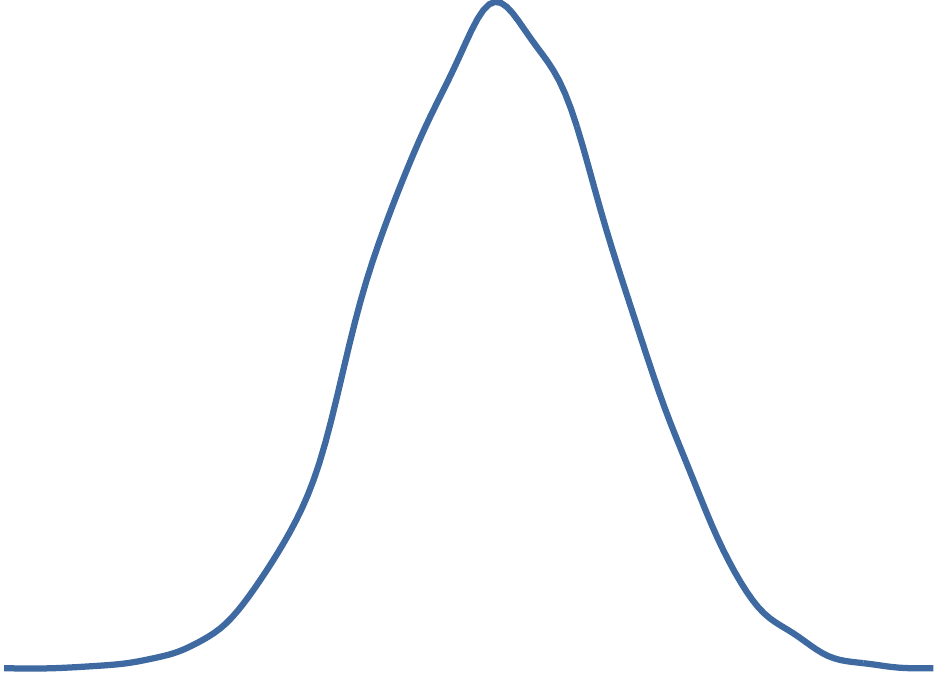};

\nextgroupplot[group/empty plot]

\nextgroupplot
\addplot graphics[xmin=-0.825,ymin=0.584,xmax=-3.191,ymax=2.990] {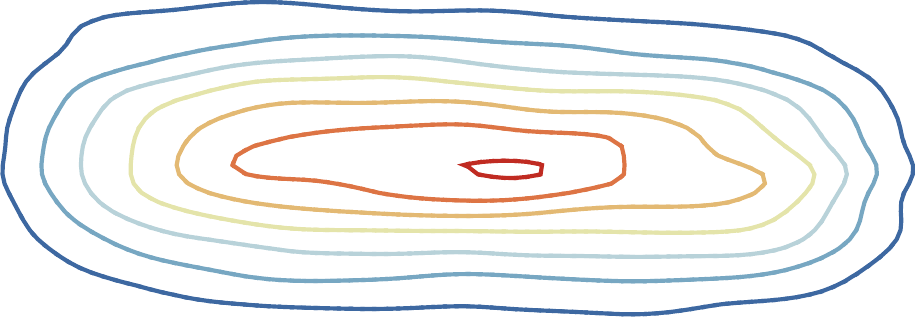};
\nextgroupplot
\addplot graphics[xmin=-0.995,ymin=0.997,xmax=-3.191,ymax=2.990] {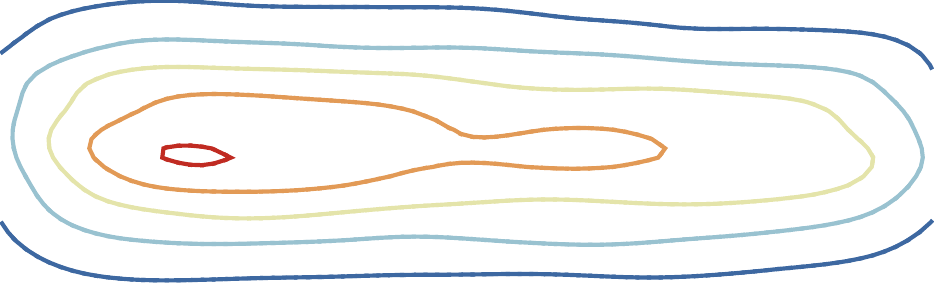};
\nextgroupplot
\addplot graphics[xmin=-0.940,ymin=0.996,xmax=-3.191,ymax=2.990] {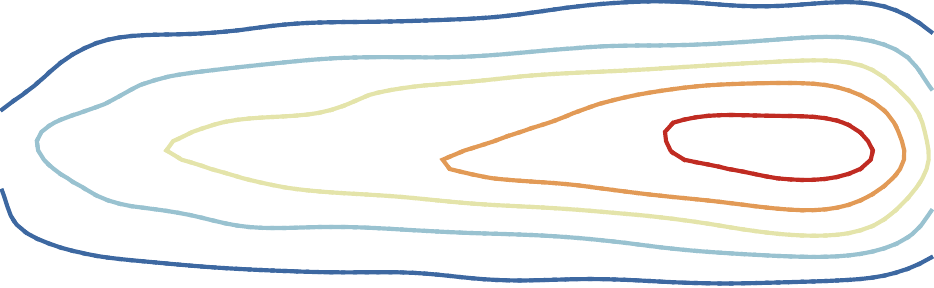};
\nextgroupplot
\addplot graphics[xmin=-1.089,ymin=1.165,xmax=-3.191,ymax=2.990] {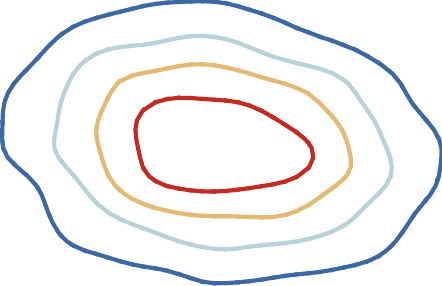};
\nextgroupplot
\addplot graphics[xmin=-1.972,ymin=1.994,xmax=-3.191,ymax=2.990] {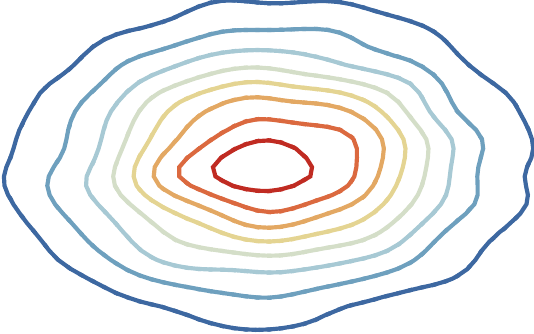};
\nextgroupplot
\addplot graphics[xmin=-0.585,ymin=0.378,xmax=-3.191,ymax=2.990] {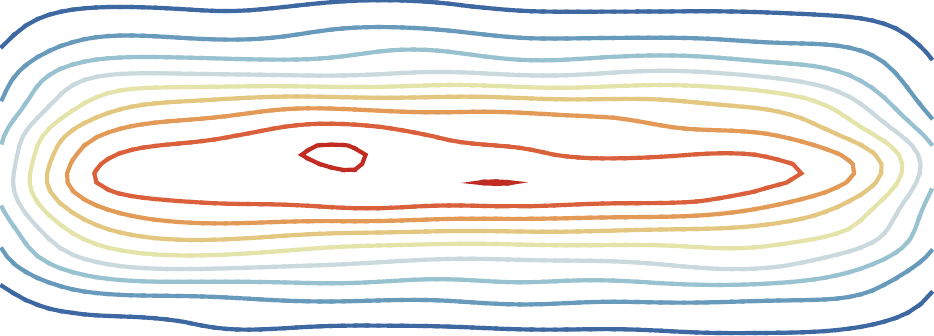};
\nextgroupplot
\addplot graphics[xmin=-0.421,ymin=1.011,xmax=-3.191,ymax=2.990] {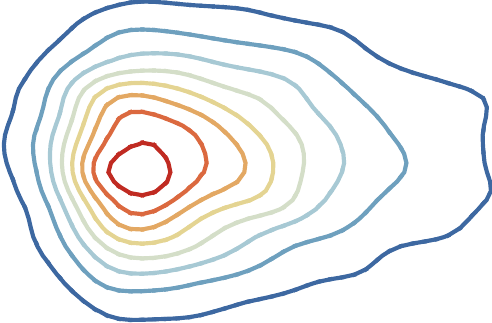};
\nextgroupplot
\addplot graphics[xmin=-2.078,ymin=1.456,xmax=-3.191,ymax=2.990] {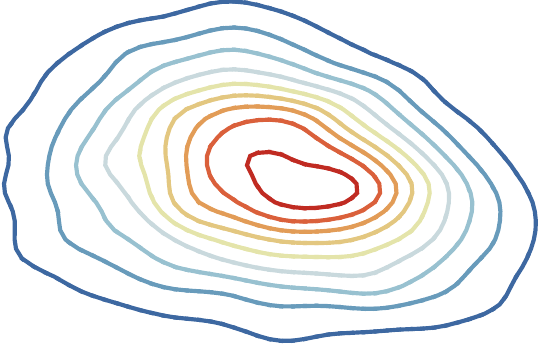};
\nextgroupplot
\addplot graphics[xmin=-2.830,ymin=2.885,xmax=-3.191,ymax=2.990] {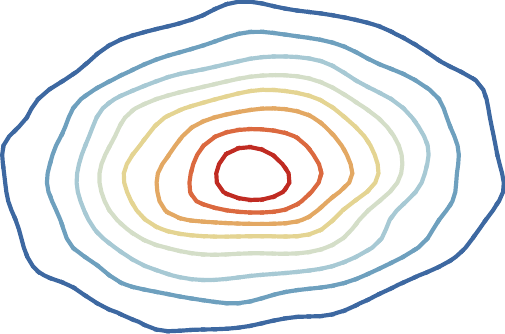};
\nextgroupplot[title={\huge$\theta_{10}$}]
\addplot graphics[xmin=-3.191,ymin=2.990,xmax=0.000,ymax=0.543] {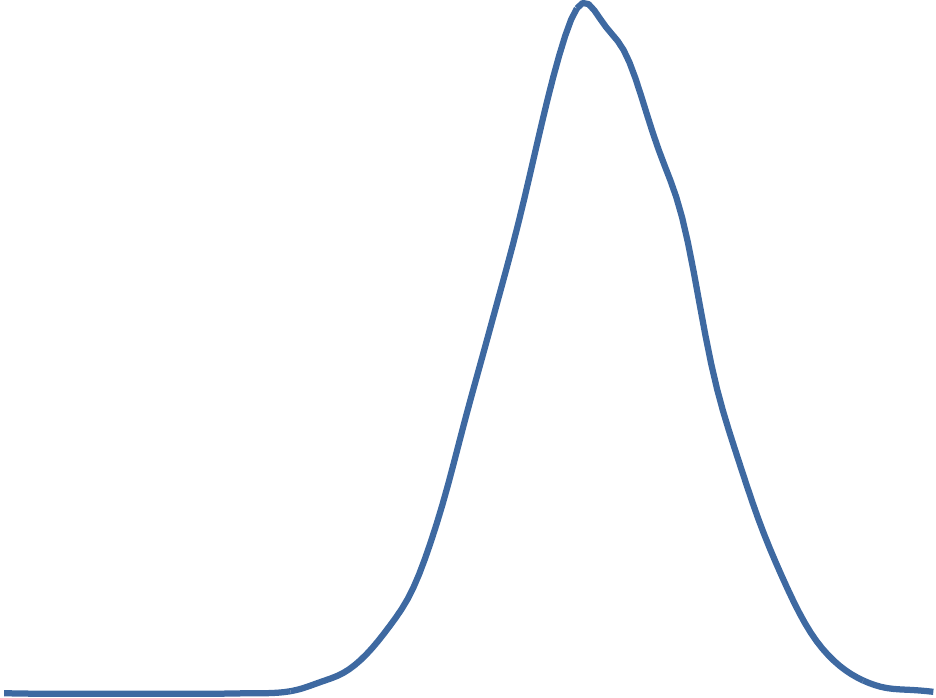};

\end{groupplot}
\end{tikzpicture}

    \caption{Posterior distribution (single and pairwise posterior marginals).}
    \label{fig:maplePost:kde}
    \end{subfigure}%
  }\cr
  \noalign{\hfill}
  \hbox{%
    \begin{subfigure}{.29\textwidth}
    \centering
    \hspace{-1cm}

\begin{tikzpicture}
\begin{axis}[width=\textwidth, height=\textwidth, xmin=-1, xmax=0.600000, ymin=-0.650000, ymax=2.5, xlabel={$\trv_1$}, ylabel={$\trv_7$}, axis on top=true]
\addplot graphics[xmin=-1,ymin=-0.65,xmax=0.600000,ymax=2.5] {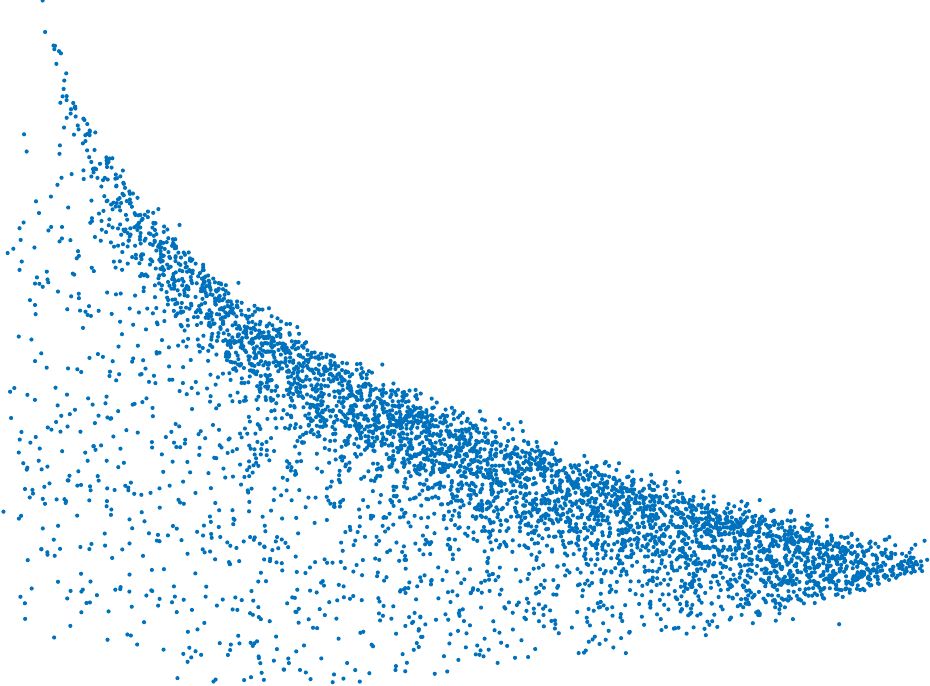};
\end{axis}
\end{tikzpicture}

    \vspace{-0.1cm}
    \caption{Scatter plot of $\theta_1$ and $\theta_7$.}
    \label{fig:maplePost:17}
    \end{subfigure}%
  }\vfill
  \hbox{%
    \begin{subfigure}{.29\textwidth}
    \centering
   \hspace{-1cm}

\begin{tikzpicture}
\begin{axis}[width=\textwidth, height=\textwidth, xmin=-1.00000, xmax=1.000000, ymin=-0.600000, ymax=0.400000, xlabel={$\trv_3$}, ylabel={$\trv_6$}, axis on top=true]
\addplot graphics[xmin=-1,ymin=-0.600000,xmax=1.000000,ymax=0.400000] {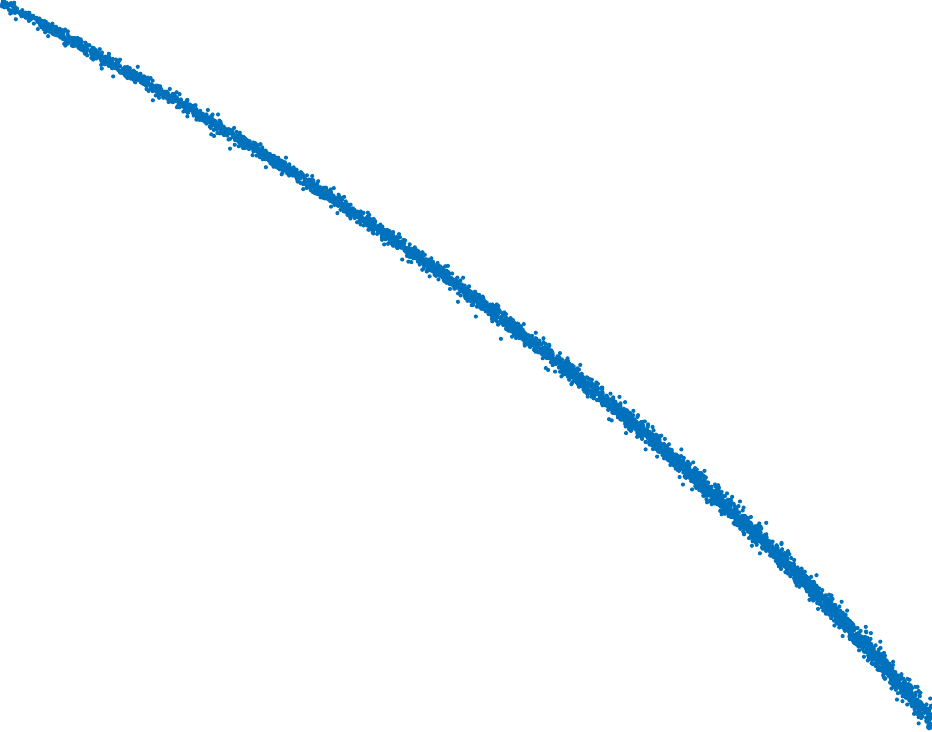};
\end{axis}
\end{tikzpicture}

    \vspace{-0.1cm}
    \caption{Scatter plot of $\theta_3$ and $\theta_6$.}
    \label{fig:maplePost:36}
    \end{subfigure}%
  }\cr
}

%
\caption{Posterior distribution of the maple sap exudation example. The kernel density estimates in Figure \ref{fig:maplePost:kde} misrepresent some sharp edges and narrow regions of the posterior, as illustrated in the scatter plots of Figures \ref{fig:maplePost:17} and \ref{fig:maplePost:36}.}
\label{fig:maplePost}
\end{figure}

Many of the two-dimensional marginal plots in Figure \ref{fig:maplePost} are close to Gaussian; however, the complicated relationships between $(\trv_1,\trv_7)$ and between $(\trv_3,\trv_6)$ yield a difficult posterior for MCMC methods.  The very tight and curved joint distribution shown in Figure \ref{fig:maplePost:36} is particularly challenging to capture and sample.  At the early stages of adaptation, both DRAM and the transport map proposals are nearly isotropic and require very small s.pdf to have a nonzero acceptance rate.  As the methods adapt, however, the proposals begin to capture the strong correlation between $\trv_3$ and $\trv_6$ and larger steps can be employed.  The nonlinear dependencies are much better captured by the transport map proposals, resulting in the order-of-magnitude performance gains shown in Table \ref{tab:maplePerf}.

\begin{table}
\caption[Maple MCMC performance comparison]{\label{tab:maplePerf} Performance of MCMC samplers on the maple parameter inference problem. Column headings are as described in Table~\ref{tab:bodPerf1}.}
\centering

\small


\fbox{\begin{tabular}{l|rr|rrr|D{.}{.}{4}D{.}{.}{2}}
Method & $\tau_{\max}$ & $\sigma_\tau$ & ESS & ESS/sec & ESS/eval &  \begin{minipage}{0.75cm}\centering Rel.\\[-0.15cm] ESS/sec\end{minipage} & \begin{minipage}{0.75cm}\centering Rel.\\[-0.15cm] ESS/eval\end{minipage} \\\hline
DRAM &  2571.4 & 1410.0 & 19 & 2.2e-06 & 5.6e-05 & 1.0 & 1.0\\\hline
TM+DRG &  1144.2 & 494.8 & 43 & 4.4e-06 & 1.2e-04 & 2.0 & 2.1\\
TM+DRL &  460.1 & 170.0 & 108 & 1.2e-05 & 3.3e-04  & 5.4 & 5.9\\
TM+RWM &  1129.7 & 775.9 & 44 & 8.0e-06 & 8.9e-04 & 3.7 & 15.8\\
\end{tabular}}
\end{table}

In contrast with the previous two examples, the TM+DRG method is not the top performers in this comparison.  The previous examples had simpler target distributions where the transport map could capture nearly all of the problem structure, allowing the independence proposal in TM+DRG to efficiently explore the parameter space.  The maple model's posterior, however, is much more challenging and cannot be entirely characterized with a cubic map; thus, the global proposals are less effective.  In this example, TM+DRL is the best-performing variant of the algorithm because it uses only local proposals and is not as sensitive to map deficiencies.


With challenging target distributions like this one, small initial proposal s.pdf are needed to begin sampling. However, small initial steps do not adequately explore the parameter space, yielding an inaccurate finite-sample approximation to the KL divergence in  \eqref{eq:regopt}. Without the regularization term in \eqref{eq:regopt}, one may then obtain transport maps that place too much probability mass on the relatively small region explored by the initial chain; in this sense, the proposal ``collapses.''  A sufficiently large regularization term prevents this collapse, but can also result in a slower adaptation process. We started adapting the map after $\num{5e3}$ steps of the chain and found that $k_R = \num{2e-5}$ was sufficiently large to prevent proposal collapse when the starting isotropic random-walk proposal was tuned to have a 1\% acceptance rate.  However, when the initial proposal was shrunk to obtain a 30\% acceptance rate, we needed a much larger value of $k_R \approx \num{1e-2}$.

\section{Conclusions}
\label{sec:conc}

%

We have introduced a new MCMC approach that uses transport maps to accelerate sampling from challenging target distributions. Our approach adaptively constructs nonlinear transport maps from MCMC samples, via the solution of a convex and separable optimization problem. From one perspective, the resulting maps transform the target to a reference distribution that is increasingly Gaussian and isotropic, and hence easier to sample. From a complementary perspective, the maps transform simple proposal mechanisms into non-Gaussian proposals on the target. Our maps are by construction invertible and continuously differentiable functions between the reference and target spaces, and hence they allow broad flexibility in choosing reference-space MCMC proposals. Yet building the maps themselves requires no derivative information from the target distribution. 

The efficiency of our approach is primarily a result of capturing nonlinear dependencies and non-Gaussian structure in the posterior and, when possible, exploiting this knowledge with global independence proposals (e.g., TM+DRG).  Of course, sequentially updating the transport map introduces an additional computational cost, which may become important in simple problems. As shown in the BOD example, however, our methods can be more efficient on strongly non-Gaussian problems, even when the target density is trivial to evaluate. On more complex posteriors, as in the ODE and DAE examples of Section~\ref{sec:perf:predprey} and \ref{sec:perf:maple}, the efficiency gains can be even more significant, both in terms of effective sample size per posterior evaluation and effective sample size per unit of wallclock time. It is also important to point out that our current implementation does not exploit the many levels of parallelism afforded by the map construction algorithm: solution of the optimization problem \eqref{eq:regopt} can made embarrassingly parallel over parameter dimensions, and additional parallelism can be introduced over samples.\footnote{Our implementation is freely available in MUQ \cite{MUQ}. This work used commit 7417f35 from MUQ's git repository.}


%

While the present work used polynomials to represent the transport map, this is not an essential aspect of the framework. In fact, the optimization problem for the map coefficients in \eqref{eq:regopt} will be unchanged for any map representation that is linear in the coefficients; we have experimented with other bases, e.g., radial basis functions, to good effect. 
%
Performance can be enhanced by an appropriate choice of basis, however, and future work will explore adaptive basis selection strategies. These will be particularly important for extending the transport map approach to higher-dimensional problems, where a more parsimonious choice of basis (versus the total-order bases used here) will be required. Recent results on the sparsity of triangular transports, and the decomposability of more general transports, may be useful in this regard \cite{Spantini2017}. Other methods for approximating the map, perhaps even nonparametric approaches, may also useful. We also note that the transport map defines a Riemannian metric on the parameter space, locally given by $(\dimap(\trv))^\top (\dimap(\trv))$. This suggests links between map-accelerated sampling and differential geometric MCMC methods, which we plan to explore.

\section*{Acknowledgments}
The authors would also like to thank F.\ Augustin, B.\ Calderhead, T.\ Cui, M.\ Girolami, T.\ Moselhy, A.\ Solonen, and A.\ Spantini for many helpful comments and suggestions.

\appendix
\section{ESS Calculation Details}\label{sec:esscalc}
Here we describe the calculation of the maximum integrated autocorrelation time $\tau_{\text{max}}$ used throughout our results.  Assume we are given $M$ independent MCMC chains on an $n$-dimensional parameter space.   Then, let $\tau_{i,j}$ be the integrated autocorrelation time of dimension $j$ on chain $i$.  This value is computed by applying the Fourier transform method from \cite{Wolff2004} to each dimension of each chain independently.  We then define $\tau_{\text{max}}$ as 
\begin{equation}
\tau_{\text{max}} = \max_{j\in\{1,\ldots,n\}} \left[ \underset{i\in\{1,\ldots, M\}}{\text{median}}\left( \tau_{i,j} \right) \right],
\end{equation}
where the median is taken over the chains and the maximum (worst case) is taken over dimensions.  

Effective sample size (ESS) is calculated similarly.  Let $\text{ESS}_{i,j} = \frac{K}{2\tau_{i,j}}$, where $K$ is the number of post-burn-in samples in each chain.  The reported ESS is then given by 
\begin{equation}
\text{ESS} = \min_{j\in\{1,\ldots, n\}} \left[ \underset{i\in\{1,\ldots, M\}}{\text{median}} \left( \frac{K}{2\tau_{i,j}}\right) \right].
\end{equation}
Note that while ESS uses only the samples produced after the burn-in period, normalized values of ESS reported in Section \ref{sec:perf} (e.g., ESS per function evaluation and ESS per second of wallclock time) use \emph{all} function evaluations or computational time in evaluating the denominator. Thus the cost of burn-in is reflected in these normalized performance metrics.
\section{Proof of ergodicity}\label{sec:detailedConvergence}
Section \ref{sec:convergence} of the paper provides an overview of the
convergence properties of our map-accelerated MCMC algorithm.  In this
appendix, we include some of the associated technical analysis.  In
particular, we provide detailed proofs of Lemma~\ref{lem:propBnds} and
Lemma~\ref{lem:dimAdapt}. The remaining results needed for Theorem \ref{thm:ergodic} are direct extensions of the proof of Lemma 6.1 in \cite{Atchade2006}.

\subsection{Bounded target proposal}
The goal of this section is to prove Lemma \ref{lem:propBnds} by finding two zero-mean Gaussian densities that bound the map-induced target-space proposal density $q_{\trv,\bar{\mapp}}$.  We assume throughout this appendix that the target density $\pi(\trv)$ is finite, continuous, and super-exponentially light. (See \eqref{eq:superLight} for the definition of super-exponentially light.)  We also assume that the reference proposal density $q_\rrv(\rrv^\prime | r)$ is a Gaussian random walk with a location-dependent bounded drift term $m(\rrv)$ and fixed covariance $\Sigma$.  Such a proposal takes the form
\begin{equation}
q_\rrv(\rrv^\prime | r) = N(\rrv+m(\rrv), \Sigma).
\end{equation}
Given this proposal density, we can follow \cite{Atchade2006} and show that there exist two zero-mean Gaussian densities $g_1$ and $g_2$, as well as two scalars $k_1$ and $k_2$, such that $0<k_1<k_2<\infty$ and 
\begin{equation}
k_1 g_1(\rrv^\prime - \rrv) \leq q_\rrv(\rrv^\prime | \rrv) \leq k_2 g_2(\rrv^\prime - \rrv)\label{eq:propBounds}.
\end{equation}
Now, we will use the bi-Lipschitz condition in \eqref{eq:dUB} to bound the target space proposal $q_{\trv,\bar{\mapp}}$ as required by Lemma \ref{lem:propBnds}.

\begin{proof}[Proof of Lemma \ref{lem:propBnds}]
The following s.pdf yield an upper bound:
\begin{eqnarray}
q_{\trv,\bar{\mapp}}(\trv^\prime \vert \trv ) & = & q_\rrv(\imap(\trv^\prime) \vert \imap(\trv)) |\det \dimap(\trv^\prime)|\label{eq:tqUB1}  \nonumber \\
& \leq & q_\rrv(\imap(\trv^\prime) \vert \imap(\trv))  \lambda_{\text{max}}^\pd\label{eq:tqUB2}  \nonumber \\
& \leq & k_2g_2(\imap(\trv^\prime) - \imap(\trv)) \lambda_{\text{max}}^\pd\label{eq:tqUB3}  \nonumber \\
& \leq & \left(k_2 \lambda_{\text{max}}^\pd\right)g_2\left( \lambda_{\text{min}}\left(\trv^\prime-\trv\right)\right)\label{eq:tqUB4} \nonumber \\
& = & k_{U}g_U\left(\trv^\prime-\trv\right)\label{eq:tqUB},
\end{eqnarray}
where $g_U$ is another zero-mean Gaussian.  Moving from the second line to the third line above is a consequence of \eqref{eq:dUB}. Moving from the third line to the fourth line uses the lower bound in \eqref{eq:dUB} and the fact that $g_2$ is a Gaussian with zero mean, which implies that $g_2(x_1)>g_2(x_2)$ when $\|x_1\|<\|x_2\|$.  Notice that $k_U$ does not depend on the particular coefficients of the map $\imap$; it only depends on the Lipschitz constant in \eqref{eq:dUB}. A similar process can be used to obtain the following lower bound:
\begin{eqnarray}
q_{\trv,\bar{\mapp}}(\trv^\prime \vert \trv ) & = & q_\rrv(\imap(\trv^\prime) \vert \imap(\trv)) |\det\dimap(\trv^\prime)|\nonumber \\
& \geq & q_\rrv(\imap(\trv^\prime) \vert \imap(\trv))  \lambda_{\text{min}}^\pd \nonumber \\
& \geq & k_1g_1(\imap(\trv^\prime) - \imap(\trv)) \lambda_{\text{min}}^\pd \nonumber \\
& \geq & \left(k_1 \lambda_{\text{min}}^\pd \right)g_1\left( \lambda_{\text{max}}\left(\trv^\prime-\trv\right)\right)\nonumber \\
& = & k_{L}g_L\left(\trv^\prime-\trv\right)\label{eq:tqLB}.
\end{eqnarray}
Lemma \ref{lem:propBnds} follows directly from (\ref{eq:tqUB}) and (\ref{eq:tqLB}).
\end{proof}

\subsection{SSAGE}

With (\ref{eq:tqUB}) and (\ref{eq:tqLB}) in hand, the proof of Lemma 6.1 in  \cite{Atchade2006} yields Lemma \ref{lem:minor}: the minorization component of the SSAGE condition.  Thus, to show SSAGE, we only need to establish Lemma \ref{lem:drift}.  Our proof of Lemma \ref{lem:drift} is built on the intermediate Lemmas \ref{lem:part1} and \ref{lem:part2} provided below and on the proof of Lemma 6.2 in \cite{Atchade2006}.

For the arguments below, we will use the Metropolis-Hastings transition kernel given by
 \begin{equation*}
P_{\bar{\mapp}}(x,dy) = \alpha_{\bar{\mapp}}(x,y)q_{\trv,{\bar{\mapp}}}(y\vert x)dy + r_{\bar{\mapp}}(x)\delta_x(dy),
\end{equation*}
where
 \begin{equation*}
r_{\bar{\mapp}}(x) = 1-\int \alpha_{\bar{\mapp}}(x,y)q_{\trv,{\bar{\mapp}}}(y\vert x)dy,
\end{equation*}
and $\alpha$ is the Metropolis-Hastings acceptance probability given by
 \begin{equation*}
\alpha_{\bar{\mapp}}(x,y) = \min\left\{1, \frac{\td(y) q_{\trv,{\bar{\mapp}}}(x \vert  y)}{\td(x) q_{\trv,{\bar{\mapp}}}(y \vert x)}
\right\} .
\end{equation*}
We will also use the set of guaranteed acceptance, given by
\begin{equation*}
A_{\bar{\mapp}}(x) = \left\{ y\in \real^\pd : \pi(y)q_{\theta,{\bar{\mapp}}}(x\vert y)\geq \pi(x)q_{\theta,{\bar{\mapp}}}(y\vert x)\right\},
\end{equation*}
and the set of possible rejection, simply defined as the complement of the set above:
\begin{equation*}
R_{\bar{\mapp}}(x) = A_{\bar{\mapp}}(x)^C .
\end{equation*}

 \begin{lemma}\label{lem:part1}
Let $V(x) = c_V\td^{-\alpha}(x)$ be a drift function defined by some $\alpha\in (0,1)$.  The constant $c_V= \sup_{x} \pi^\alpha(x)$ is chosen so that $\inf_x V(x) = 1$.  Then the following holds:
\begin{equation}
\limsup_{\|x\|\rightarrow \infty} \, \sup_{\bar{\mapp}} \, \frac{\int_{\real^\pd} V(y) P_{\bar{\mapp}}(x,dy)}{V(x)} < \limsup_{\|x\|\rightarrow \infty} \, \sup_{\bar{\mapp}} \,  \int_{R_{\bar{\mapp}}(x)}q_{\trv,{\bar{\mapp}}}(y |x) dy \label{eq:part1Goal}.
\end{equation}
 \end{lemma}
\begin{proof}
First, we decompose the left-hand side of (\ref{eq:part1Goal}) into
\begin{eqnarray}
\frac{\int_{\real^\pd} V(y) P_{\bar{\mapp}}(x,dy)}{V(x)} &=& \int_{A_{\bar{\mapp}}(x)} \frac{\td^{-\alpha}(y)}{\td^{-\alpha}(x)} q_{\trv,{\bar{\mapp}}}(y |x) dy + \int_{R_{\bar{\mapp}}(x)}  \frac{\td^{-\alpha}(y)}{\td^{-\alpha}(x)} \frac{\td(y) q_{\trv,{\bar{\mapp}}}(x | y)}{\td(x) q_{\trv,{\bar{\mapp}}}(y |x)}q_{\trv,{\bar{\mapp}}}(y |x)dy \nonumber\\
&&+ \int_{R_{\bar{\mapp}}(x)} \left(1-\frac{\td(y) q_{\trv,{\bar{\mapp}}}(x | y)}{\td(x) q_{\trv,{\bar{\mapp}}}(y |x)} \right) q_{\trv,{\bar{\mapp}}}(y |x) dy \label{eq:part1Portions}.
\end{eqnarray}
Following the proof of Lemma 6.2 in \cite{Atchade2006}, we can show that the first two integrals in (\ref{eq:part1Portions}) go to zero as $\|x\|\rightarrow \infty$.  With that, we have
\begin{eqnarray}
\limsup_{\|x\|\rightarrow \infty} \, \sup_{\bar{\mapp}} \, \frac{\int_{\real^\pd} V(y) P_{\bar{\mapp}}(x,dy)}{V(x)}  &=& \limsup_{\|x\|\rightarrow \infty} \, \sup_{\bar{\mapp}} \, \int_{R_{\bar{\mapp}}(x)} \left(1-\frac{\td(y) q_{\trv,{\bar{\mapp}}}(x | y)}{\td(x) q_{\trv,{\bar{\mapp}}}(y |x)} \right) q_{\trv,{\bar{\mapp}}}(y |x) dy \nonumber \\
&<& \limsup_{\|x\|\rightarrow \infty} \, \sup_{\bar{\mapp}} \, \int_{R_{\bar{\mapp}}(x)} q_{\trv,{\bar{\mapp}}}(y |x) dy.
\end{eqnarray}
The inequality results from the fact that  $[\td(y) q_{\trv,{\bar{\mapp}}}(x | y)]/[\td(x) q_{\trv,{\bar{\mapp}}}(y |x)]<1$ when $y \in R_{\bar{\mapp}}(x)$.
\end{proof}
\smallskip

\begin{lemma}\label{lem:part2}
The proposal has a nonzero probability of acceptance, i.e.,
 \begin{equation}
\int_{R_{\bar{\mapp}}(x)}q_{\trv,{\bar{\mapp}}}(y |x) dy < 1. \label{eq:part2Goal}
\end{equation}
 \end{lemma}
\begin{proof}
A nonzero probability of acceptance occurs if and only if there is a measurable set $W(x)\subset A_{\bar{\mapp}}(x)$.  To show that $W(x)$ exists, consider a small ball of radius $R$ around $x$.  Since $g_L$ and $g_U$ are zero mean and have positive variance, this implies
\begin{equation}
\inf_{y\in B(x,R)} \inf_{\bar{\mapp}} \frac{q_{\trv,{\bar{\mapp}}}(x|y)}{q_{\trv,{\bar{\mapp}}}(y|x)}  \geq  \inf_{y\in B(x,R)} \frac{k_Lg_L(x-y)}{k_Ug_U(y-x)} \geq  c_0 \, ,
\end{equation}
for some $c_0>0$.  Because $\td(x)$ is super-exponentially light, for any $u\in(0,R)$, there exists a radius $r_4$ such that $\| x \| > r_4$ implies
\[
\td\left(x-u\frac{x}{\|x\|}\right) \geq \frac{\td(x)}{c_0}.
\]
Subsequently, for any map coefficients, the acceptance probability for $x_1 = x-u\frac{x}{\|x\|}$ is one, which implies that $x_1\in A_{\bar{\mapp}}(x)$.  Now, define $W(x)$ as
\[
W(x) = \left\{ x_1-a\zeta, \, 0<a<R-u, \, \zeta \in S^{\pd-1},\, \left\|\zeta - \frac{x_1}{\|x_1\|}\right\|<\frac{.pdfilon}{2}\right\},
\]
where $.pdfilon$ is an arbitrarily small scalar and $S^{\pd-1}$ is the unit sphere in $\real^\pd$ dimensions.  
Note that $\left\|\zeta - x_1 / \|x_1\| \right\|<\frac{.pdfilon}{2}$ ensures that $W(x)$ is a cone of points closer to the origin than $x_1$.
Now, using the final paragraph of the proof of Lemma 6.2 in \cite{Atchade2006}, the curvature condition from \eqref{eq:curveCond} ensures that the target density is larger in $W(x)$ than at $x_1$.  Since $x_1$ was accepted, this means that everything in $W(x)$ will also be accepted and that $W(x)\subseteq A_{\bar{\mapp}}(x)$.  Subsequently, we obtain
\begin{eqnarray}
\lim_{\|x\|\rightarrow \infty}\int_{R_{\bar{\mapp}}(x)} q_{\trv,{\bar{\mapp}}}(y |x) dy  &=& \lim_{\|x\|\rightarrow \infty} \left(1 - \int_{A_{\bar{\mapp}}(x)} q_{\trv,{\bar{\mapp}}}(y |x) dy\right)\nonumber\\
& \leq &  \lim_{\|x\|\rightarrow \infty}\left(1 -  \int_{W(x)} q_{\trv,{\bar{\mapp}}}(y |x) dy\right) \nonumber\\
& < &  1 \label{eq:finalOneDrift2},
\end{eqnarray}
where we have used the fact that $W(x)$ is a measurable subset of $A_{\bar{\mapp}}(x)$ for large $x$.
\end{proof}
\smallskip

With Lemmas \ref{lem:part1} and \ref{lem:part2} in hand, we can now proceed to the proof of Lemma \ref{lem:drift} (the drift condition) from the main text.

\begin{proof}[Proof of Lemma \ref{lem:drift}]
%
Recall our choice of drift function: $V(x) = c_V\td^{-\alpha}(x)$ for $\alpha\in (0,1)$.  Using this function and the definitions of $ P_{\bar{\mapp}}$, $R_{\bar{\mapp}}$, and $A_{\bar{\mapp}}$ we can show that
\begin{eqnarray}
\frac{\int_{\real^\pd} V(y) P_{\bar{\mapp}}(x,dy)}{V(x)} &=& \frac{\int_{\real^\pd} \td^{-\alpha}(y) P_{\bar{\mapp}}(x,dy)}{\td^{-\alpha}(x)} \label{eq:simple1} \nonumber \\
&& + \int_{R_{\bar{\mapp}}(x)} \left(1-\frac{\td(y) q_{\trv,{\bar{\mapp}}}(x | y)}{\td(x) q_{\trv,{\bar{\mapp}}}(y |x)} \right) q_{\trv,{\bar{\mapp}}}(y |x) dy \nonumber\\
& = & \int_{R_{\bar{\mapp}}(x)}q_{\trv,{\bar{\mapp}}}(y |x) dy  + \int_{A_{\bar{\mapp}}(x)} \frac{\td^{-\alpha}(y)}{\td^{-\alpha}(x)} q_{\trv,{\bar{\mapp}}}(y |x) dy \nonumber\\
&&+ \int_{R_{\bar{\mapp}}(x)}  \left(\frac{\td^{-\alpha}(y)}{\td^{-\alpha}(x)} -1\right)\frac{\td(y) q_{\trv,{\bar{\mapp}}}(x | y)}{\td(x) q_{\trv,{\bar{\mapp}}}(y |x)}q_{\trv,{\bar{\mapp}}}(y |x)dy\nonumber\\
&&+ \int_{R_{\bar{\mapp}}(x)}  \frac{\td^{-\alpha}(y)}{\td^{-\alpha}(x)} \frac{\td(y) q_{\trv,{\bar{\mapp}}}(x | y)}{\td(x) q_{\trv,{\bar{\mapp}}}(y |x)}q_{\trv,{\bar{\mapp}}}(y |x)dy\nonumber\\
&\leq & 1  + \int_{A_{\bar{\mapp}}(x)} \frac{\td^{-\alpha}(y)}{\td^{-\alpha}(x)} q_{\trv,{\bar{\mapp}}}(y |x) dy \\
&& + \int_{R_{\bar{\mapp}}(x)}  \frac{\td^{-\alpha}(y)}{\td^{-\alpha}(x)} \frac{\td(y) q_{\trv,{\bar{\mapp}}}(x | y)}{\td(x) q_{\trv,{\bar{\mapp}}}(y |x)}q_{\trv,{\bar{\mapp}}}(y |x)dy . \nonumber 
\end{eqnarray}
Within the region of possible rejection $R_{\bar{\mapp}}(x)$, the acceptance rates are all in $[0,1)$, which allows us to further simplify \eqref{eq:simple1} to obtain
\begin{eqnarray}
\frac{\int_{\real^\pd} V(y) P_{\bar{\mapp}}(x,dy)}{V(x)} &\leq& 1 +\int_{A_{\bar{\mapp}}(x)} \frac{\td^{-\alpha}(y)}{\td^{-\alpha}(x)} q_{\trv,{\bar{\mapp}}}(y |x) dy +  \int_{R_{\bar{\mapp}}(x)}  \frac{\td^{-\alpha}(y)}{\td^{-\alpha}(x)}q_{\trv,{\bar{\mapp}}}(y |x)dy\nonumber\\
  &<& 1 +\int_{A_{\bar{\mapp}}(x)} \frac{\td^{-\alpha}(y)}{\td^{-\alpha}(x)} q_{\trv,{\bar{\mapp}}}(y |x) dy +  \int_{R_{\bar{\mapp}}(x)}  \frac{q_{\trv,{\bar{\mapp}}}^{-\alpha}(y |x)}{q_{\trv,{\bar{\mapp}}}^{-\alpha}(x|y)}q_{\trv,{\bar{\mapp}}}(y |x)dy\nonumber\\
  &=& 1 +\int_{A_{\bar{\mapp}}(x)} \frac{\td^{-\alpha}(y)}{\td^{-\alpha}(x)} q_{\trv,{\bar{\mapp}}}(y |x) dy +  \int_{R_{\bar{\mapp}}(x)} q_{\trv,{\bar{\mapp}}}^{1-\alpha}(y |x)q_{\trv,{\bar{\mapp}}}^{\alpha}(x|y)dy\nonumber\\
  &\leq& 1 +\int_{A_{\bar{\mapp}}(x)} \frac{\td^{-\alpha}(y)}{\td^{-\alpha}(x)} q_{\trv,{\bar{\mapp}}}(y |x) dy\nonumber\\
  &\leq& 1 +\int_{A_{\bar{\mapp}}(x)} \frac{\td^{-\alpha}(y)}{\td^{-\alpha}(x)} q_{\trv,{\bar{\mapp}}}(y |x) dy +  k_U^2 \int_{R_{\bar{\mapp}}(x)} g_U(y-x)dy \nonumber\\
  &=& 1+C_R +\int_{A_{\bar{\mapp}}(x)} \frac{\td^{-\alpha}(y)}{\td^{-\alpha}(x)} q_{\trv,{\bar{\mapp}}}(y |x) dy,
\end{eqnarray}
where we have used the density upper bound in (\ref{eq:tqUB}) and $C_R$ is a finite constant.  A similar application of (\ref{eq:tqUB}) over $A_{\bar{\mapp}}(x)$ yields
\begin{eqnarray}
\frac{\int_{\real^\pd} V(y) P_{\bar{\mapp}}(x,dy)}{V(x)} &\leq&1+C_R +\int_{A_{\bar{\mapp}}(x)} \frac{\td^{\alpha}(x)}{\td^{\alpha}(y)} q_{\trv,{\bar{\mapp}}}(y |x) dy\nonumber\\
&\leq& 1+C_R +\int_{A_{\bar{\mapp}}(x)} q_{\trv,{\bar{\mapp}}}^\alpha(x |y) q_{\trv,{\bar{\mapp}}}^{1-\alpha}(y |x)dy\nonumber\\
&\leq& 1+C_R +k_U^2\int_{A_{\bar{\mapp}}(x)} g_U(x-y)dy\nonumber\\
&<& \infty \label{eq:driftFiniteVerify}.
\end{eqnarray}
Using Lemma \ref{lem:part1} and Lemma \ref{lem:part2}, we also have that 
\begin{equation}
\limsup_{\|x\|\rightarrow \infty} \, \sup_{\bar{\mapp}} \, \frac{\int_{\real^\pd} V(y) P_{\bar{\mapp}}(x,dy)}{V(x)} < \limsup_{\|x\|\rightarrow \infty} \, \sup_{\bar{\mapp}} \,  \int_{R_{\bar{\mapp}}(x)}q_{\trv,{\bar{\mapp}}}(y |x) dy < 1.\label{eq:oneVerify}
\end{equation}
From the proof of Lemma 6.2 in \cite{Atchade2006}, which resembles the proofs in \cite{Jarner2000}, Lemma \ref{lem:drift} follows from simultaneously satisfying the bounds \eqref{eq:driftFiniteVerify} and \eqref{eq:oneVerify}.  
 \end{proof}



\subsection{Diminishing adaptation}
In addition to SSAGE and containment, Theorem \ref{thm:ergodic} requires diminishing adaptation (Definition \ref{def:dim}).  The following proof establishes the diminishing adaptation proposed in Lemma \ref{lem:dimAdapt}.
\begin{proof}[Proof of Lemma \ref{lem:dimAdapt}]
The proof of this lemma relies on continuity of the map with respect
to $\bar{\mapp}$ and the convergence of \eqref{eq:regopt} as the number
of samples $K\rightarrow \infty$.  Note that we do not require
\eqref{eq:regopt} (or \eqref{eq:mckl})) to converge to the minimizer of the true KL divergence.

When the MCMC chain is not at an adaptation step, $\bar{\mapp}^{(k+1)}=\bar{\mapp}^{(k)}$.  Thus, to show diminishing adaptation, we need to show that the difference between transition kernels at step $K$ and $K+K_U$ decreases as $K\rightarrow \infty$.  Mathematically, we require
\begin{equation}
\lim_{K\rightarrow \infty} \mathbb{P}\left(\sup_{x \in \real^\pd}\left\| \tk_{\bar{\mapp}^{(K)}}(x,\cdot) - \tk_{\bar{\mapp}^{(K+K_U)}}(x,\cdot)\right\|_{TV} \geq \delta_1\right) = 0
\label{eq:tvinprob} 
\end{equation}
for any $\delta_1>0$.  Because the maps are linear in $\bar{\mapp}$, for a fixed $x$, the mapping from $\bar{\mapp}$ to $\tk_{\bar{\mapp}}(x,A)$ is continuous for any $A$.  Combined with the fact that $q_{\trv,{\bar{\mapp}}}$ is bounded, we have that \eqref{eq:tvinprob} will be satisfied when  
 \begin{equation}
\lim_{K\rightarrow \infty}\mathbb{P}\left( \left\| \mapp_i^{(K+K_U)}  -\mapp_i^{(K)} \right\| \geq \delta\right) = 0, \label{eq:da_goal1}
\end{equation}
for any $\delta>0$ and all $i\in\{1,2,\ldots,\pd\}$ .  We now turn to proving \eqref{eq:da_goal1}.

Recall that $\bar{\mapp}^{(K)}$ is the minimizer of 
\eqref{eq:regopt}, which is based on a $K$-sample Monte Carlo approximation of the KL 
divergence.   To notationally simplify \eqref{eq:regopt}, we will now use the convention that $\log(0)=-\infty$ and define the objective functions $f_i^{(K)}(\mapp_i)$ and $f_i^{(K+K_U)}(\mapp_i)$
as
\begin{eqnarray}
f_i^{(K)}(\mapp_i) & = & \frac{1}{K}g(\mapp_i) +\frac{1}{K}\sum_{k=1}^{K}\left[  \frac12\imap^2_i(\trv^{(k)}; \mapp_i) - \log\left.\frac{\partial \imap_i(\trv;\mapp_i)}{\partial \trv_i}\right|_{\trv^{(k)}}\right] \label{eq:fkk}\\
f_i^{(K+K_U)}(\mapp_i) & = & \frac{1}{K}g(\mapp_i) + \frac{1}{K}\sum_{k=1}^{K+K_U}\left[  \frac12\imap^2_i(\trv^{(k)}; \mapp_i) - \log\left.\frac{\partial \imap_i(\trv;\mapp_i)}{\partial \trv_i}\right|_{\trv^{(k)}}\right].\label{eq:fkku}
\end{eqnarray}
From \eqref{eq:regopt}, it should be clear that 
\begin{eqnarray}
\mapp_i^{(K)} &=& \argmin f_i^{(K)}(\mapp_i)\\
\mapp_i^{(K+K_U)} &=& \argmin f_i^{(K+K_U)}(\mapp_i),
\end{eqnarray}
for all $i=\{1,2,...,\pd\}$.\footnote{Using the factor $\frac{1}{K}$
  in both \eqref{eq:fkk} and \eqref{eq:fkku} is intentional.
  Multiplying the objective in \eqref{eq:regopt} by any positive
  scalar will not affect the solution and the common value of $\frac{1}{K}$ used here simplifies the results later on.}
Combining these expressions, we have
\begin{equation}
f_i^{(K+K_U)}(\mapp_i) =  f_i^{(K)}(\mapp_i) + \frac{1}{K}\sum_{k=K+1}^{K+K_U}\left[  \frac12\imap^2_i(\trv^{(k)}; \mapp_i) - \log\left.\frac{\partial \imap_i(\trv;\mapp_i)}{\partial \trv_i}\right|_{\trv^{(k)}}\right].\label{eq:da_parts1}
\end{equation}
From Markov's inequality, we then have
\begin{equation}
\mathbb{P}\left[ \left| f_i^{(K+K_U)}(\mapp_i) -  f_i^{(K)}(\mapp_i) \right| \geq \delta_2 \right] \leq \frac{1}{K \delta_2}\mathbb{E}\left[\left| \sum_{k=K+1}^{K+K_U}\left(  \frac12\imap^2_i(\trv^{(k)}; \mapp_i) - \log\left.\frac{\partial \imap_i(\trv;\mapp_i)}{\partial \trv_i}\right|_{\trv^{(k)}}\right)\right| \right],\label{eq:func_markov}
\end{equation}
for any $\delta_2>0$ and all $\gamma_i$.  Notice that the expectation on the right hand side of this expression is finite because the map is bi-Lipschitz \eqref{eq:dUB}, the proposal density is bounded by Gaussian densities (see Lemma \ref{lem:propBnds}), and the map is linear for large $\|\theta\|$ (see \eqref{eq:ubMap}).  Thus, 
\begin{equation}
\lim_{K\rightarrow \infty} \mathbb{P}\left[ \left| f_i^{(K+K_U)}(\mapp_i) -  f_i^{(K)}(\mapp_i) \right| \geq \delta_2 \right]  = 0 \quad \forall \gamma_i. \label{eq:funcConv}
\end{equation}

We now show that this implies the convergence of $\| \mapp_i^{(K+K_U)}-\mapp_i^{(K)} \|$.  First, consider a set $\mathcal{C}^{(K)}$ that depends on $\delta_2$ and takes the form
\begin{equation}
\mathcal{C}^{(K)} = \left\{ \mapp_i \, : \, f_i^{(K)}(\mapp_i)-\delta_2 \leq f_i^{(K)}(\mapp_i^{(K)})+\delta_2 \right\} \label{eq:ck_conv}.
\end{equation}
By definition, $\mathcal{C}^{(K)}$ will always contain $\mapp_i^{(K)}$.    Recall that $f_i^{(K)}$ is convex and admits a unique global minimizer.  Thus, as $\delta_2\rightarrow 0$, the set $\mathcal{C}^{(K)}$ will collapse on $\mapp_i^{(K)}$ and the maximum distance between any two points in $\mathcal{C}^{(K)}$ will go to zero.  This implies that for any $\delta>0$, there exists a $\delta_2$ such that 
\begin{equation}
 \sup_{\mapp_i, \mapp_i^\prime \in \mathcal{C}^{(K)}} \|\mapp_i - \mapp_i^\prime\|  < \delta. \label{eq:ck_bounds}
\end{equation}
We will now combine this expression with \eqref{eq:funcConv}.  Notice that for any $\delta_2>0$, \eqref{eq:funcConv} implies that
\begin{equation}
\lim_{K\rightarrow \infty} \mathbb{P}\left( \mapp_i^{(K+K_U)} \in \mathcal{C}^{(K)} \right) = 1.
\end{equation}
Combining this result with \eqref{eq:ck_bounds} yields 
\begin{equation}
\lim_{K\rightarrow \infty} \mathbb{P}\left( \left\| \mapp_i^{(K+K_U)}
    -  \mapp_i^{(K)} \right\| \geq \delta \right) = 0 \, ,
\end{equation}
which is the desired condition in \eqref{eq:da_goal1}.
\end{proof}

\section{Maple exudation model details}\label{sec:maplemodel}
The forward model in Section \ref{sec:perf:maple} is a complicated system of differential-algebraic equations describing maple sap dynamics.  Here we give a minimal description of the model. Interested readers should consult the original derivation in \cite{Ceseri2013}. 

In addition to the differential algebraic system defined by \eqref{eq:MapleDae1}--\eqref{eq:MapleAlg5}, the volumes $V^v$ and $V_g^v(t)$ are given by
\begin{eqnarray*}
V^v &=& \pi (R^v)^2L^v\\
V_g^v(t) &=& \pi r(t)^2L^v\\
N&=& \frac{2\pi(R^f+R^v+W)}{2R^f+W}.
\end{eqnarray*}
The system is solved using MUQ~\cite{MUQ}, which in turn links to SUNDIALS~\cite{Sundials}.  The initial conditions for the state variables $s_{gi}$, $s_{iw}$, and $r(t)$ are derived from a steady state solution. We put $U(0)=0$.

The temperature field is assumed to be quasi-steady and is defined by the heat equation
\begin{align}
\partial_{x}\left(\kappa(x) \partial_x T(x,t)\right) & =  0 & \text{for } & x\in(\,s_{iw}(t),\, R^f+2R^v\,) \label{eq:Temperature} \\
T(x,t) & = 0 & \text{at } & x=s_{iw}(t) \nonumber \\ 
\kappa_w \partial_x  T(x,t)&= h(T_a(t)-T(x,t)) & \text{at } &
                                                              x=R^f+2R^v, \nonumber
\end{align}
where $T_a(t)$ is a transient temperature forcing at the edge of the computational domain ($x=R^f+2R^v$), $h=10$ is a heat transfer coefficient, and the thermal conductivity is defined piecewise as
\begin{equation*}
\kappa(x) = \left\{\begin{array}{cl}
 \kappa_w & x\in[\, s_{iw}(t),\, R_f+R_v-r(t)\,) \\
 \kappa_g & x\in[\, R_f+R_v-r(t),\, R_f+R_v+r(t)\,) \\
 \kappa_w & x\in[\, R_f+R_v+r(t),\, R_f+2R_v\,]\end{array}\right.  ,
\end{equation*}
where $\kappa_w$ is the thermal conductivity of water and  $\kappa_g$ is the thermal conductivity of air.  At any particular time, it is straightforward to solve \eqref{eq:Temperature} analytically, yielding a piecewise linear temperature field.

\begin{table}
\centering
\caption{Relationship between inference targets $\theta$ and model parameters for the maple problem.}
\begin{tabular}{|c|c|}\hline
Model variable & Transformation from $\theta$ \\\hline
$s_{gi}(0)$ & $(0.5\theta_2+ 0.5)\exp\left[0.2\log(\bar{R}^f)\theta_7 + \log(\bar{R}^f)\right]$\\
$s_{iw}(0)$ & $\exp\left[0.2\log(\bar{R}^f)\theta_7 + \log(\bar{R}^f)\right]$\\ 
$r(0)$ & $(0.5\theta_2+ 0.5)\exp\left[0.2\log(\bar{R}^v)\theta_8 + \log(\bar{R}^v)\right]$\\
$p_g^f(0)$ & $\num{50e3}\theta_3+\num{150e3}$\\
$K$ & $0.2\log(\bar{K})\theta_4 + \log(\bar{K})$\\
$W$ & $0.2\log(\bar{W})\theta_5 + \log(\bar{W})$ \\
$c_s^v$ & $0.2\log(\bar{c}_s^v)\theta_6 + \log(\bar{c}_s^v)$\\
$R^f$ & $0.2\log(\bar{R}^f)\theta_7 + \log(\bar{R}^f)$\\
$R^v$ & $0.2\log(\bar{R}^v)\theta_8 + \log(\bar{R}^v)$\\
$L^f$ & $0.2\log(\bar{L}^f)\theta_9 + \log(\bar{L}^f) $\\
$L^v$ & $0.2\log(\bar{L}^v)\theta_{10} + \log(\bar{L}^v)$\\\hline
\end{tabular}
\label{tab:MapleTransforms}
\end{table}

The inference parameters $\trv$ are related to the model parameters in \eqref{eq:MapleDae1}--\eqref{eq:MapleAlg5} using the transformations in Table \ref{tab:MapleTransforms}; variables with an overbar are default parameters taken from \cite{Ceseri2013} and shown in Table \ref{tab:MapleDefaults}.  Values for the remaining physical constants are listed in Table \ref{tab:MapleConstants}.

\begin{table}
\centering
\caption{Default values used to generate synthetic data and to scale the inference parameters  in the maple problem.}
\begin{tabular}{|c|c|c|p{7cm}|}\hline
Symbol & Value & Units & Description\\\hline
$\bar{s}_{gi}(0)$  & $0.7\bar{R}^f$ & $\si{\meter}$ & Initial location of gas-ice interface in fiber.\\
$\bar{s}_{iw}(0)$  &  $\bar{R}^f$ & $\si{\meter}$ & Initial location of ice-water interface in fiber.\\
$\bar{r}(0)$  &  $0.3\bar{R}^v$ & $\si{\meter}$ & Initial radius of vessel gass bubble.\\
$\bar{p}_g^f(0)$  &  $\num{200e3}$ & $\si{\pascal}$ & Initial gas pressure in fiber. \\
$\bar{K}$  & $\num{1.98e-14}$ & $\si{\meter\per\second}$ & Hydraulic conductivity of fiber-vessel wall.\\
$\bar{W} $ & $\num{3.64e-6}$ & $\si{\meter}$ & Thickness of fiber-vessel wall. \\ 
$\bar{c}_s^v$  &  $\num{58.4}$ & $\si{\mole\per\meter\cubed}$ & Sucrose concentration in vessel sap. \\
$\bar{R}^f$  & $\num{3.5e-6}$ & $\si{\meter}$ & Fiber radius. \\
$\bar{R}^v$ & $\num{2e-5}$ & $\si{\meter}$ & Vessel radius. \\
$\bar{L}^f$ & $\num{1e-3}$ & $\si{\meter}$ & Fiber length. \\
$\bar{L}^v$ &  $\num{5e-4}$ & $\si{\meter}$ & Vessel length.\\
\hline
\end{tabular}
\label{tab:MapleDefaults}
\end{table}

\begin{table}
\centering
\caption{Physical constants used in maple exudation model.}
\begin{tabular}{|c|c|c|p{5.5cm}|}\hline
Symbol & Value & Units & Description\\\hline
$\rho_i$ & $917$ & $\si{\kilogram\per\meter\cubed}$ & Density of water ice.\\
$\rho_w$ & $1000$ & $\si{\kilogram\per\meter\cubed}$ & Density of water. \\
$\lambda$ &   $\num{3.34e5}$ & $\si{\joule\per\kilogram}$ & Latent heat of fusion for water. \\
$R$ & $8.314$ & $\si{\joule\per\mole\per\kelvin}$ & Universal gas constant.\\
$g$ & $9.81$ & $\si{\meter\per\second\squared}$ & Acceleration due to gravity. \\
$\sigma$ & $0.0756$ & $\si{\newton\per\meter}$ & Surface tension of water. \\
$H$ & $0.0274$ & -- & Henry's constant for air and water. \\
$M_g$ & $0.0290$ & $\si{\kilogram\per\mole}$ & Molar mass of air.\\
$\kappa_w$ & $0.580$ & $\si{\watt\per\meter\per\kelvin}$ & Thermal conductivity of water.\\
$\kappa_g$ & $0.0243$ & $\si{\watt\per\meter\per\kelvin}$ & Thermal conductivity of air.\\
\hline
\end{tabular}
\label{tab:MapleConstants}
\end{table}

\bibliographystyle{siamplain} 
\bibliography{arxiv_main}

\end{document}